\documentclass[final]{elsarticle}

\usepackage{amsthm}
\usepackage{amssymb}
\usepackage[dvipsnames]{xcolor}
\usepackage{tikz}
\usetikzlibrary{matrix}
\usetikzlibrary{arrows}
\usetikzlibrary{positioning}
\usepackage{mathrsfs}

\usepackage[shortlabels]{enumitem}
\setlist{nosep,noitemsep}

\usepackage{nicematrix}
\usepackage{bm}
\usepackage{xparse}
\usepackage[percent]{overpic}

\newcommand{\formulaNb}[2]{\widetilde{N}^{(#1)}}
\newcommand{\formulaN}[2]{N^{(#1)}}
\newcommand{\Dmac}{\mathbb{D}}

\newcommand{\mac}[6]{\operatorname{Mac}_{#1#2}(#3,#4)}
\newcommand{\supo}[1]{\operatorname{Supp}(#1)}

\newcommand{\lm}[2]{\operatorname{\mathsf{lm}}_{#1}(#2)}
\newcommand{\lmdrl}[1]{\lm{\grs}{#1}}

\newcommand{\lt}[2]{\operatorname{\mathsf{lt}}_{#1}(#2)}

\newcommand{\lc}[2]{\operatorname{\mathsf{lc}}_{#1}(#2)}

\newcommand{\mon}[2]{\operatorname{\mathsf{Mon}}_{#1 #2}}
\newcommand{\hd}[1]{\operatorname{\mathsf{hom}}(#1)}

\newcommand{\K}{\mathbb{K}}
\newcommand{\Kp}[1]{\mathcal{R}_{#1}}
\newcommand{\Khp}[2]{\mathcal{R}_{#1}^{(#2)}}%homogeneous polynomials of a given degree

\newcommand{\nbpol}{s}
\newcommand{\nbgb}{r}

\newcommand{\NF}[2]{\operatorname{NF}\left( #1,#2 \right) }

\newcommand{\thedeg}[0]{\delta}

\newcommand{\ind}[1]{\operatorname{\mathsf{idx}}(#1)}
\newcommand{\ech}[1]{\mathsf{EchelonForm}(#1)}
\newcommand{\modify}[4]{\mathsf{Modify}(#1,#2,#3,#4)}
\newcommand{\polrows}[2]{\operatorname{Pol}(#2,#1)}
\newcommand{\seqpol}[0]{\bm{f}}
\newcommand{\seqmono}[0]{\mathbb{M}}
\newcommand{\seqmonot}[1]{\mathbb{M}_{#1}}
\newcommand{\redensmon}[2]{\zeta(#1,#2)}
\newcommand{\normfsubr}[2]{\mathsf{NormalForm}(#1,#2)}
\newcommand{\seqpolt}[0]{\textbf{h}}
\newcommand{\ideal}[1]{\langle #1 \rangle}

\newcommand{\field}{\K} % base field
\newcommand{\vectorspace}[2]{\operatorname{Span}_{#2}(#1)}

\newcommand{\grs}[0]{\succ_{\drl}}
\newcommand{\grp}[0]{\prec_{\drl}}
\newcommand{\grseq}[0]{\succeq_{\drl}}
\newcommand{\gs}[0]{\succ_{\dl}}
\newcommand{\gp}[0]{\prec_{\dl}}

\newcommand{\bdgshort}[1]{G_{#1}}
\newcommand{\bdg}[4]{G_{#1#2}(#3,#4)}
\newcommand{\bdgh}[4]{\mathscr{M}_{#1#2}(#3,#4)}
\newcommand{\bdgb}[4]{\mathscr{B}_{#1#2}(#3,#4)}

\newcommand{\lcmt}[2]{\bm{\lambda}_{#1}(#2)}
\newcommand{\taylorbasis}{\mathscr{T}}

\newcommand{\enspol}[3]{\mathsf{P}_{#1,#2}(#3)}
\newcommand{\enspoltr}[3]{\mathsf{TP}_{#1,#2}(#3)}
\newcommand{\truncfu}[2]{\phi(#1,#2)}

\newcommand{\hilse}[2]{\operatorname{HS}_{#1}(#2)}
\newcommand{\zarimore}[2]{\mathcal{D}_{#1,#2}}
\newcommand{\zarisemi}[3]{\mathcal{S}_{#1,#2,#3}}
\newcommand{\zarifilling}[3]{\mathcal{A}_{#1,#2,#3}}

\newcommand{\pair}[2]{p_{#1,#2}}
\newcommand{\Spol}[2]{\mathsf{sp}(#1,#2)}
\newcommand{\setpair}[0]{P}
\newcommand{\Sel}[1]{\operatorname{Sel}(#1)}

\newcommand{\gbalgo}[1]{G_{#1}}
\newcommand{\allpalgo}[1]{P(#1)}
\newcommand{\curdeg}[1]{\degalgo(#1)}
\newcommand{\asseralgo}[2]{\mathcal{H}_{#2}(#1)}
\newcommand{\asseralgoshort}[1]{\mathcal{H}_{#1}}
\newcommand{\syzy}[2]{S_{#1,#2}}
\newcommand{\type}[0]{distinguished}
\newcommand{\proj}[2]{\pi_{#2}\left(#1\right)}

\newcommand{\drl}[0]{\mathrm{grevlex}}
\newcommand{\dl}[0]{\mathrm{grlex}}
\newcommand{\opark}[1]{\mathsf{T}{(#1)}}

\newcommand{\selered}[3]{\Psi \left( #1,#2,#3 \right)}

\newcommand{\nbred}[2]{\rho_{#1,#2}}
\newcommand{\nbredw}[2]{\widetilde{\rho}_{#1,#2}}
\newcommand{\card}[1]{\sharp \left( #1 \right)}

\newcommand{\compfq}[0]{N_{L}}
\newcommand{\compfqt}[0]{N_{F4T}}
\newcommand{\compfqtb}[0]{\widetilde{N}_{F4T}}
\newcommand{\compfc}[0]{N_{F5}}
\newcommand{\cofc}[2]{b_{#1}^{(#2)}}

\newcommand{\degalgo}[0]{d}
\newcommand{\indalgo}[0]{\iota}
\newcommand{\dregprev}[0]{\mathbb{D}_{n-1}}
\newcommand{\dregmb}[0]{\Dmac}

\newcommand{\explpairs}[0]{P}
\newcommand{\explred}[0]{R}
\newcommand{\explSp}[0]{S}
\newcommand{\expluni}[0]{L}

\newcommand{\serieb}[2]{B_{#1}(#2)}
\newcommand{\thegoodr}[2]{r(#1,#2)}
\newcommand{\thegoodd}[1]{d(#1)}

\newcommand{\valint}[2]{t(#1)}
\newcommand{\pourcbm}[0]{p}

\newcommand{\proporMB}{p}

\usepackage{algorithm}
\usepackage[noend]{algpseudocode}
\algnewcommand\algorithmicassume{\emph{Assume:}}
\algnewcommand\Assume{\item[\algorithmicassume]}
\algrenewcomment[1]{\hfill \(\triangleright\) \emph{\small #1}} % to italicize text in comments
\algnewcommand{\CommentLine}[1]{\(\triangleright\) \emph{\small #1}}
\newcommand{\algoName}[1]{Algorithm \nameref{#1}}
\makeatletter
  \newcommand{\algoCaptionLabel}[2]{
     \caption[\textproc{\upshape #1}]{\textproc{\upshape #1}\ifthenelse{\equal{#2}{}}{}{$(#2)$}}%
     \NR@gettitle{\textproc{\upshape #1}}%
      \label{algo:#1}%
     }%
\makeatother

\usepackage{geometry}
\usepackage{graphicx}
\usepackage{xspace}
\usepackage{epstopdf}
\usetikzlibrary{decorations.pathreplacing}
\usepackage[hypertexnames=false]{hyperref}
\hypersetup{
	colorlinks=true,
	linkcolor=blue,
	filecolor=magenta,      
	urlcolor=cyan,
}
\usepackage{chngcntr}
\usepackage[capitalise]{cleveref}

\newtheorem{thm}{Theorem}[section]

\newtheorem{lemma}[thm]{Lemma}
\newtheorem{cor}[thm]{Corollary}
\newtheorem{definition}[thm]{Definition}
\newtheorem{definition*}{Definition}
\newtheorem{prop}[thm]{Proposition}

\newtheorem{remark}[thm]{Remark}

\crefname{prop}{Proposition}{Propositions}
\Crefname{algorithm}{Algorithm}{Algorithm}

\newcommand{\msolve}{\texttt{msolve}}

\begin{document}

\begin{frontmatter}

\title{A complexity analysis of the F4 Gröbner basis algorithm with tracer data}

% \author{\footnotesize{\shortstack{Robin Kouba$^{\ast}$ \\ Sorbonne Université, CNRS, LIP6 \\ F-75005 Paris, France \\ Robin.Kouba@lip6.fr} }
% \hspace{0.6cm}
% \footnotesize{\shortstack{Vincent Neiger \\ Sorbonne Université, CNRS, LIP6 \\ F-75005 Paris, France \\ Vincent.Neiger@lip6.fr} }
% \hspace{0.6cm}
% \footnotesize{\shortstack{ Mohab Safey El Din  \\ Sorbonne Université, CNRS, LIP6 \\ F-75005 Paris, France \\ Mohab.Safey@lip6.fr} }}

% \author{\shortstack{Robin Kouba ~~~~ Vincent Neiger$^{\ast}$ ~~~~ Mohab Safey El Din \\[0.2cm] Sorbonne Université, CNRS, LIP6, F-75005 Paris, France}}
\author{\shortstack{Robin Kouba ~~~~ Vincent Neiger ~~~~ Mohab Safey El Din \\[0.2cm] Sorbonne Université, CNRS, LIP6, F-75005 Paris, France}}

% \cortext[cor1]{Corresponding author.}

\begin{abstract}
We provide a new complexity bound for the computation of grevlex Gröbner bases
in the generic zero-dimensional case, relying on Moreno-Socías' conjecture. We
first formalize a property of regular sequences that implies a well-known
folklore consequence, which we call the increasing degree property.  We then
derive a new understanding of the selection of pairs in the F4 algorithm based
on Moreno-Socías' conjecture. Moreover, we obtain an exact formula for the
number of elements in the grevlex Gröbner basis of a given degree, for half of
the relevant degrees. Combining these results, we derive a precise complexity
formula for the F4 Tracer algorithm, together with its asymptotic behavior when
the number of variables tends to infinity. These results yield an improvement
over the state-of-the-art complexity bounds by a factor which is exponential in
the number of variables. 
\end{abstract}


\end{frontmatter}

\section{Introduction}
\label{section_intro}

\paragraph*{Gröbner basis computations}

Gröbner bases are a key tool for representing polynomial ideals with respect to
some monomial ordering; we refer the reader to \citep{cox94} for basic
definitions and properties. Since Gröbner bases allow one to solve the ideal
membership problem through a multivariate division algorithm, or also to
compute ``modulo'' an ideal,  they have also become a versatile and widely used
tool for solving systems of polynomial equations. 

For these reasons, the complexity of computing them has been extensively studied.
In the 80's, much of the focus was on delivering lower bounds (see
\citep{mayr1982complexity,Huynh86}) and upper
bounds (see \citep{giusti1984, Moller_Mora84} for the characteristic zero
case and \citep{dube90} for the positive characteristic case),
establishing a worst-case behaviour of Gröbner
bases which is doubly exponential in the number of
variables. These worst-case results have been greatly improved under some
additional assumptions (generically satisfied, or conjectured to be
so, in the sense of Zariski topology), such as
zero-dimensionality, regularity, simultaneous Noether position (see
\citep{lazard83, giusti1984, bardet2014complexity}), which are
quite realistic in many application contexts. 

Still, these bounds do not depend much on the algorithm used to
compute Gröbner bases and were proved at a time when only
Buchberger's algorithm (see \cite{buchberger1965algorithmus}) was the single algorithmic
masterpiece of Gröbner basis theory. It has been observed that the practical behaviour of this algorithm,
which essentially consists in iteratively reducing $S$-polynomials
associated with critical pairs modulo a current basis, strongly
depends on the order in which these pairs are selected
(see, e.g., \cite{Traverso_sugar}).
It was also observed that, usually, most of these reductions lead to compute
zero, and that being able to detect these zero reductions before actually
performing the reduction would save significant computational time.

\paragraph*{Algorithm F5}

This F5 algorithm \cite{faugere2002F5} and its variants
\cite{eder2017survey} provide an algorithmic framework which avoids
all reductions to zero under some regularity assumption. These 
algorithms can be combined with a linear algebra view of the
aforementioned reductions through row echelon form computations on Macaulay 
matrices \cite{macaulay1902some, lazard83}. 
This comes at the cost of introducing \emph{signatures}
associated with the polynomials arising during the Gröbner basis
computation, which significantly restricts the set of row
permutations that can be used in the linear algebra
steps. These signatures can also be exploited, using adapted criteria,
to take into account the structure of the polynomials
\cite{gopalakrishnan2024optimized}.
(Concerning column permutations in these row echelon form computations,
note that they are basically forbidden, independently of the used algorithm,
due to the construction of the Macaulay matrix and the fact that Gröbner bases
are computed with respect to monomial orderings.)

Still, one of the variants enjoys a \emph{dedicated}
complexity analysis in \cite{bardet2014complexity}, which leverages
the special design of the algorithm, under some genericity
assumptions.
However, probably because of the inflexibility in linear algebra steps,
combined with the difficulty of handling signatures efficiently, the
F5 algorithm is for example not the one which is implemented in the fastest
available Gr\"obner basis software tools.

\paragraph*{Algorithm F4 and Gr\"obner traces}

The predecessor of F5, the so-called F4 algorithm \cite{faugere1999new} has led to 
significant speed-up in
practice against Buchberger's algorithm, and it is currently the one
which is implemented in most leading computer algebra software. 
However, to the best of our knowledge, this observed efficiency of F4 is not
explained by any dedicated complexity analysis which would leverage its
specific design, even under some additional assumptions. As a matter of fact,
all complexity statements related to F4 that we could find in works spanning
the last two decades are based on regularity assumptions and estimates on the
degree of regularity (to control the number and the size of Macaulay matrices);
see for example 
\cite{FaSaSpa13, FaSaVe16, Hashemi2011, Hashemi2017, GORLA2022386, CAMINATA2023322}.

In this framework, unlike in F5, the reduction to row echelon form leaves 
great freedom in choosing reducers or pivots for Gaussian elimination,
and in which row permutations can be applied.

Avoiding reductions to zero in the F4 algorithm has not found
yet a definitive solution in this framework.  A criterion first introduced in
\cite{buchberger1965algorithmus} and further developed in
\cite{gebauer1988installation} allows one to detect a certain number
of pairs that necessarily reduce to zero, but not all. When the base
field is large enough, one can take random linear
combinations of subsets of the rows to reduce to detect when some of
these all reduce to zero. This technique was introduced in Magma \citep{magma} and
has also implemented in Maple \citep{maple}
and in \msolve{} \citep{berthomieu}.
 While it leads to significant speedup, its impact on the overall complexity
of the computation is not studied yet, even under some regularity
assumptions. This technique does not eliminate all reductions to
zero, but it does remove a significant fraction of them.

A more systematic way to avoid reductions to zero is inspired from the
tracing approach introduced in \cite{traverso1989grobner} for
Buchberger's algorithm, which readily extends to the F4 context. 
Running the F4 algorithm on a given input sequence of polynomials, 
using a prescribed selection strategy for critical pairs such as the 
so-called ``normal strategy'' in \cite{faugere1999new} produces the 
structure of the leading monomials of all algebraic combinations of the 
input polynomials at each degree. 
This information of all leading monomials generated throughout the computation 
is referred to as the \emph{Gröbner trace} associated with the 
Gr\"obner basis computation. More precisely, using a data structure for this trace
which stores monomial information indicating which rows are not reduced to zero 
and which reducers are used, it becomes possible to eliminate all reductions 
to zero when computing a Gröbner basis for input sequences that are \emph{compatible} 
with this Gröbner trace.
Although this is a strong assumption, it often occurs in practice when 
solving polynomial systems over the rationals via a multi-modular approach
\citep{traverso1989grobner}, 
or when dealing with systems depending on parameters (see e.g.\ \cite{LeSa22}). 
This leads to the so-called F4-Tracer algorithm, implemented for instance in 
Maple \citep{maple} or 
in the \msolve{} library, which has demonstrated practical efficiency 
(see \cite{berthomieu}).

\begin{quotation}
\emph{This paper aims to establish complexity results for
the F4-Tracer algorithm, leveraging its structure, under some
classical genericity assumptions, which we make explicit below.}
\end{quotation}

\paragraph*{Main results}

To describe our results, we need to introduce some notation.
Throughout the paper, we denote by $\Kp{n} = \K[x_1,\ldots,x_n]$ the ring of multivariate polynomials
in the variables $x_1, \ldots , x_n$ over a field $\K$, and by
$\Kp{n}^{(\leq \delta)}$ the subset of $\Kp{n}$ consisting of all polynomials
of degree at most $\delta$.
For a polynomial $f$, we denote its homogeneous component of highest degree by $\hd{f}$.
For a sequence of polynomials $\seqpol = (f_1,\ldots , f_{\nbpol})$, we define
$\hd{\seqpol} = (\hd{f_1},\ldots , \hd{f_{\nbpol}})$.

For a positive real number $\proporMB \in (0,1]$ and an integer $\delta \geq 2$, we define
\[
E_{\proporMB}(\delta) =
\ln\left( \frac{1+\proporMB(\delta-1)}{\delta} \right)
+
\proporMB(\delta-1)\ln\left(\frac{1+\proporMB(\delta-1)}{\proporMB(\delta-1)}\right)
\]
as well as
\[
L_{\omega}(\delta)
=
\frac{\delta}{1-e^{-\frac{1}{(\omega-2)}}}
-
\frac{1}{1-e^{-\frac{1}{\delta(\omega-2)}}}.
\]
Here $\omega$ denotes the exponent of matrix multiplication over $\K$.
Note that $L_{\omega}(\delta)$ is not defined when $\omega = 2$; in this case,
by convention, we set $L_{2}(\delta) = 1$ for all $\delta \geq 2$.
Furthermore, we denote by $\bm{\ell}(\omega)$ the constant $L_{\omega}(2)$,
and for any real number $\varepsilon > 0$, we define the constant
\[
\bm{c}(\varepsilon, \delta, \omega)
=
E_{1/2}(\delta) + E_{\bm{\ell}(\omega)+\varepsilon}(\delta).
\]

As discussed above, the F4-Tracer algorithm takes as input a sequence of
polynomials together with a Gröbner trace $\mathcal{T}$, assuming that
both are compatible.
The algorithm that constructs this Gröbner trace from a given sequence of
polynomials is denoted by \algoName{algo:F4B}, whereas the F4-Tracer
algorithm itself is denoted by \algoName{algo:F4T}.

We use the classical notation $\mathcal{O}$ and $\widetilde{\mathcal{O}}$
  for asymptotic bounds, as defined in
\cite[Definitions~25.7 and~25.8]{vzgg03}.
We recall that for two functions $f,g : \mathbb{N} \to \mathbb{R}_{\ge 0}$,
the function $f(n)$ belongs to $\widetilde{\mathcal{O}}\!\left(g(n)\right)$
if there exist constants $N,c$ such that
\[
f(n) \le g(n)\bigl(\log_2(3+g(n))\bigr)^c
\]
for all $n \ge N$.

\begin{thm}
\label{thm:main}
Let $\seqpol = (f_1, \ldots , f_n)$ be a sequence of nonzero polynomials
in the ring $\Kp{n}$, all of degree $\delta \geq 2$, with $n \geq 2$.
Assume that $\delta$ is fixed.
Let $\mathcal{T}$ be the Gr\"obner trace returned by
\algoName{algo:F4B} when executed with input $\seqpol$ and the graded
reverse lexicographical ordering $\grs$.

There exists a Zariski open set
$\Omega \subset \big( \Kp{n}^{(\leq \delta)} \big)^n$
such that, if $\hd{\seqpol}$ lies in $\Omega$, then running
\algoName{algo:F4T} on input a sequence $\bm{g} = \left( g_1, \ldots,
g_n\right)$ of polynomial equations which is compatible with
$\mathcal{T}$,
computes a minimal Gröbner basis for $(\bm{g}, \grs)$ using 
\[
\widetilde{\mathcal{O}}\!\left(
\delta^{\omega n+1} e^{\bm{c}(\varepsilon, \delta, \omega)\, n}
\right)
\]
arithmetic operations in $\K$, for all $\varepsilon > 0$.
\end{thm}

A more precise version of this result is given in
\cref{thm_complexity_F4_tracer_assymp} in
\cref{sec:complexity_proof}. 
The non-emptiness of the Zariski open set $\Omega$ relies on
\cite[Conjecture~4.1]{moreno2003degrevlex},
\cite[Conjecture~1.1]{froberg1994hilbert} and 
\cite[Conjecture~1.6]{moreno2003degrevlex}.

The hidden factor in the $\widetilde{\mathcal{O}}$ notation of
\cref{thm:main} is bounded by a polynomial in $n$ whose degree
grows linearly with $\omega$ and does not depend on $\delta$
or $\varepsilon$.
The coefficient $\bm{c}(\varepsilon, \delta, \omega)$ increases when
$\varepsilon$ increases, and for any $\delta \geq 2$ we have that
\[
  \lim_{\varepsilon \to 0^+} \bm{c}(\varepsilon, \delta, \omega) =
E_{1/2}(\delta) + E_{\bm{\ell}(\omega)}(\delta)  .
\tag*{($\star$)} \label{introval}
\] 
For example, with $\delta = 2$, we have that \ref{introval} is approximately 
equal to $0.2616$ for $\omega = 3$, to $0.6743$ for $\omega = 2.81$ and to 
$0.7956$ for $\omega = 2.38$.

For algorithms such as \algoName{algo:F4T}, complexity results
follow from the understanding of three different quantities.  
The first one is the number of new elements found at each step, which has been
thoroughly studied in \cite{bardet2014complexity}.  
The second one is the number of reductors used at each step to compute these
new polynomials; this is where our analysis leads to an improvement with
respect to the state of the art.  
Finally, the third quantity is the support of the polynomials involved,
an aspect that we treat here the same way as has been done in the state of the art.

In order to provide a sharper view of the asymptotic result of
\cref{thm:main}, we compare it with
\cite[Theorem~2 and Section~3.2]{bardet2014complexity}.
That result defines a formula $\compfc$, 
which provides an upper bound on the number of arithmetic operations in $\K$ 
needed to compute a grevlex Gröbner basis of a generic sequence 
$(f_1, \ldots, f_n)$ of polynomials in $\Kp{n}$ of degree $\delta$ 
using the F5 algorithm.
Moreover, for a fixed value of $\delta$, we have
\[
\compfc \in \widetilde{\mathcal{O}}\!\left(
\delta^{3n} h(\delta)^n
\right),
\]
where, for $\delta \geq 2$,
 it holds that $1 \leq h(\delta) \leq 3$.
Besides, the function $h(\delta)$ is increasing and
\[
\lim_{\delta \to +\infty} h(\delta) \approx 2.814.
\]

We estimate the gain provided by our bound on the arithmetic complexity of
\algoName{algo:F4T} with respect to that of F5 by considering the following ratio:
\begin{align*}
\frac{\delta^{3n} 2.81^n}{\delta^{\omega n} 
e^{n ( E_{1/2}(\delta)+E_{\bm{\ell}(\omega)}(\delta) )}} 
= 
\left( 
\frac{2.81\, \delta^{3-\omega}}
     {e^{E_{1/2}(\delta)+E_{\bm{\ell}(\omega)}(\delta)}} 
\right)^n .
\end{align*}
\Cref{fig:compassymp} shows the base of this exponential gain,
\[
\bm{g}(\omega,\delta) = \frac{2.81\,\delta^{3-\omega}}{e^{E_{1/2}(\delta)+E_{\bm{\ell}(\omega)}(\delta)}},
\]
as a function of $\delta$ for several values of $\omega$.
The two plots in \cref{fig:compassymp} represent the same function over different
ranges of $\delta$.
The first plot illustrates the behavior for small values of $\delta$,
whereas the second one highlights the asymptotic behavior for larger values of
$\delta$.

\vspace{0.5cm}

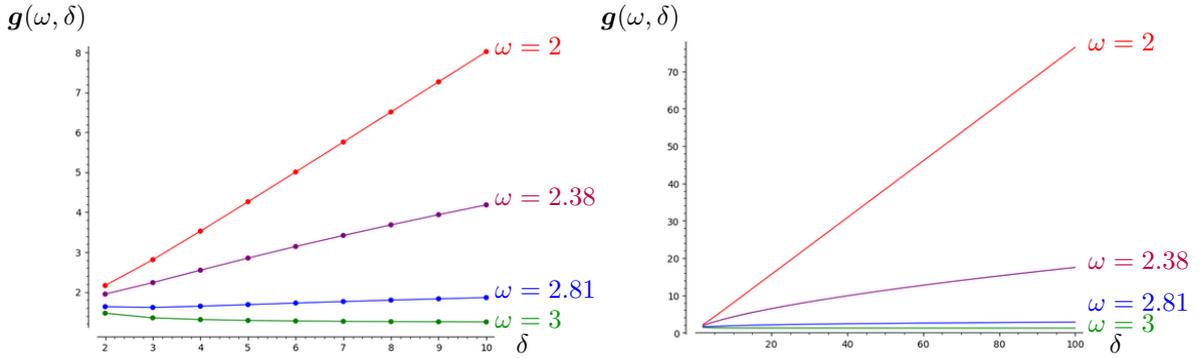
\begin{figure}[htb]
\centering
        \begin{overpic}[width=0.35\textwidth]{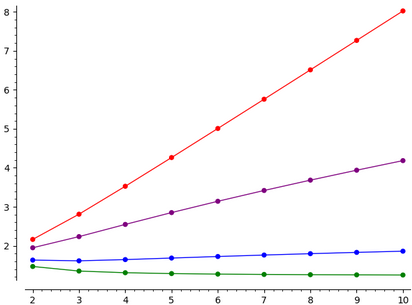}
          \put(-15,78){$\bm{g}(\omega,\delta)$}
          \put(105,0){$\delta$}
          \put(100,6){\color{green!60!black}$\omega = 3$}
          \put(100,13){\color{blue}$\omega = 2.81$}
          \put(100,35){\color{purple}$\omega = 2.38$}
          \put(100,71){\color{red}$\omega = 2$}
        \end{overpic}
        \hspace{2cm}
        \begin{overpic}[width=0.35\textwidth]{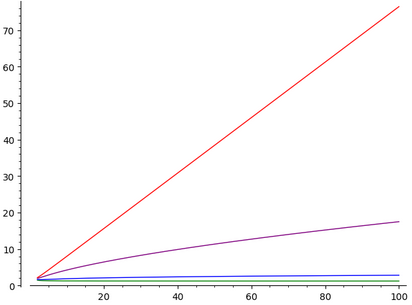}
        \put(-15,78){$\bm{g}(\omega,\delta)$}
          \put(105,0){$\delta$}
          \put(100,5){\color{green!60!black}$\omega = 3$}
          \put(100,10){\color{blue}$\omega = 2.81$}
          \put(100,20){\color{purple}$\omega = 2.38$}
          \put(100,72){\color{red}$\omega = 2$}
        \end{overpic}
\caption{exponential gain in the asymptotic}
\label{fig:compassymp}
\end{figure}

The gain can be interpreted as $\bm{g}(\omega,\delta)^{n}$. For example, setting
$\delta = 2$, we have $\bm{g}(\omega,2) = 1.4687$ when $\omega = 3$,
$\bm{g}(\omega, 2) = 1.6331$ when $\omega = 2.81$ (Strassen), and
$\bm{g}(\omega, 2)=1.9531$ when $\omega = 2.38$,
(slightly above the best currently known bound for 
$\omega$).
Moreover, for $\delta = 3$, $\bm{g}(\omega,3)$ equals $1.3518$ for $\omega = 3$,
$\bm{g}(\omega, 3) = 1.6139$ for $\omega = 2.81$, and 
$\bm{g}(\omega, 3) = 2.2348$ for $\omega = 2.38$.
More generally, for a fixed value of $\omega$, the function
$\bm{g}(\omega,\delta)$ is asymptotically equivalent to
$\bm{t}(\omega)\,\delta^{3-\omega}$ as $\delta$ tends to infinity,
where $\bm{t}(\omega)$ is a constant.

\paragraph*{Conclusions and perspectives}

To the best of our knowledge, this is the
first dedicated complexity analysis of a variant of the F4 algorithm,
even under the strong assumption of a priori knowing a
Gr\"obner trace and under genericity assumptions.
As a matter of fact, the complexity bound above yields
better estimates than the best known ones on the variant of the F5
algorithm studied in \cite{bardet2014complexity}. This may be the
first step to a better understanding of the folklore corresponding observation on running times
with the available software implementations. However, this raises 
more questions for future investigation than it provides definitive conclusions 
regarding the comparison of these algorithms.
First, let us
recall that our assumptions are much stronger than the ones made in
\cite{bardet2014complexity}. Second, here, only upper bounds are
compared and both rely on methodologies that overestimate the number
of columns of Macaulay matrices arising during the computation.

Still, this result shows that it may be relevant to analyze the
complexity of the probabilistic variant of the F4 algorithm which relies
on the aforementioned probabilistic linear algebra steps and which
removes a significant number of reductions to zero. It also settles
the question of removing all reductions to zero in the framework of
the F4 algorithm, while keeping its appealing freedom on linear
algebra steps. 

\paragraph*{Structure of the paper}
In \cref{sec:degree_property}, we establish several properties of regular sequences
that imply what we call the \emph{increasing degree property}.
In \cref{sec:criticalpairs}, we prove a main original theorem concerning the form of
the selected critical pairs using
a criterion first introduced in \cite{buchberger1965algorithmus} and further
developed in \cite{gebauer1988installation}.
In \cref{sec:shapemonomialstaircase}, we show some structural properties on the
grevlex Macaulay matrices and we provide a customized version of a particular case of
\cite[Corollary~3.2]{moreno2003degrevlex}.
\cref{sec:Algorithmsandtheirproperties} is devoted to the description of the
algorithms and their main properties, whereas
\cref{sec:complexity_analysis} focuses on complexity analyses and comparisons.

\paragraph*{Acknowledgments} The authors are supported by the Agence Nationale
de la Recherche, which funds the ANR-23-CE48-0003 project CREAM and the
ANR-25-CE39-7662 project PQMC, and by the grant FA8655-25-1-7469 of the
European Office of Aerospace Research and Development of the Air Force Office
of Scientific Research (AFOSR).

\section{On the increasing degree property}
\label{sec:degree_property}

Further, for a nonzero polynomial $f \in \Kp{n}$, we denote by $\supo{f}$ its support.
For a finite set $\seqpol$ of polynomials, the support of $\seqpol$ is
defined as the union
$\supo{\seqpol} = \bigcup_{f \in \seqpol} \supo{f}$.

This section is devoted to prove \cref{prop_no_degree_fall}, stated below.
We start by recalling the definition of a homogeneous regular sequence 
\cite[Section 16]{matsumura1989commutative}.

\begin{definition}
\label{def_reg_seq}
Let $ \seqpol=(f_1, \dots, f_\nbpol)$ be a sequence of nonzero homogeneous polynomials  in $\Kp{n}$. 
The sequence $\seqpol$ is said to be a \emph{regular} sequence if
for all $i$ in $\{2, \ldots , \nbpol \} $, 
$ f_i $ does not divide zero in the quotient ring $ \Kp{n} / \langle f_1, \dots, f_{i-1} \rangle$.
\end{definition}

We recall that for a polynomial \( f \) in $\Kp{n}$, we denote by \( \hd{f} \) its homogeneous
component of highest degree.  For a subset \(\seqpol \subset \Kp{n}\), 
\(\hd{\seqpol}\) denotes the set \(\{\hd{f}\mid f\in \seqpol\}\). 

\begin{prop}
\label{prop_no_degree_fall}
Let $ \seqpol=(f_1, \dots, f_{\nbpol})$ be a sequence of nonzero polynomials 
in $\Kp{n} $ with $\deg(f_i)=\delta_i$ and suppose that  $\hd{\seqpol}$ is a
regular sequence. Let $d$ be an integer. 
For any combination \(w= q_1 f_1 + \cdots + q_\nbpol f_\nbpol\)
with $w, q_1, \ldots, q_\nbpol$ in \(\Kp{n}\),
such that $\deg(w)<d$, and $\deg(q_i) \leq d-\delta_i$ for all \(i\),
there exist polynomials $h_1, \ldots, h_\nbpol$ such that
$w=h_1 f_1 + \cdots + h_\nbpol f_\nbpol$ and $\deg(h_i)<d-\delta_i$ for all \(i\).
\end{prop}

An important consequence of this result is
a folklore fact concerning Gr\"obner basis computations:
when computing a graded Gröbner basis using linear algebra-based algorithms
such as Lazard's algorithm \citep{lazard83}, F4 \cite{faugere1999new}, or F5 \cite{faugere2002F5},
the degrees of the polynomials added in the Gröbner basis of
the algorithm strictly increase when processing critical pairs.

\medskip The proof of \cref{prop_no_degree_fall} relies on the two auxiliary
lemmata below. We start with some notation used hereafter.  For $\bm{\alpha} =
(\alpha_1,\dots,\alpha_n)$ in $\mathbb{N}^n$, a monomial $\prod_{i=1}^n
x_i^{\alpha_i}$ is written  as $\bm{x}^{\bm{\alpha}}$. Moreover, we set \(
|\bm{\alpha}| = \sum_{i = 1}^n \alpha_i \). We denote by $\mon{d,}{n}$ the set
of monomials in $\Kp{n}$ of total degree $d$, and by $\mon{\leq d,}{n}$ (resp.\
$\mon{< d,}{n}$) the set of monomials in $\Kp{n}$ of 
degree at most $d$ (resp.\ less than $d$). The set of all monomials in $\Kp{n}$
is denoted  by $\mon{}{n}$. For \( d \in \mathbb{N} \), we denote by \(
\Khp{n}{d} \) the
\( \K \)-vector space generated by all homogeneous polynomials of degree \( d
\) in \( \Kp{n} \). For a set $S$ in a $\K$-vector space, we denote by 
$\operatorname{Span}_{\K}\left(S\right)$ the $\K$-vector subspace generated by $S$.
\begin{lemma}
\label{lemma_generator_homog_space}
Let \( \seqpol = (f_1,\dots,f_\nbpol) \) be homogeneous polynomials in \(
\Kp{n} \) with \( \deg(f_i) = \delta_i \), and let $d$ be in $\mathbb{N}$.
Then
\[
\enspol{d}{n}{\seqpol} \;=\; \left\{m\,f_i \mid i \in \{1,\dots,\nbpol\},\; m \in \mon{d-\delta_i,}{n} \right\}
\]
generates the \( \K \)-vector space of all homogeneous polynomials of degree \(
d \) that lie in 
\( \langle \seqpol \rangle \).
\end{lemma}

\begin{proof}
    Let \(I\) denote the ideal 
\( \langle \seqpol \rangle \).
We show that
\[
I \cap \Khp{n}{d} \;=\;
\operatorname{Span}_{\K}\left(\enspol{d}{n}{\seqpol}\right).
\]

\textbf{$(\supseteq)$}  
Every polynomial \( m f_i \) lies in \( I \cap \Khp{n}{d} \); linear combinations preserve membership.

\medskip
\textbf{$(\subseteq)$}  
Conversely, let $f$ be in $I \cap \Khp{n}{d}$.
If $f = 0$, then $f$ lies in $\operatorname{Span}_{\K}\left(\enspol{d}{n}{\seqpol}\right)$.
Suppose that \( f \neq 0 \).
Since $f$ lies in $I$, we can write  \( f = \sum_{i=1}^\nbpol u_i f_i \) with \( u_i
\in \Kp{n} \). Since \(f \neq 0\) by assumption, it holds that there exists
\(1\le  i \le  \nbpol\) such that \(u_i \neq 0\). We denote by \(\mathscr{Q}\) the
set of such integers. 

For \(i \in \mathscr{Q}\), write
\[
u_i = \sum_{ \bm{x}^{\bm{\beta}} \in \supo{u_i}} b_{i,\bm{\beta}}\,\bm{x}^{\bm{\beta}}, 
\qquad
f_i = \sum_{ \bm{x}^{\bm{\alpha}} \in \supo{f_i}} a_{i,\bm{\alpha}}\,\bm{x}^{\bm{\alpha}},
\]
with \( a_{i,\bm{\alpha}}, b_{i,\bm{\beta}} \in \K \).  Then
\[
u_i f_i
  \;=\;
\sum_{\substack{\bm{x}^{\bm{\beta}} \in \supo{u_i} \\ \bm{x}^{\bm{\alpha}} \in \supo{f_i} }}
      a_{i,\bm{\alpha}} b_{i,\bm{\beta}}\,
      \bm{x}^{\bm{\alpha} + \bm{\beta}},
\quad
\text{so}\;\;
f
  \;=\;
\sum_{i\in \mathscr{Q}}
\sum_{\substack{\bm{x}^{\bm{\beta}} \in \supo{u_i} \\ \bm{x}^{\bm{\alpha}} \in \supo{f_i}} }
      a_{i,\bm{\alpha}} b_{i,\bm{\beta}}\,
      \bm{x}^{\bm{\alpha} + \bm{\beta}}.
\]

By separating the terms according to whether $|\bm{\beta}| = d - \delta_i$ or not, we obtain

\[
f \;=\;
\underbrace{\sum_{i\in \mathscr{Q}}
            \!\!\!
            \sum_{\substack{\bm{x}^{\bm{\alpha}} \in \supo{f_i} \\ \bm{x}^{\bm{\beta}} \in \supo{u_i}, |\bm{\beta}| \neq d-\delta_i } }
            a_{i,\bm{\alpha}} b_{i,\bm{\beta}}\,
            \bm{x}^{\bm{\alpha} + \bm{\beta}}}_{\text{degree}\;\neq d}
\;+\;
\underbrace{\sum_{i\in \mathscr{Q}}
            \!\!\!
            \sum_{\substack{\bm{x}^{\bm{\alpha}} \in \supo{f_i} \\ \bm{x}^{\bm{\beta}} \in \supo{u_i}, |\bm{\beta}| = d-\delta_i } }
            a_{i,\bm{\alpha}} b_{i,\bm{\beta}}\,
            \bm{x}^{\bm{\alpha} + \bm{\beta}}}_{\text{homogeneous of degree }d}.
\]

Because \( f \) is homogeneous of degree \( d \geq 0 \), the first sum above
vanishes.  Hence
\[
f
  \;=\;
  \sum_{i\in \mathscr{Q}} \tilde{u}_i\,f_i,
\quad
\text{where }
\tilde{u}_i
  \;=\;
\sum_{\substack{\bm{x}^{\bm{\beta}} \in \supo{u_i} \\ |\bm{\beta}| = d-\delta_i}}
      b_{i,\bm{\beta}}\,\bm{x}^{\bm{\beta}}
  \in \Khp{n}{\,d-\delta_i}.
\]
Since each \(\tilde{u}_i f_i\) lies in  
\(\operatorname{Span}_{\K}\bigl(\{\,m f_i \mid  i \in \{1,\dots,\nbpol \}, m \in
\mon{d-\delta_i,}{n} \}\bigr)\), 
 we are done.
\end{proof}

\begin{lemma}
\label{lemma_homog_deg_fall}
Let $ \seqpol=(f_1, \dots, f_{\nbpol})$ be a sequence of nonzero polynomials 
in $\Kp{n} $ with $\deg(f_i)=\delta_i$ and suppose that  $\hd{\seqpol}$ is a
regular sequence.
For any combination \(w= q_1 f_1 + \cdots + q_\nbpol f_\nbpol\) with $w, q_1,
\ldots, q_\nbpol$ in \(\Kp{n}\), such that
$\deg(w)<d$, and \(q_i\) is either zero or homogeneous of degree $d-\delta_i$ for
all \(i\),
there exist polynomials $p_1, \ldots, p_\nbpol$ such that
$w=p_1 f_1 + \cdots + p_\nbpol f_\nbpol$ and $\deg(p_i)<d-\delta_i$ for all \(i\).
\end{lemma}

\begin{proof}
The first part of the proof establishes that for all $(i,j)$ in $\{1, \ldots , \nbpol \}^2$ with $i \geq j$,
there exists a homogeneous polynomial $a_{i,j}$ which is zero or of degree $d-\delta_i-\delta_j$ such that for all $k$ in $\{1, \ldots ,\nbpol\}$ we have 

\begin{align*}
\sum_{j = 1}^{k} \left( q_j + \sum_{i = k+1}^{\nbpol}a_{i,j} \hd{f_i}  \right) \hd{f_j} = 0 \tag*{$(1.k)$}
\end{align*}
and 
\begin{align*}
q_{k} = -\sum_{i = k+1}^{\nbpol} a_{i,k}\hd{f_i} + \sum_{j= 1}^k a_{k,j}\hd{f_j}. \tag*{$(2.k)$}
\end{align*}
Moreover, we prove that for all $k$ in $\{1,\ldots ,\nbpol \}$, we have that $a_{k,k}$ can be chosen to be zero.

We start by proving $(1.\nbpol)$ and $(2.\nbpol)$.
Since 
\(w=\sum_{j =1}^\nbpol q_jf_j\)
with $\deg(w)<d$, the homogeneous component of degree $d$ of $w$ is zero.
As for all $j$ in $\{1, \ldots , \nbpol \}$, the polynomial $q_j$ is zero or homogeneous 
of degree $d-\delta_j$ and $\deg(f_j) = \delta_j$,
then it follows that:
\begin{align*}
\sum_{j =1}^\nbpol q_j\hd{f_j}=0. \tag*{$(1.\nbpol)$}
\end{align*}

This implies that $q_{\nbpol}\hd{f_{\nbpol}}$ lies in the ideal $\langle \hd{f_1}, \ldots , \hd{f_{\nbpol-1}} \rangle$.
Since the sequence $\hd{\seqpol}$ is regular,
by \cref{def_reg_seq}, the polynomial $q_{\nbpol}$ can be written as:
\begin{align*}
q_{\nbpol}= \sum_{j=1}^{\nbpol-1}a_{\nbpol,j} \hd{f_j} 
\end{align*}
where $a_{\nbpol,j}$ is an element of $\Kp{n}$ for all $j$ in $\{ 1, \ldots , \nbpol - 1 \}$.
By assumption, if not \(0\), $q_\nbpol$ is homogeneous of degree $d-\delta_\nbpol$
and $\hd{f_1}, \ldots , \hd{f_{{\nbpol-1}}}$ are homogeneous of degree $\delta_1 ,
\ldots , \delta_{\nbpol}$ respectively.
Thus, by \cref{lemma_generator_homog_space}, the polynomial $a_{\nbpol,j}$
can be chosen to be zero or homogeneous of degree $d-\delta_{\nbpol}-\delta_{j}$ for $1 \le  j
\le  \nbpol-1$.
By setting $a_{\nbpol,\nbpol} = 0$, we can write 
\begin{align}
 q_{\nbpol}= \sum_{j=1}^{\nbpol}a_{\nbpol,j} \hd{f_j} \tag*{$(2.\nbpol)$}. 
\end{align}
If \(q_\nbpol\) is zero, we set \(a_{\nbpol, j} = 0\)  for \(1 \le  j \le  \nbpol\). This
establishes $(1.\nbpol)$ and $(2.\nbpol)$. 

Our reasoning is now by decreasing
induction.
We pick $k\in \{1, \ldots , \nbpol-1\}$ and we assume that for all $\ell$ in 
$\{k+1, \ldots ,\nbpol\}$,  
the statements $(1.\ell)$ and $(2.\ell)$ hold. We establish $(1.k)$ and $(2.k)$ starting with
\((1.k)\). 

Let us write $(1.k+1)$:
\begin{align*}
& \sum_{j = 1}^{k+1} \left( q_j + \sum_{i = k+2}^{\nbpol}a_{i,j} \hd{f_i}  \right) \hd{f_j} = 0 \tag*{$(1.k+1)$}. 
\end{align*}
Using $(2.k+1)$ to replace and isolate the last term of the sum, we obtain
\begin{align*}
\hd{f_{k+1}}\sum_{j = 1}^{k+1}a_{k+1,j}\hd{f_{j}} +\sum_{j = 1}^{k} \left( q_j + \sum_{i = k+2}^{\nbpol}a_{i,j} \hd{f_i}  \right) \hd{f_j} = 0.
\end{align*}
Since, by our induction assumption, $a_{k+1,k+1} = 0$, 
this can also be expressed as follows:
\begin{align*}
\sum_{j = 1}^{k} \left( q_j + \sum_{i = k+1}^{\nbpol}a_{i,j} \hd{f_i}  \right) \hd{f_j} = 0
\end{align*}
which establishes $(1.k)$.

Now, we prove $(2.k)$ and let
\( f = q_{k} + \sum_{i = k+1}^{\nbpol}a_{i,k} \hd{f_i} \).
Since by $(1.k)$ we have that 
\[ f\hd{f_k} = -\sum_{j = 1}^{k-1} \left( q_j + \sum_{i = k+1}^{\nbpol}a_{i,j} \hd{f_i}  \right) \hd{f_j}, \]
and since $\hd{\seqpol}$ is a regular sequence, we deduce that  
$f\in \langle \hd{f_1}, \ldots , \hd{f_{k-1}} \rangle$.
Since $a_{i,j}$ (resp.\ \(q_j\)) is zero or homogeneous of degree $d-\delta_i-\delta_{j}$
(resp.\ \(d-\delta_j\)), the polynomial
$f$ is then zero or homogeneous of degree $d-\delta_{k}$.

Thus, by 
\cref{lemma_generator_homog_space} applied to \(\langle \hd{f_1},\ldots,
\hd{f_{i-1}} \rangle\), there exist 
homogeneous polynomials $a_{k,j}$ of degree \(d-\delta_k-\delta_j\) if not  \(0\), with $j$ in $\{1, \ldots , j\}$
such that 
\[ f = \sum_{j = 1}^{k-1} a_{k,j} \hd{f_j}  \]
Setting additionally $a_{k,k} = 0$ proves $(2.k)$.

\medskip

We can now prove the statement of the lemma.
Consider the polynomials $(b_{i,j})_{1 \leq i \leq \nbpol, 1 \leq j \leq \nbpol}$ with $b_{i,j}=a_{i,j}$ when $i \geq j$ and $b_{i,j}=-a_{j,i}$ when $j>i$. 
Note that the square matrix 
$(b_{i,j})_{1 \leq i \leq \nbpol, 1 \leq j \leq \nbpol}$ is anti-symmetric and that for all $i$ in $\{1, \dots , \nbpol\}$:
\begin{align}
q_i=\sum_{j \in \{1 , \ldots , \nbpol \}} b_{i,j}\hd{f_j} \tag*{$(\ast)$}
\end{align} 
by $(2.i)$.

For all $i$ in $\{1 , \ldots , \nbpol \}$, write $f_i=\hd{f_i}+\tilde{f_i}$ and 
\begin{align*}
w&=\sum_{i \in \{1,\ldots ,\nbpol \}}q_if_i
    = \sum_{i \in \{1,\ldots ,\nbpol \}}q_i\tilde{f_i} + \sum_{i \in \{1,\ldots ,\nbpol \}} q_i\hd{f_i}
    = \sum_{i \in \{1,\ldots ,\nbpol \}} q_i\tilde{f_i} \tag*{\text{ (by } $(1.\nbpol)$)} \\
&= \sum_{i,j \in \{1,\ldots ,\nbpol \}} b_{i,j}\hd{f_j} \tilde{f_i} 
= \sum_{i,j \in \{1,\ldots ,\nbpol \}} b_{i,j}f_j \tilde{f_i}  - \sum_{i,j \in \{1,\ldots ,\nbpol \}} b_{i,j}\tilde{f_j} \tilde{f_i}
\tag*{(by $(\ast)$)}.
\end{align*}
Since the matrix $(b_{i,j})_{1 \leq i \leq \nbpol, 1 \leq j \leq \nbpol}$ is
anti-symmetric and $\Kp{n}$ is commutative, it holds that the sum
\(
  \sum_{i,j \in \{1,\ldots ,\nbpol \}} b_{i,j}\tilde{f_j} \tilde{f_i}
\)
is equal to zero.
This implies that 
\[
w = \sum_{i,j \in \{1,\ldots ,\nbpol \}} b_{i,j}f_j \tilde{f_i}   
= \sum_{i,j \in \{1,\ldots ,\nbpol \}} b_{i,j}\tilde{f_i} f_j.
\]
For all $j$ in $\{1, \ldots, \nbpol \}$, let:
\[p_j= \sum_{i \in \{1 , \ldots , \nbpol \}} b_{i,j}\tilde{f_i}.\]
It remains to prove that $\deg(p_j) <d-\delta_j$. By definition, 
\(\deg(p_j) \leq \max_{i \in \{1, \ldots ,
\nbpol\}}(\deg(b_{i,j})+\deg(\tilde{f_i})) \) while for \(1\le  i \le  \nbpol\)
$\deg(\tilde{f_i})<\delta_i$ and $\deg(b_{i,j}) \leq d-\delta_i-\delta_j$. This implies that $\deg(b_{i,j})+\deg(\tilde{f_i})<d-\delta_j.$
As it is true for all $i$ in $\{ 1, \ldots , \nbpol \}$, then $\deg(p_j)<d-\delta_j$.
\end{proof}

We can now prove \cref{prop_no_degree_fall}.

\begin{proof}[Proof of \cref{prop_no_degree_fall}]
We set $E=\{i \in \{1 , \ldots , \nbpol \} \mid \deg(q_i)=d-\delta_i \}$ and we 
decompose $w$ as follows
\begin{align*}
w&=\sum_{i \in \{1, \ldots, \nbpol \}}q_if_i = \left( \sum_{i \in \{1, \ldots, \nbpol \}}(q_i-\hd{q_i})f_i \right) + \left( \sum_{i \in \{1, \ldots, \nbpol \}}\hd{q_i}f_i \right) \\
&= \left( \sum_{i \in \{1, \ldots, \nbpol \}}(q_i-\hd{q_i})f_i \right) + \left( \sum_{i \notin E}\hd{q_i}f_i \right) +\left( \sum_{i \in E}\hd{q_i}f_i \right).
\end{align*}
By hypothesis, the polynomial
\[ 
w -\left( \sum_{i \in \{1, \ldots, \nbpol \}}(q_i-\hd{q_i})f_i \right) - \left( \sum_{i \notin E}\hd{q_i}f_i \right)
 \]
has degree strictly less than $d$.
Then one can apply \cref{lemma_homog_deg_fall} to the last sum. This yields the following expressions for $w$:
\begin{align*}
w&= \left( \sum_{i \in \{1, \ldots, \nbpol \}}(q_i-\hd{q_i})f_i \right) + \left( \sum_{i \notin E}\hd{q_i}f_i \right) +\left( \sum_{i \in \{1, \ldots, \nbpol\}}p_if_i \right) \\
&=\left( \sum_{i \notin E}(q_i+p_i)f_i \right) + \left( \sum_{i \in E}(q_i-\hd{q_i}+p_i)f_i \right)
\end{align*}
where $\deg(p_i)<d-\delta_i$.
If $i$ does not lie in $E$, then we have that $\deg(q_i)<d-\delta_i$ which implies that
$\deg(q_i+p_i)<d-\delta_i$ and we set \(h_i = q_i+p_i\).
If $i$ lies in $E$, then $\deg(q_i-\hd{q_i})<d-\delta_i$ which implies that
$\deg(q_i-\hd{q_i}+p_i)<d-\delta_i$ and we set \(h_i = q_i-\hd{q_i}+p_i\). 
This concludes the proof.
\end{proof}

\section{Critical pairs}
\label{sec:criticalpairs}

The main result in this section is stated in \cref{thm_type1_good_pairs}, and makes
explicit a certain structure of the pairs selected in a
Buchberger-type algorithm (see \cite{buchberger1965algorithmus}). This
structure arises from the criterion developed in \cite{gebauer1988installation}.

The section is divided into three subsections.
\Cref{subsection_Onthemonomialstaircase} is dedicated to prove
\cref{thm_bijection_monomials}, a key preliminary result for establishing
\cref{thm_type1_good_pairs}. In \cref{subsection_Onthesyzygiesmodule}, we
introduce necessary context and notations to state 
\cref{thm_type1_good_pairs}, and we give a few preliminary lemmata for its
proof which is given in \cref{subsection_Onthesyzygiesmodule_proof}.

\subsection{On the monomial staircase}
\label{subsection_Onthemonomialstaircase}

We first introduce definitions and notations. 
We use classical notation and terminology concerning
(admissible) monomial orderings $\succ$ on the monomials of $\Kp{n}$,
as well as the notions of leading monomial $\lm{\succ}{f}$, leading term
$\lt{\succ}{f}$, and leading coefficient $\lc{\succ}{f}$ of a nonzero polynomial
$f$ (see \citep[Chapter~2, \S2]{cox94} for more details).
For a subset $S \subseteq \Kp{n}$, we denote by $\lm{\succ}{S}$ the set of
leading monomials of the elements of $S$.

We recall that
for an integer $i$, the set
 $\mon{}{i}$ is the set of monomials of $\Kp{i}$.
Let $f_1, \ldots ,f_{\nbpol}$ be polynomials in the ring $\Kp{n}$
and set $I = \langle f_1, \ldots ,f_{\nbpol} \rangle$.
For a given monomial ordering \(\succ\) on the monomials of $\Kp{n}$,
we first define \[\bdgb{}{}{I}{\succ} =  \mon{}{n} \setminus \lm{\succ}{I}\] the
set of monomials which are not leading monomials of polynomials in \(I\), i.e.\ a
monomial basis for \(\Kp{n} / I\) (see e.g.\ \cite[Chapter~5, Section~3,Proposition~4]{cox94}). This
set is then refined through the intersection with \(\mon{}{i}\) by defining 
\(\bdgb{}{i}{I}{\succ} = \bdgb{}{}{I}{\succ} \cap \mon{}{i}\) 
and next by degree by defining 
\[ \bdgb{d,}{i}{I}{\succ} = \{m \mid m \in \bdgb{}{i}{I}{\succ} \text{ and } \deg(m) = d \}.\]
Observe that 
\(\bdgb{}{}{I}{\succ} = \bdgb{}{n}{I}{\succ}\), that 
\(\bdgb{}{i-1}{I}{\succ} \subset \bdgb{}{i}{I}{\succ}\) for all 
\(i \in \{1,\ldots,n\}\), and that we have
\(\bdgb{}{i}{I}{\succ} = \bigsqcup_{d \geq 0} \bdgb{d,}{i}{I}{\succ}\).

For the convenience of the reader, we recall the notion of stability, 
see e.g.~\cite[Definition 4.2.2]{HerzogMono}.

\begin{definition}
A set $S$ of monomials of $\Kp{n}$ is said to be \emph{stable} if for any monomial $m$ in $S$
such that $x_i$ divides $m$ then $\frac{x_jm}{x_i}$ belongs to $S$ for all $j<i$.
\end{definition}

We denote by 
$\bdg{}{}{I}{\succ}$ 
the reduced $\succ$-Gröbner basis of $I$.
We define the following sets:
\begin{itemize}
  \item $\bdg{d,}{\ast}{I}{\succ}$ is the set of elements in $\bdg{}{}{I}{\succ}$ of degree $d$;
  \item $\bdg{\ast,}{i}{I}{\succ}$ is the set of elements in $\bdg{}{}{I}{\succ}$ whose leading monomial lies in $\Kp{i}$ but not in $\Kp{i-1}$;
  \item $\bdg{d,}{i}{I}{\succ} = \bdg{d,}{\ast}{I}{\succ} \cap \bdg{\ast,}{i}{I}{\succ}$.
\end{itemize}
Observe that
\(
\bdg{}{}{I}{\succ} = \bigsqcup_{d = 0}^{\infty} \bdg{d,}{\ast}{I}{\succ}
= \bigsqcup_{i = 1}^{n} \bdg{\ast,}{i}{I}{\succ}
= \bigsqcup\limits_{\substack{d \in \mathbb{N} \\ i \in \{1, \ldots, n\}}} \bdg{d,}{i}{I}{\succ}
\).

We also introduce notations for the leading monomials of $\bdg{}{}{I}{\succ}$:
\begin{align*}
  \bdgh{d,}{i}{I}{\succ}    & = \{ \lm{\succ}{g} \mid g \in \bdg{d,}{i}{I}{\succ}  \},  \\
  \bdgh{\ast,}{i}{I}{\succ} & = \{ \lm{\succ}{g} \mid g \in \bdg{\ast,}{i}{I}{\succ}  \}, \\
 \bdgh{d,}{\ast}{I}{\succ}  & = \{ \lm{\succ}{g} \mid g \in \bdg{d,}{\ast}{I}{\succ}  \} 
\end{align*}
and finally $\bdgh{}{}{I}{\succ} = \lm{\succ}{\bdg{}{}{I}{\succ}}$. 

When the ideal $I$ and the ordering $\succ$ are clear from the context, we omit them.
Recall that an ideal $I$ is said to be \emph{zero-dimensional} if the dimension of the
quotient $\Kp{n} / I$ as a $\field$-vector space is finite
(see \citep[Chapter~2, \S2, Proposition 4]{cox94}).

\begin{thm}
\label{thm_bijection_monomials}
Set $\delta \geq 0$, $n \geq 1$, $\nbpol \geq 1$ and consider a sequence $\seqpol = (f_1, \ldots , f_{\nbpol})$ 
in $\Kp{n}$ of degree $\delta$ and $I = \langle \seqpol \rangle$.
Let $\succ$ be a monomial ordering on the monomials of $\Kp{n}$. Suppose that $\lm{\succ}{I}$
is stable and that $I$ is zero-dimensional.
Let $\ell$ be in $\{1, \ldots , n\}$
and $m$ be in $\bdgb{}{\ell-1}{I}{\succ}$. 
There exists an integer \(\alpha_{m,\ell}\) such that:
\[ \alpha_{m,\ell} = 
    \min_{k \in \mathbb{N}}\{k \mid m\ x_\ell^k \in \lm{\succ}{I} \} 
\quad \text{ and } \quad 
m\ x^{\alpha_{m, \ell}}\in \bdgh{\ast,}{\ell}{I}{\succ}.\]
Moreover, the application 
\(\psi_\ell : 
m \in \bdgb{}{\ell-1}{I}{\succ} \longmapsto  
m\ x_\ell^{\alpha_{m,\ell}}\in \bdgh{\ast,}{\ell}{I}{\succ}
\)
is a bijection. 
\end{thm}

\begin{proof}
    Take \(\ell\in \{1, \ldots, n\}\) and $m \in \bdgb{}{\ell-1}{I}{\succ}$.
    Since \(I\) is zero-dimensional, 
    there is \(k\) such that $x_\ell^k \in \lm{\succ}{I}$ (see e.g.\ \cite[Prop.
    12.7]{BaPoRo06}). This establishes the existence of
    $\alpha_{m,\ell}$ for any
  $m$ in $\bdgb{}{\ell-1}{I}{\succ}$.
We first prove that $m\ x_\ell^{\alpha_{m,\ell}}$ lies in
  $\bdgh{\ast,}{\ell}{I}{\succ}$.
  By de\-finition, it holds that $m\
  x_\ell^{\alpha_{m,\ell}}$ lies in $\Kp{\ell}$. Also, 
  since \(m\not\in \lm{\succ}{I}\) and 
  $m\ x_\ell^{\alpha_{m,\ell}}\in \lm{\succ}{I}$, we deduce that 
  $\alpha_{m,\ell} \geq 1$ and then $m\
  x_\ell^{\alpha_{m,\ell}}$ does not lie
  in $\Kp{\ell-1}$. It remains to show that $m
  x_\ell^{\alpha_{m,\ell}}$ is a leading
  monomial of an element of the reduced $\succ$-Gröbner basis of $I$.

  Assume by contradiction that this is not the case.
  This implies that there exists a monomial $m_0$ in $\lm{\succ}{I}$ that
  divides $m\ x_\ell^{\alpha_{m,\ell}}$ without being equal to it.  
  We distinguish two cases
  depending on whether $m_0$ and $mx_\ell^{\alpha_{m,\ell}}$ have the same degree in
  $x_\ell$ or not.

  \emph{Case 1.} Suppose that $m_0 = m_1\ x_\ell^k$ with
  \(m_1 \in \Kp{\ell-1}\) and $k
  < \alpha_{m,\ell}$. Hence, $m_1$ divides the monomial $m$, and
  therefore \(m\ x_\ell^k = \frac{m}{m_1} m_0\). 
  We deduce that \(m\ x_\ell^k\in \lm{\succ}{I}\) with \( k <
  \alpha_{m, \ell}\) (since
  \(m_0\in \lm{\succ}{I}\) by assumption). This contradicts
  the minimality of \(\alpha_{m, \ell}\). 

  \emph{Case 2.} Suppose that $m_0 = m_1\
  x_\ell^{\alpha_{m,\ell}}$ with \(m_1 \in
  \Kp{\ell-1}\). Thus, $m_1$ divides $m$ and $m \neq m_1$. Since $m$ and $m_1$ lie
  in $\Kp{\ell-1}$, there exists $j<\ell$ such that $m_1\ x_j$ divides $m$.  
  We established that $\alpha_{m,\ell} \geq 1$. Then,
  $x_\ell$
  divides $m_0$. Since, by assumption, $\lm{\succ}{I}$ is stable, 
  the monomial  $m_0x_j/x_\ell = 
  m_1\ x_j\ x_\ell^{\alpha_{m,\ell}-1}$ lies in 
  $\lm{\succ}{I}$. Since $m_1\ x_j$ divides
  $m$ and $\alpha_{m,\ell}-1 < \alpha_{m,\ell}$, we are bound to
  apply {\it Case 1} to $\left(m_1\ x_j\right
      )x_\ell^{\alpha_{m,\ell}-1}$ and
  deduce that it does not lie in $\lm{\succ}{I}$. This
  yields 
  a contradiction and we conclude that $m_0 \not\in \lm{\succ}{I}$. 

  All in all, we have established that 
  $m\ x_\ell^{\alpha_{m,\ell}}$ lies in $\bdgh{\ast,}{\ell}{I}{\succ}$. 

Now we prove the bijection statement.
  For two distinct monomials $m_1$ and $m_2$ in
  $\bdgb{}{\ell-1}{I}{\succ}$, the exponents of
  $m_1\ x_\ell^{\alpha_{m_1,\ell}}$ and $m_2\
  x_\ell^{\alpha_{m_2,\ell}}$
  on the first \(\ell-1\) variables differ. This suffices to
  prove that \(\psi_\ell\) is injective. 
  It remains to prove that $\psi_\ell$ is a surjective function. Let
  $\widetilde{m} \in \bdgh{\ast,}{\ell}{I}{\succ}$, and write
  $\widetilde{m} =
  m\ x_\ell^{\alpha}$ with \(m \in \Kp{\ell-1}\). By definition of
  \(\bdgh{\ast,}{\ell}{I}{\succ}\), it holds that \(\alpha >
  0\). 
  Since $\widetilde{m}$ is a leading monomial of an element of the reduced Gröbner
  basis of $I$, for all $k$ in $\{0, \ldots, \alpha-1\}$ the monomial
  $m\ x_\ell^k$ does not lie in  $\lm{\succ}{I}$ (this would
  contradict the minimality of reduced Gr\"obner bases). 
  This implies both $m \in
  \bdgb{}{\ell-1}{I}{\succ}$ and $\psi(m) = \widetilde{m}$. Thus,
  $\psi_\ell$ is surjective, and is a bijection.
\end{proof}

\subsection{On the module of syzygies}
\label{subsection_Onthesyzygiesmodule}

We use the standard {\em graded reverse lexicographical} and 
{\em graded lexicographic}
monomial 
orderings; {\em grevlex} and {\em grlex} for short, see \cite[Chapter 2, Section 2, Definition 6]{cox94} for
more details. They are orderings on the monomials $\mon{n}{}$ of $\Kp{n}$
denoted by $\grs$ and \(\gs\). Here we
take the conventions $x_1 \grs \cdots
\grs x_n$ and $x_1 \gp \cdots
\gp x_n$. For $(m_1, \ldots , m_\ell)$ in $\mon{n}{}$,
their least common multiple is denoted by $\operatorname{lcm}(m_1, \ldots ,
m_\ell)$.

Let $G = \{g_1, \ldots, g_r\}$ be a set of polynomials in $\Kp{n}$ such that
their leading monomials for the \emph{grevlex} ordering are all distinct, we say that
$G$ is ordered with respect to $\succ_{\mathrm{ap}}$ if:
\[
  \forall (i,j) \in \{1, \ldots, r\}^2, \quad i < j \Longleftrightarrow \lm{\grs}{g_i} \gp \lm{\grs}{g_j}.
\]
Let us recall the definition of the basis $L^{\star}(G)$ of the module of
syzygies of leading monomials of elements of $G$ introduced in 
\cite[Proposition 3.5]{gebauer1988installation}. For $i,j,k \in \{ 1,
\ldots , r \}$, write
\begin{align*}
  \lcmt{G}{i}  = \lm{\grs}{g_i}, \quad 
  \lcmt{G}{i,j}  =
  \operatorname{lcm}(\lcmt{G}{i},\lcmt{G}{j}), \quad \text{
  and } 
  \lcmt{G}{i,j,k}  = \operatorname{lcm}(\lcmt{G}{i},\lcmt{G}{j},\lcmt{G}{k}).
\end{align*}
The $\Kp{n}$-module of the syzygies of leading monomials of elements of $G$ is
\[ S^{(1)}(G) = \left\{ (p_1, \ldots , p_r)  \,\middle|\, p_i \in \Kp{n} \text{ and
    }\sum_{i = 1}^r p_i \lcmt{G}{i} = 0 \right\}.
\]
For all $k$ in $\{1, \ldots , r \}$, we set $\bm{e}_k = (0, \ldots , 1, \ldots , 0)$ 
in $\K^{r}$ the
\(k\)-th coordinate vector.
By \cite[Section 3.2]{gebauer1988installation}, 
the module $S^{(1)}(G)$ has a Taylor basis 
$\taylorbasis^{(1)}(G) = \{ \syzy{i}{j} \mid 1 \leq i < j \leq r \} $ with 
\[ \syzy{i}{j} = \frac{\lcmt{G}{i,j}}{\lcmt{G}{i}}\bm{e}_i - 
\frac{\lcmt{G}{i,j}}{\lcmt{G}{j}}\bm{e}_j. \]

We denote by $\succ_{\mathrm{syz}}$ the following ordering on
$\taylorbasis^{(1)}(G)$: for $i<j$ and $k<\ell$,
\[
\begin{array}{c}
  \syzy{i}{j} \prec_{\mathrm{syz}} \syzy{k}{\ell}  \\
  \Updownarrow \\
  \lcmt{G}{i,j} \gp \lcmt{G}{k,\ell} \text{ or } \left
  (\lcmt{G}{i,j} = \lcmt{G}{k, \ell} \text{ and } 
  \left (j < \ell \text{ or } \left (j = \ell
\text{ and } i<k \right ) \right )\right ).  
\end{array}
\]

We define $\mathcal{B}(G)$  as the set of $\syzy{i}{j}$
(with \( i < j\)) for which there exists
\(k\in \{1, \ldots, r\} \setminus \{i, j\}  \) such that 
\begin{description}
    \item[\((\mathsf{T}_1)\)\hypertarget{T1}] $\lcmt{G}{i, j, k} = \lcmt{G}{i,j}$ and
    \item[\((\mathsf{T}_2 ) \hypertarget{T2}\)] 
 $\syzy{i}{j} =
\max_{\succ_{\mathrm{syz}}}\{\syzy{i}{j},\syzy{\alpha}{\beta},\syzy{\gamma}{\delta}
\}$ with \(\alpha = \min(i, k), \beta = \max(i, k), \gamma =
\min\left(j, k \right) \) and \(\delta=\max\left( j, k \right) \).
\end{description}

By \cite[Proposition 3.5]{gebauer1988installation}, the
module of syzygies $S^{(1)}(G)$ is generated by
\[ \taylorbasis^{\star}(G) = \{ \syzy{i}{j} \mid 1 \leq i < j \leq r, \syzy{i}{j} \notin \mathcal{B}(G) \}. \]

For any nonzero
  polynomial \(f\), we define its variable index \(\ind{f}\) as the unique
  integer $i$ in $\{1, \ldots , n\}$ such that \(\lm{\grs}{f}\)
lies in \(\Kp{i} \setminus \Kp{i-1}\).

\begin{definition}
\label{def_type1_pairs}
Let $(t,f) \in \mon{}{n} \times \Kp{n}$. 
Then, the tuple $(t,f)$ is said to be
\emph{\type{}} if $t = x_j$
with $j<\ind{f}$.
\end{definition}

We denote $\zarifilling{\delta}{n}{\nbpol}$ the set of sequences
$\seqpol = (f_1, \ldots , f_{\nbpol})$ of homogeneous polynomials of
degree $\delta$ in $\Kp{n}$ that satisfy \cite[Conjecture~4.1]{moreno2003degrevlex}.
 We can now state
\cref{thm_type1_good_pairs}.

\begin{thm}
\label{thm_type1_good_pairs}
Let $\seqpol$ in
$\zarifilling{\delta}{n}{\nbpol}$, $\delta \geq 0$ and
$n,\nbpol $ being nonzero integers and $I = \langle \seqpol
\rangle$ which is assumed to be zero-dimensional. Let \(G =
(g_1, \ldots, g_r)\) be the
reduced $\grs$-Gröbner basis of $I$, ordered with respect to
$\succ_{\mathrm{ap}}$.
Then, for any $1 \leq i < j \leq r$ such that
$\syzy{i}{j}$ lies in $\taylorbasis^{\star}(G)$, the tuple 
$\left( \frac{\lcmt{G}{i,j}}{\lcmt{G}{j}}, g_j \right)$ is
\type. 
\end{thm}

We start with four auxiliary lemmata needed for the proof of
\cref{thm_type1_good_pairs}. We reuse the notations
of its statement until the end of this section. 

The following lemma makes explicit that the monomial staircase closes
 dimension by dimension, when its generators are considered by increasing degree.

\begin{lemma}
\label{prop_starecaise_dimension_onebyone}
Suppose that $\seqpol$ lies in $\zarifilling{\delta}{n}{\nbpol}$.
Let \(i\) and  \(j\) in  \(\{1, \ldots, n\} \) with \(i < j\). Let \(d\)
and  \(e\) in  \(\mathbb{N}\) with \(d < e\). 
It holds that, if \(\bdg{d,}{j}{I}{\grs}\) is nonempty, then 
\(\bdg{e,}{i}{I}{\grs}\) is empty.
\end{lemma}

\begin{proof}
Let us fix $j$ in $\{1, \ldots, n\}$ and $d$ in
$\mathbb{N}$, and assume that $\bdg{d,}{j}{I}{\grs}$ is
nonempty. Let $m$ be the leading monomial with respect to
 $\grs$ of an arbitrary
element in $\bdg{d,}{j}{I}{\grs}$. By construction, $m$
lies in $\mon{d,}{j} \setminus \mon{d,}{j-1}$. By definition
of the {\em grevlex} ordering, for all $\widetilde{m}$ in
$\mon{d,}{j-1}$, 
we have $m \grp \widetilde{m}$.

Let \(c\) be the cardinality of \(\sharp
\bdg{d,}{*}{I}{\grs}\).
Since $\seqpol$ lies in $\zarifilling{\delta}{n}{\nbpol}$,
one can apply \cite[Conjecture~4.1]{moreno2003degrevlex} and
deduce the \(c\) largest monomials in 
 \[
     \mathscr{E} =\left( \bdgb{}{}{I}{\grs}\cap \mon{d,}{*}
     \right)  \bigsqcup \bdgh{d,}{*}{I}{\grs}
\]
are the monomials of \(\bdg{d,}{*}{I}{\grs}\). 
Let \(\widetilde{m}\) be in $\mon{d,}{j-1}$.
Note that if \(\widetilde{m}\) does not lie in
\(\mathscr{E}\), then it lies in  \(\lm{\grs}{I}\), by
definition of  \(\bdgb{}{}{I}{\grs}\). 
If \(\widetilde{m}\) lies in \(\mathscr{E}\), then it lies
in \(\bdg{d,}{*}{I}{\grs}\) since it is larger than  \(m\)
for  \(\grs\). Since, by definition,
\(\bdg{d,}{*}{I}{\grs}\) is contained in \(\lm{\grs}{I}\),
we again deduce that  \(\widetilde{m}\) lies in \(\lm{\grs}{I}\).

Let $i < j$ and suppose by contradiction that
$\bdg{e,}{i}{I}{\grs}$ is nonempty for some
$e > d$. Let $g$ be in
$\bdg{e,}{i}{I}{\grs}$, so that
$\lm{\grs}{g} \in \mon{e,}{j-1}$. Since, by
assumption, $e > d$, there
exists a monomial $\widetilde{m}$ in $\mon{d,}{j-1}$
such that $\widetilde{m}$ divides $\lm{\grs}{g}$
and $\widetilde{m} \ne
\lm{\grs}{g}$. By the previous paragraph, 
\(\widetilde{m}\) lies in $\lm{\grs}{I}$. 
This is in contradiction with the fact that $g$ lies in the
reduced $\grs$-Gröbner basis of $I$.
\end{proof}

\begin{lemma}
\label{lemma_orderonfi_inGj}
Let $\seqpol$ in
$\zarifilling{\delta}{n}{\nbpol}$, $\delta \geq 0$ and
$n,\nbpol $ being nonzero integers and $I = \langle \seqpol
\rangle$. Let \(G =
(g_1, \ldots, g_r)\) be the
reduced $\grs$-Gröbner basis of $I$, ordered with respect to
$\succ_{\mathrm{ap}}$.
For any integers $i,j$ such that 
$1 \leq i < j \leq \nbgb$, then one has $\ind{g_i} \leq \ind{g_j}$.
\end{lemma}

\begin{proof}
  Write $d = \deg(\lcmt{G}{i})$ and $e =
  \deg(\lcmt{G}{j})$. By assumption we have $\lcmt{G}{i} \gp
  \lcmt{G}{j}$, and thus $d \leq e$.

  If \(e = d\), by definition of the {\em
  grlex} order, it holds that  \(\ind{g_i} \leq \ind{g_j}\).
  If \(d < e\), it holds that both \(G_{d, \ind{g_i}}\) and
      \(G_{e, \ind{g_j}}\)
  are nonempty, which, by
  \cref{prop_starecaise_dimension_onebyone}, implies that
  \(\ind{g_i} \le \ind{g_j}\). 
\end{proof}

\begin{lemma}
\label{lemma_critere2}
  Let $k<i<j$ be integers in $\{1, \ldots , \nbgb \}$. If $\lcmt{G}{k}$ divides
  $\lcmt{G}{i,j}$, then $\syzy{i}{j} \notin \taylorbasis^{\star}(G)$.
\end{lemma}

\begin{proof}
  Suppose that $\lcmt{G}{k}$ divides $\lcmt{G}{i,j}$. This implies that
  $\lcmt{G}{i,j} = \lcmt{G}{k,i,j}$ and that $\lcmt{G}{k,i}$ divides
  $\lcmt{G}{i,j}$. Thus, $\deg(\lcmt{G}{k,i}) < \deg(\lcmt{G}{i,j})$ or
  $\lcmt{G}{k,i} = \lcmt{G}{i,j}$. First, if the inequality $\deg(\lcmt{G}{k,i}) <
  \deg(\lcmt{G}{i,j})$ holds, then $\syzy{k}{i}
  \prec_{\mathrm{syz}}
  \syzy{i}{j}$. If $\lcmt{G}{k,i} = \lcmt{G}{i,j}$, using the
  fact that $i<j$, we also have $\syzy{k}{i}
  \prec_{\mathrm{syz}} \syzy{i}{j}$.
  Similar arguments on $\syzy{k}{j}$ and $\syzy{i}{j}$, with the fact that
  \(k < i\), yield $\syzy{i}{j} =
  \max_{\succ_{\mathrm{syz}}}(\syzy{k}{i},\syzy{k}{j},\syzy{i}{j})$. Thus, by
  definition of \(\mathcal{B}(G)\) and $\taylorbasis^{\star}(G)$,
  the syzygy $\syzy{i}{j}$ does not lie in $\taylorbasis^{\star}(G)$.
\end{proof}

The degree of a polynomial $f$ in the variable $x_i$ is denoted by $\deg_{x_i}(f)$.

\begin{lemma}
\label{lemma_degreeinxq_order_fillingup}
Let $\seqpol$ in
$\zarifilling{\delta}{n}{\nbpol}$, $\delta \geq 0$ and
$n,\nbpol $ being nonzero integers and $I = \langle \seqpol
\rangle$. Let \(G =
(g_1, \ldots, g_r)\) be the
reduced $\grs$-Gröbner basis of $I$, ordered with respect to
$\succ_{\mathrm{ap}}$.
Consider $1 \leq i<j \leq \nbgb$ and let \(a\) denote
$\ind{g_j}$. Then the following inequality holds:
\[ \deg_{x_a}(\lcmt{G}{i}) \leq \deg_{x_a}(\lcmt{G}{j}). \]
\end{lemma}

\begin{proof}
  First, \cref{lemma_orderonfi_inGj} establishes that $\ind{g_i} \leq
  \ind{g_j} = a$. If $\ind{g_i} < a$, then by definition
  $\deg_{x,_a}(\lcmt{G}{j}) >0$ and $\deg_{x_a}(\lcmt{G}{i}) = 0$, hence
  $\deg_{x_a}(\lcmt{G}{i}) < \deg_{x_a}(\lcmt{G}{j})$. In the rest of the
  proof, we consider the case $\ind{g_i} = a$. As $\bdg{}{}{I}{\grs}$ is
  ordered with $\succ_{\mathrm{ap}}$, one has $\lcmt{G}{i} \gp \lcmt{G}{j}$
  which implies that $\deg(\lcmt{G}{i}) \leq \deg(\lcmt{G}{j})$.

  Suppose first that $\deg(\lcmt{G}{i}) = \deg(\lcmt{G}{j})$. Since     
  $\ind{g_i} = \ind{g_j} = a$, the monomials $\lcmt{G}{i}$ and $\lcmt{G}{j}$
  both lie in  $\Kp{a} \setminus \Kp{a-1}$; combined with
  $\lcmt{G}{i} \gp \lcmt{G}{j}$, this gives
  $\deg_{x_a}(\lcmt{G}{i}) \leq \deg_{x_a}(\lcmt{G}{j})$.

  Now suppose that $\deg(\lcmt{G}{i}) < \deg(\lcmt{G}{j})$. Then there exists a
  monomial $m$ in $\Kp{a}$ that divides $\lcmt{G}{j}$ such that  $\deg(m) =
  \deg(\lcmt{G}{i})$. Assume $\deg_{x_a}(\lcmt{G}{i}) > \deg_{x_a}(\lcmt{G}{j})$
  by contradiction. It then holds that  $\deg_{x_a}(\lcmt{G}{i}) >
  \deg_{x_a}(\lcmt{G}{j}) \ge \deg_{x_a}(m)$, hence $m \grs \lcmt{G}{i}$.
  From the fact that \(\lcmt{G}{i} \in \lm{\grs}{\bdg{}{}{I}{\grs}}\)
  and \(\seqpol \in \zarifilling{\delta}{n}{\nbpol}\),
  we deduce that $m$ lies in $\lm{\grs}{I}$. Since $m$ divides $\lcmt{G}{j}$
  and is different from it, this
  contradicts the minimality of the $\grs$-Gröbner
  basis $\bdg{}{}{I}{\grs}$.
\end{proof}

\begin{remark}
  \label{remark_fillingup_implies_stable}%
  The assumption that $\seqpol \in \zarifilling{\delta}{n}{\nbpol}$ implies the
  stability of $\lm{\grs}{\langle \seqpol \rangle}$. Indeed, suppose that $I$
  is the ideal generated by a sequence $\seqpol$ which lies in
  $\zarifilling{\delta}{n}{\nbpol}$. Using  \cite[Lemma 4.2.3]{HerzogMono} we
  just need to verify the stability property on monomials of $\bdgh{}{}{I}{\grs}$. Let
  $m$ be in  $\bdgh{}{}{I}{\grs}$ such that $x_i$ divides $m$. Then the 
  monomial $\widetilde{m} = \frac{x_jm}{x_i}$ with $j<i$ is greater than $m$
  with respect to the {\em grevlex} ordering. Since $\seqpol$ lies in
  $\zarifilling{\delta}{n}{\nbpol}$, we have that $\widetilde{m}$ lies in
  $\lm{\grs}{I}$.  Then $\lm{\grs}{I}$ is stable.
\end{remark}

\subsection{Proof of Theorem~\ref{thm_type1_good_pairs}}
\label{subsection_Onthesyzygiesmodule_proof}

For a monomial $m = x_1^{\epsilon_1} \cdots x_n^{\epsilon_n}$ and an integer $i$ in 
$\{1, \ldots , n\}$, let \(\proj{m}{i}\) denote 
$x_1^{\epsilon_1} \cdots x_i^{\epsilon_i}$. 

Let us take $\syzy{i}{j}$ with $1 \leq i < j \leq \nbgb$.
For the remainder of the proof, we set the following notation:
\[\ind{g_i} = a, \quad  \lcmt{G}{i} = \prod_{t =
1}^{a}x_t^{\alpha_{t}} \quad \text{and} \quad
\ind{g_j} = b, \quad 
\lcmt{G}{j} = \prod_{t =
1}^{b}x_t^{\beta_{t}}.\]
Since the sequence of polynomials $\seqpol$ lies in $\zarifilling{\delta}{n}{\nbpol}$, 
one can apply \cref{lemma_orderonfi_inGj} to deduce that 
$a \leq b$. 

\begin{lemma}
    Under the above notation and assumptions, it holds that 
$\frac{\lcmt{G}{i,j}}{\lcmt{G}{j}}\in \Kp{a}$.
\end{lemma}

\begin{proof}
If $a = b$, it is obvious that 
$\frac{\lcmt{G}{i,j}}{\lcmt{G}{j}} \in \Kp{a}$.
If $a < b$, then 
\[ \frac{\lcmt{G}{i,j}}{\lcmt{G}{j}}  = \left( \prod_{t =
    1}^{a}x_t^{\max(\alpha_{t},\beta_{t})-\beta_{t}} \right) 
\left( \prod_{t = a + 1}^{b} 
x_t^{\max(\alpha_{t},\beta_{t})-\beta_{t}}
\right). \] 

Since, by definition, for all $t >
a$,  $\alpha_{t}\ = 0$, we deduce that  
\( \frac{\lcmt{G}{i,j}}{\lcmt{G}{j}}  = \left( \prod_{t =
    1}^{a}x_t^{\max(\alpha_{t},\beta_{t})-\beta_{t}}
\right) \)
which lies in $\Kp{a}$.
\end{proof}

Let us prove that either
$\left( \frac{\lcmt{G}{i,j}}{\lcmt{G}{j}}, g_j \right)$ is
\type, or $\syzy{i}{j}$ does not lie in
$\taylorbasis^{\star}(G)$. We perform several case analyses,
starting with whether $a$ and $b$ are equal or not (recall
that we have established already \(a \le  b\)).

\paragraph*{Case 1: \(a < b\)}
Since $g_i$ and $g_j$ are elements of a minimal $\grs$-Gröbner basis,
$\lcmt{G}{i}$ does not divide $\lcmt{G}{j}$. We deduce that there exists 
$c$ in $\{1, \ldots , a\}$ such that $\alpha_{c} >
\beta_{c}$.
This implies that $x_{c}$ divides $\frac{\lcmt{G}{i,j}}{\lcmt{G}{j}}$.
Consider: 
\[\mu = \proj{x_{c}\lcmt{G}{j}}{a}  = 
x_{c}\prod_{t = 1}^{a}x_t^{\beta_{t}}.\] 
Obviously, the monomial $\mu$ lies in $\Kp{a}$. 
We distinguish cases depending on $\mu$ being in
$\bdgb{}{a}{I}{\grs}$ or in 
$\lm{\grs}{I}$.

\subparagraph*{Case 1.1: $\mu$ lies in
$\bdgb{}{a}{I}{\grs}$.}
Since $\seqpol$ is in $\zarifilling{\delta}{n}{\nbpol}$,
we deduce by \cref{remark_fillingup_implies_stable}
that $\lm{\grs}{I}$ is stable. Moreover, $I$ is zero-dimensional by assumption.

Thus, we are in position to apply
\cref{thm_bijection_monomials} with \(\ell  = a + 1\) and
denote by \(\nu_{a+1}\) the image of  \(\mu\) by the map
\[
\psi_{a+1} : m\in \bdgb{}{a}{I}{\grs}\mapsto m\
x_{a+1}^{\sigma_{m, a+1}}\in \bdgh{\ast,}{a + 1}{I}{\grs} \]
with \( \sigma_{m, a + 1} = \min_{k \in \mathbb{N}}\{k \mid
mx_{a+1}^k \in \lm{\grs}{I} \} \).

We use repeatedly this construction to define a sequence of
monomials, \(\mu_{a}, \mu_{a+1}, \ldots,\) starting with
\(\mu_a = \mu\). The next
elements of this sequence are recursively defined as
follows. For \(\ell \in \{a + 1, \ldots, b\} \), we have that 
 \(\deg_{x_{\ell}}(x_c\ \lcmt{G}{j} ) = \beta_\ell\).
Also, for elements \(\mu_\ell\) of that sequence, lying in 
\(\bdgb{}{\ell}{I}{\grs}\), we denote by  \(\nu_{\ell + 1}\) the
image by the map  \(\psi_{\ell + 1}\) of  \(\mu_\ell\) (where
\(\psi_{\ell + 1}\) is defined by \cref{thm_bijection_monomials}).
Given \(\mu_\ell\), we define \(\mu_{\ell + 1}\) as : 
\begin{itemize}
    \item[{\em (i)}] if \(\beta_{\ell + 1} < \deg_{x_{\ell +
        1}}(\nu_{\ell + 1})\) then
        \(\mu_{\ell + 1} = \mu_\ell\ x_{\ell +
        1}^{\beta_{\ell + 1}}\);
    \item[{\em (ii)}] otherwise, \(\mu_{\ell + 1} =
        \nu_{\ell + 1} \) which does not
            lie in \(\bdgb{}{\ell + 1}{I}{\grs}\), and the
        sequence is stationary at \(\mu_{\ell + 1}\). 
\end{itemize}

\begin{lemma}\label{lemma:stationarity}
   We reuse the above notations and assumptions.  
   There exists \(\mathfrak{e}\in \{a + 1, \ldots, b\}\)
   such that the sequence is stationary at
   \(\mu_{\mathfrak{e}}\) and for all \(a \le \ell \le
   \mathfrak{e} - 1\), it holds
   that  \(\mu_{\ell} =
   \proj{x_{c}\lcmt{G}{j}}{\ell} = x_c \left(
   \prod_{t=1}^{\ell} x_{t}^{\beta_{t}}\right) \). 
   Moreover, for \( a\le \ell \le \mathfrak{e} - 1 \), the 
   monomial $\mu_{\ell}$ lies in \(\bdgb{}{\ell}{I}{\grs}\).
\end{lemma}
\begin{proof}

By definition of the sequence $\mu_a, \mu_{a+1}, \ldots$, we have that if,
for a $\ell > a$, the monomial $\mu_\ell$ is defined by $(ii)$, then the same
holds for all subsequent monomials of the sequence.
Let us call $\mathfrak{e}$ the index of the first monomial of the sequence
defined by $(ii)$, with $\mathfrak{e} = + \infty$ if it is never the case.

It is obvious that for \( a\le \ell \le \mathfrak{e} - 1 \), the 
monomial $\mu_{\ell}$ lies in \(\bdgb{}{\ell}{I}{\grs}\).

Now, we prove by induction that for \( a\le \ell < \mathfrak{e} \), $\mu_{\ell} =
\proj{x_{c}\lcmt{G}{j}}{\ell}$. The statement holds by
construction for \(\ell = a\). Assume now that $\mu_{\ell-1} =
\proj{x_{c}\lcmt{G}{j}}{\ell-1}$ for $\ell$ such that \( a\le \ell < \mathfrak{e} \).
 Since we have \(\ell < \mathfrak{e}\),
 it holds by construction that \(\mu_\ell =
\mu_{\ell - 1}x_\ell^{\beta_\ell}\) with \(\beta_\ell =
\deg_{x_\ell}\left( x_c\lcmt{G}{j} \right) \). Since we have
$\mu_{\ell-1} =
\proj{x_{c}\lcmt{G}{j}}{\ell-1}$ and $c \leq a < \ell$, our conclusion follows by
the definition of \(\lcmt{G}{j}\).

It remains to prove that $\mathfrak{e}$ lies in $\{a+1, \ldots, b\}$.
Suppose, for the sake of contradiction, that this is not the case.
Then $\mathfrak{e} > b$.
This implies that
\(\mu_b = \proj{x_c \lcmt{G}{j}}{b} = x_c \lcmt{G}{j}
\)
lies in $\bdgb{}{b}{I}{\grs}$.
This contradicts the fact that $\lcmt{G}{j}$ lies in
$\bdgh{}{}{I}{\grs}$.
Hence, $\mathfrak{e}$ lies in $\{a+1, \ldots, b\}$.
\end{proof}

Using \cref{thm_bijection_monomials}, the monomial
$\mu_{\mathfrak{e}}$ lies in 
$\bdgh{\ast,}{\mathfrak{e}}{I}{\grs}$. Then
there exists an integer $k$ in $\{1, \ldots , \nbgb\}$
such that $\lcmt{G}{k} = \mu_{\mathfrak{e}}$.
We prove that $\syzy{i}{j}$ is not in
$\taylorbasis^{\star}(G)$, hence proving that both
\hyperlink{T1}{\(\mathsf{T}_1\)} and
\hyperlink{T2}{\(\mathsf{T}_2\)} hold, using the index $k$.

By \cref{lemma:stationarity}, \(\mu_{\mathfrak{e} - 1} =
\pi_{\mathfrak{e} - 1}\left(x_c \lcmt{G}{j}  \right) = x_c
\prod_{t = 1}^{\mathfrak{e} - 1} x_t^{\beta_t} \) and, by
definition, we have that 
\[ \lcmt{G}{k} = \mu_{\mathfrak{e}} = \psi_{\mathfrak{e}}\left(
\mu_{\mathfrak{e} - 1} \right) = \mu_{\mathfrak{e} -
1}x_\mathfrak{e}^{\sigma} = x_c \left(
\prod_{t=1}^{\mathfrak{e}-1}
x_{t}^{\beta_{t}}\right)x_{\mathfrak{e}}^{\sigma} \]
 where
\(\sigma\) is the smallest
integer such that  \(\mu_{\mathfrak{e}}\in \lm{\grs}{I}\).

Since $\mu_{\mathfrak{e}}$ is a monomial of the sequence that is defined by
$(ii)$, this implies \(\sigma \leq \beta_{\mathfrak{e}}\). 
We can then deduce from the above expression that $\lcmt{G}{k,j} = x_{c}\lcmt{G}{j}$. Moreover, since, by
definition of \(c\), $\lcmt{G}{k}$ divides
$\lcmt{G}{i,j}$, we deduce that $\lcmt{G}{i,j} =
\lcmt{G}{i,k,j}$. 

By \cref{lemma:stationarity}, $a = \ind{g_i} <
\mathfrak{e} = \ind{g_k}$. 
This implies that $i<k$, in particular $i\neq k$. 
Also, since $\deg_{x_{c}}(\lcmt{G}{k}) = \deg_{x_{c}}(\lcmt{G}{j})+1$, 
$\lcmt{G}{k} \ne \lcmt{G}{j}$ and we deduce that $k \ne j$.
This ends to prove \hyperlink{T1}{\(\mathsf{T}_1\)}. 

\medskip

Now, we prove that $\syzy{i}{j} \succ_{\mathrm{syz}} \syzy{i}{k}$ and 
$\syzy{i}{j} \succ_{\mathrm{syz}} \syzy{k}{j}$ which is
\hyperlink{T2}{\(\mathsf{T}_2\)}. This is
done by establishing that 
\(\deg\left(\lcmt{G}{k,j} \right) < \deg\left(
    \lcmt{G}{i,j}
\right) \)
and \(\deg\left(\lcmt{G}{i,k} \right) < \deg\left(
    \lcmt{G}{i,j}
\right) \) 
which is sufficient by definition of
\(\succ_{\mathrm{syz}}\). 

\smallskip
We proved in the above paragraph that \(\lcmt{G}{k,j} =
x_c\lcmt{G}{j}\). In fact, by
definition, we have that \(x_c \lcmt{G}{j} = x_c
\prod_{t=1}^{b}x_t^{\beta_t} = \mu_a
\prod_{t=a+1}^{b}x_t^{\beta_t}\).
 Also, 
recall that, by definition, \(x_c\) divides
\(\frac{\lcmt{G}{i,j}}{\lcmt{G}{j}}\). We deduce that 
\(\lcmt{G}{i,j}= 
\operatorname{lcm}(\lcmt{G}{i},x_{c}\lcmt{G}{j})\).
It then holds that \(\lcmt{G}{i,j}= 
\left(
\prod_{t=a+1}^{b}x_t^{\beta_t} \right)
\operatorname{lcm}(\lcmt{G}{i},\mu_a)\).
 Since, by
assumption, $\mu_{a}
\in \bdgb{}{a}{I}{\grs}$ and $\lcmt{G}{i}
\in \lm{\grs}{I}$, there exists a monomial \(\nu\) of
positive degree in  \(\Kp{a}\) such that
\(\operatorname{lcm}\left(\lcmt{G}{i}, \mu_a  \right) = \mu_a\nu \). 
We deduce that \(\lcmt{G}{i,j} = 
\left(
\prod_{t=a+1}^{b}x_t^{\beta_t} \right)\mu_a \nu = x_c
\lcmt{G}{j}\ \nu 
\) and then \(\deg\left(\lcmt{G}{k,j} \right) <
\deg\left(\lcmt{G}{i,j}  \right)  \)
as requested. 

Recall that \(\lcmt{G}{k} = \mu_{\mathfrak{e}}  = x_c \left(
\prod_{t=1}^{\mathfrak{e}-1} x_t^{\beta_t}\right)
x_{\mathfrak{e}}^\sigma = \mu_a \left(
    \prod_{t=a+1}^{\mathfrak{e}-1}x_t^{\beta_t}
\right)x_{\mathfrak{e}}^\sigma \), with \(\mathfrak{e} > a\) and 
\(\sigma\) defined above, and
that  \(\lcmt{G}{i} = \prod_{t=1}^a x_t^{\alpha_t}\). Hence, 
\(\lcmt{G}{i,k} = \operatorname{lcm}\left( \lcmt{G}{i},\lcmt{G}{k} \right) = 
\operatorname{lcm}\left(\lcmt{G}{i}, \mu_a  \right) \left(
\prod_{t=a+1}^{\mathfrak{e}-1}
x_t^{\beta_t}\right)x_{\mathfrak{e}}^{\sigma} \). By
definition of \(\nu\) above, we deduce that 
\(\operatorname{lcm}\left( \lcmt{G}{i},\lcmt{G}{k} \right) =\nu
\mu_a \left (\prod_{t=a+1}^{\mathfrak{e}-1}
x_t^{\beta_t}\right)x_{\mathfrak{e}}^{\sigma}=\nu\lcmt{G}{k}
\). We proved in the above paragraph that 
\(\lcmt{G}{i,j} = x_c\lcmt{G}{j}\nu\). Hence, it remains to prove the inequality \(\deg\left(\lcmt{G}{k}  \right) < \deg\left( x_c \lcmt{G}{j} \right)  \) to
establish that \(\deg\left(\lcmt{G}{i,k}  \right) < \deg\left(
\lcmt{G}{i,j} \right)  \). 
By construction, \(\lcmt{G}{k}\) divides \(x_c \lcmt{G}{j}\).
Since \(\lcmt{G}{k}\) and \(\lcmt{G}{j}\) are leading monomials of a 
minimal basis, \(\lcmt{G}{k}\) is not equal to \(x_c \lcmt{G}{j}\).
This implies the degree
inequalities and ends the proof. 

\subparagraph*{Case 1.2: $\mu$ does not lie in
$\bdgb{}{a}{I}{\grs}$.}
    This implies that there
exists an integer $k$ in $\{1, \ldots , \nbgb\}$ such that 
$\lcmt{G}{k}$ divides $\mu$.
Moreover, as $\mu$ belongs to $\Kp{a}$. Then the monomial $\lcmt{G}{k}$
 lies in $\Kp{a}$ and consequently
$\ind{g_k} \leq a$. Let \(e\) denote  \(\ind{g_k}\). 

We first prove that when $k\in \{i,j\}$, 
$\left ( \frac{\lcmt{G}{i,j}}{\lcmt{G}{j}}, g_j\right )$ is
\type.
Since $a < b$ by assumption
and $e \leq a$, we have that $k \ne j$ and $k = i$. 
Since $k = i$, then $\lcmt{G}{i}$ divides $\mu$.
By definition, \(\mu\) divides $x_{c}\lcmt{G}{j}$. Then, 
$\lcmt{G}{i}$ divides $x_{c}\lcmt{G}{j}$. We deduce 
that $\lcmt{G}{i,j}$ divides $x_{c}\lcmt{G}{j}$.
Since $\lcmt{G}{j}$ divides $\lcmt{G}{i,j}$ and
$\lcmt{G}{j}\neq \lcmt{G}{i,j}$ (because \(G\)
 is a minimal $\grs$-Gröbner basis),
 $\lcmt{G}{i,j} = x_{c}\lcmt{G}{j}$.
Since $1 \le  c \le  a$ and $ a < b$, the tuple  
$\left ( \frac{\lcmt{G}{i,j}}{\lcmt{G}{j}}, g_j\right )$ is
\type.

\smallskip
Now, we prove that when $k$ can not be chosen in $\{i,j\}$, properties \hyperlink{T1}{\(\mathsf{T}_1\)} and
\hyperlink{T2}{\(\mathsf{T}_2\)} hold, i.e.\ \(S_{i,j}\notin
\taylorbasis^{\star}(G)\). We start with
\hyperlink{T2}{\(\mathsf{T}_2\)}. 

We have 
 \( e \leq a < b \).
Since $\seqpol$ lies in $\zarifilling{\delta}{n}{\nbpol}$, 
\cref{lemma_orderonfi_inGj} implies that $k<j$.
Since $\lcmt{G}{k}$ divides $\mu$ which divides
\(x_c\lcmt{G}{j}\), $\lcmt{G}{k}$ divides
$x_{c}\lcmt{G}{j}$. We deduce that $\lcmt{G}{k,j}$
divides $x_{c}\lcmt{G}{j}$. 
Note that $\lcmt{G}{i,j}=
\operatorname{lcm}(\lcmt{G}{j},\lcmt{G}{j}) = \left(
    \prod_{t = a+1}^{b}x_{t}^{\beta_{t}} \right)
\operatorname{lcm}\left( \lcmt{G}{i},  \prod_{t =
1}^{a}x_{t}^{\beta_{t}} \right)$.
By definition of $c$, we have that
 \(\operatorname{lcm}\left( \lcmt{G}{i},  \prod_{t = 1}^{a}x_{t}^{\beta_{t}} \right) = \operatorname{lcm}\left(
\lcmt{G}{i},  x_{c}\prod_{t = 1}^{a}x_{t}^{\beta_{t}} \right) 
= \operatorname{lcm}\left( \lcmt{G}{i},  \mu \right)\).
By assumption, \(\lcmt{G}{i}\) does not divide
 \(\mu\) and \(\mu\neq\lcmt{G}{i}\).
  Also, since \(\mu\in \lm{\grs}{I}\) while
 \(\lcmt{G}{i}\in \lm{\grs}{G}\) and \(\mu\neq \lcmt{G}{i}\),
 we deduce that  \(\mu\) does not divide  \(\lcmt{G}{i}\). 
Thus, there exists a monomial $\nu$ such that the expression 
\( \operatorname{lcm}\left( \lcmt{G}{i},  \prod_{t =
    1}^{a}x_{t}^{\beta_{t}} \right) = \nu x_{c}\prod_{t =
1}^{a}x_{t}^{\beta_{t}} \)
with $\deg(\nu) \geq 1$ holds.
This implies that $\lcmt{G}{i,j} =
\nu x_{c}\lcmt{G}{j}$.
Since $\lcmt{G}{k,j}$
divides $x_{c}\lcmt{G}{j}$, we deduce that the inequality $\deg(\lcmt{G}{k,j}) < \deg(\lcmt{G}{i,j})$ holds.
It follows that $\syzy{k}{j} \prec_{\mathrm{syz}} \syzy{i}{j}$. 

Since $\lcmt{G}{k}$ divides $\mu$, it holds that 
$\lcmt{G}{i,k}$ divides $\operatorname{lcm}(\lcmt{G}{i},\mu)$.
Moreover, we already proved that 
\( \lcmt{G}{i,j} = \left( \prod_{t =
a+1}^{b}x_{t}^{\beta_{t}} \right)
\operatorname{lcm}(\lcmt{G}{i},\mu)\) with $\beta_b >0$,
we deduce that $\deg{\lcmt{G}{i,k}} < \deg(\lcmt{G}{i,j})$.
It follows that 
$\syzy{i}{k} \prec_{\mathrm{syz}} \syzy{i}{j}$ and that 
$\syzy{i}{j} = \max_{\succ_{\mathrm{syz}}}\{ \syzy{i}{k} , \syzy{k}{j} , \syzy{i}{j}\}$.
This ends to establish \hyperlink{T2}{\(\mathsf{T}_2\)}. 

We prove now \hyperlink{T1}{\(\mathsf{T}_1\)}.
It suffices to prove that the monomial $\lcmt{G}{k}$
divides $\lcmt{G}{i,j}$.
Since $\lcmt{G}{k}$ divides $x_{c}\prod_{t =
1}^{a}x_{t}^{\beta_{t}} $
 and $x_{c} \prod_{t = 1}^{a}x_{t}^{\beta_{t}} $ divides
 $\lcmt{G}{i,j}$ by definition of $c$, it follows that $\lcmt{G}{i,k,j} = \lcmt{G}{i,j}$. 

\paragraph*{Case 2: \(a=b\)}
Since $\bdg{}{}{I}{\grs}$ is reduced, it holds that
$\sharp(\bdg{\ast,}{1}{I}{\grs}) \leq 1$. We deduce that $a
\geq 2$ (because \(i < j\)).
Consider \(\proj{\lcmt{G}{i}}{a-1} = \prod_{t =
1}^{a-1}x_{t}^{\alpha_{t}}\) and \(\proj{\lcmt{G}{j}}{a-1} = \prod_{t =
1}^{a-1}x_{t}^{\beta_{t}}\) and let \(\bm{\lambda}\) be
their least common multiple. 

Suppose first that $\bm{\lambda} \not\in \bdgb{}{a-1}{I}{\grs}$. 
Then it lies in $\lm{\grs}{I}$
and there exists $k\in \{1, \ldots , \nbgb\}$ such that $\lcmt{G}{k}$ divides 
$\bm{\lambda}$ with 
$\ind{g_k} < a$. Since $\seqpol$ is in 
$\zarifilling{\delta}{n}{\nbpol}$, by 
\cref{lemma_orderonfi_inGj} we have that $k<i<j$.  
Observe that $\lcmt{G}{k}$ divides $\lcmt{G}{i,j}$. Then
$\syzy{i}{j}$ is not in $\taylorbasis^{\star}(G)$ using
\cref{lemma_critere2}.

\smallskip
Now, suppose that $\bm{\lambda} \in \bdgb{}{a-1}{I}{\grs}$.
Since $\seqpol$ lies in $\zarifilling{\delta}{n}{\nbpol}$, 
\cref{lemma_degreeinxq_order_fillingup} implies that 
$ \alpha_{a} \leq \beta_{a}$. This implies that $\frac{\lcmt{G}{i,j}}{\lcmt{G}{j}}$
lies in $\Kp{a-1}$. 
We assume
that $\left(\frac{\lcmt{G}{i,j}}{\lcmt{G}{j}},
g_j\right)$ is not \type,
 and we prove that $\syzy{i}{j}$ does not belong to
 $\taylorbasis^{\star}(G)$. 

Since $\frac{\lcmt{G}{i,j}}{\lcmt{G}{j}}$ lies in $\Kp{a-1}$
and $\left(\frac{\lcmt{G}{i,j}}{\lcmt{G}{j}}, g_j\right)$ is
not \type, we have that 
$\deg \left( \frac{\lcmt{G}{i,j}}{\lcmt{G}{j}}\right) \geq 2$.
Thus, one can write 
$ \frac{\lcmt{G}{i,j}}{\lcmt{G}{j}}$ as the product of two
monomials $ \tau $ and $\kappa$ in
\(\Kp{a-1}\) with 
\(\tau\) being of degree  \(2\).
Then there exists a variable $x_{c}$
that divides the monomial $\tau$ with $c \in \{1,
\ldots a-1\}$.

Observe that 
\(
\bm{\lambda} = \proj{\lcmt{G}{i,j}}{a-1} =
\proj{\tau\kappa \lcmt{G}{j}}{a-1}
\)
lies in $\bdgb{}{a-1}{I}{\grs}$ by hypothesis.
We deduce that 
\( \proj{x_{c}\lcmt{G}{j}}{a-1} \in \bdgb{}{a-1}{I}{\grs} \)
 since it divides 
 $\proj{\tau\kappa \lcmt{G}{j}}{a-1}$.
Since $\seqpol$ lies in $\zarifilling{\delta}{n}{\nbpol}$, 
we deduce by \cref{remark_fillingup_implies_stable}
that $\lm{\grs}{I}$ is stable. Moreover, $I$ is
zero-dimensional by assumption.
Thus, by \cref{thm_bijection_monomials},  
there exists $k$ in $\{1, \ldots , \nbgb\}$ such that 
\( \lcmt{G}{k} = \psi_a\left(
    \proj{x_{c}\lcmt{G}{j}}{a-1}\right) \).
Now, we prove \hyperlink{T1}{\(\mathsf{T}_1\)}.

One can write 
\( \lcmt{G}{k} = \left (x_{c}\prod_{t =
1}^{a-1}x_{t}^{\beta_{t}}\right ){x_{a}^{\gamma}} \) for
some \(\gamma > 0\). Since $1 \le c \le  a-1$, it implies that $k \ne j$. Note also that \(\gamma < \beta_a\)
since  \(\lcmt{G}{j}\in \lm{\grs}{G}\). 
We deduce that \(\lcmt{G}{j,k} = x_c\lcmt{G}{j}\).

Note
that \(\lcmt{G}{k,i,j} = \operatorname{lcm}(\lcmt{G}{i,j},
\lcmt{G}{j,k}) \), with \(\lcmt{G}{i,j} = \kappa \tau
\lcmt{G}{j}\) and  \(\lcmt{G}{j,k} = x_{c} \lcmt{G}{j}\). The least common
multiple of these two is \(\kappa\tau \lcmt{G}{j}\), since
\(x_c\) divides  \(\tau\). This 
proves \hyperlink{T1}{\(\mathsf{T}_1\)}, provided that \(k
\not\in \{i,j\} \), which we prove now. 

Since $\gamma < \beta_{a}$ and $G$ is ordered by
$\succ_{\mathrm{ap}}$, we deduce that $k<j$.
Since $\left( \frac{\lcmt{G}{i,j}}{\lcmt{G}{j}}, g_j \right)$ is not 
\type\ and \(\lcmt{G}{j,k} = x_c\lcmt{G}{j}\), it holds that
$i \neq k$. This finishes the proof of
\hyperlink{T1}{\(\mathsf{T}_1\)}.

\smallskip
It remains to prove \hyperlink{T2}{\(\mathsf{T}_2\)}.
By the above paragraphs, the following inequalities hold 
\[
\deg(\lcmt{G}{i,j})  \geq 2 + \deg(\lcmt{G}{j}) > 1+
\deg(\lcmt{G}{j}) \geq \deg(\lcmt{G}{k,j}) \] 
which implies that $\syzy{i}{j} \succ_{\mathrm{syz}} \syzy{k}{j}$. 

Note also that \(
\lcmt{G}{i,k} = \operatorname{lcm}\left( \lcmt{G}{i} ,
\left (x_{c}\prod_{t =
1}^{a-1}x_{t}^{\beta_{t}}\right ) x_{a}^{\gamma}\right)\). 
By the definition of $c$, it follows that
\(\lcmt{G}{i,k}\) equals \( \operatorname{lcm}\left( \lcmt{G}{i} ,
\left (\prod_{t = 1}^{a-1}x_{t}^{\beta_{t}}\right ) x_{a}^{\gamma} \right) 
= x_{a}^{{\sigma}}\bm{\lambda}\)
with ${\sigma} =\max( \gamma, \alpha_{a})$.
Since, as established before, $\gamma < \beta_{a}$ and
$\alpha_{a} \leq \beta_{a}$, we 
deduce that $\lcmt{G}{i,k}$ divides $\bm{\lambda}x_a^{\beta_{a}}$.
It follows that $\lcmt{G}{i,k}$ divides $\lcmt{G}{i,j}$. We
already established that $j>k$.
Then, we deduce that $\syzy{i}{j} \succ_{\mathrm{syz}} \syzy{i}{k}$.
This finishes the proof of \hyperlink{T2}{\(\mathsf{T}_2\)}.

\section{On the shape of the monomial staircase}
\label{sec:shapemonomialstaircase}

This section is devoted to prove \cref{thm_card_bdg_deg_moitMB}, which 
provides an explicit description of certain leading monomials in a minimal 
$\grs$-Gröbner basis generated by $n$ generic polynomials. 

\subsection{Gr\"obner bases of truncated polynomials}
\label{subsection_Macaulay_matrices}

We start with an auxiliary result
(\cref{thm_semieregseq_macmat} below).

For a homogeneous polynomial $f\in \Kp{n}$, we define for $1 \leq j
\leq n-1$ the polynomial \(\truncfu{f}{j}\in \Kp{j}\) as the
polynomial obtained by instantiating its variables \(x_{j+1}, \ldots,
x_n\) to \(0\).
By convention $\truncfu{f}{n} = f$.
Observe that for all $1 \leq j \leq n$, the polynomial $\phi(f, j)$ is homogeneous and 
has the same degree as $f$.

For a sequence $\seqpol = (f_1, \ldots , f_\nbpol)$ of polynomials, we set 
\[ \truncfu{\seqpol}{j} = \left( \truncfu{f_1}{j}, \ldots , \truncfu{f_{\nbpol}}{j} \right) \]
 for $j$ in $\{1, \ldots n\}$.
We denote by $I_j(\seqpol)$ the ideal $\langle \truncfu{\seqpol}{j} \rangle $ of $\Kp{j}$ generated by the elements of $\truncfu{\seqpol}{j}$.

\begin{thm}
\label{thm_semieregseq_macmat}
Let $\seqpol=(f_1, \ldots , f_\nbpol)$ be a sequence of homogeneous polynomials of degree $\delta$ in $\Kp{n}$.
Let $\bdg{}{}{I}{\grs}$ be the reduced $\grs$-Gröbner basis of $I = \langle \seqpol \rangle$.
Let $j$ be an integer in $\{1, \ldots , n \}$. 

Then the reduced $\grs$-Gröbner basis $\bdg{}{}{I_j(\seqpol)}{\grs}$ of the ideal $I_j(\seqpol)$ in 
$\Kp{j}$ is the set:
\[ \{ \truncfu{g}{j} \mid g \in \bigsqcup_{i = 0}^{j} \bdg{\ast,}{i}{I}{\grs} \}. \]
\end{thm}

The proof of \cref{thm_semieregseq_macmat} relies on
\cref{lemma_Macmat_echlo_contains_GB_homog} and
\cref{lemma_sub_Macaulay_matrix} below and the analysis of
properties of Macaulay matrices. 

Let $\seqpol$ be a finite subset of
$\Kp{n}$ ordered with an ordering $\succ_1$ and $T$ be a finite set of monomials ordered with an
ordering $\succ_2$.  The matrix
$\mac{\succ_1,}{\succ_2}{\seqpol}{T}{\succ_1}{\succ_2}$ is defined as follows:
its rows are indexed by the elements of $\seqpol$, and its
columns are indexed by the monomials in $T$ in decreasing
order. The entry at row $f$ in $\seqpol$ and column $t$ in
$T$ is the coefficient of the monomial $t$ in the polynomial
$f$. 
When the context already fixes the orderings, we use the notation $\mac{}{}{\seqpol}{T}{}{}$.

For the remainder of the document, we adopt the convention that the pivots of a 
matrix in row echelon or row reduced echelon form are positioned on the left.

\begin{remark}
\label{remark_pivot_equal_leadingmonomial}
Consider a row $L$ of a Macaulay matrix as a submatrix.
Remark that $L$ is in row echelon form. The matrix $L$
represents an associated polynomial $f$.  The pivot column
of $L$  corresponds to the leading monomial of $f$, with
respect to the monomial ordering used for the columns.
\end{remark}

Let $\seqpol = (f_1, \ldots , f_\nbpol)$ be a sequence of polynomials with \( \deg(f_i) = \delta\).
For $e$ in $\mathbb{N}$, as in \cref{lemma_generator_homog_space} we define the set
\[ \enspol{e}{n}{\seqpol} = \{ m\ f_i \mid i \in \{1, \ldots
, \nbpol \}, m \in \mon{e-\delta,}{n} \}, \] 
which is ordered using the order \(\succ_{\seqpol}\) defined below:
\[ m_1\ f_{i_1} \succ_{\seqpol} m_2\ f_{i_2} \quad
    \Longleftrightarrow \quad m_1 \grs m_2 \quad \text{or}  \quad
(m_1 = m_2 \quad \text{and} \quad i_1 > i_2). \]

For the remainder of this subsection, we assume that for all $j$ in $\mathbb{N}$ and
 $d$ in $\mathbb{N}$,
the set of monomials $\mon{d,}{j}$ is ordered with the grevlex ordering. 
In the same way, the set $\enspol{e}{n}{\seqpol}$ is ordered by $\succ_{\seqpol}$.

\begin{lemma}
\label{lemma_Macmat_echlo_contains_GB_homog}
Let $\seqpol = (f_1, \ldots, f_\nbpol)$ be a sequence of
homogeneous polynomials of degree $\delta$ in $\Kp{n}$, and
let $e\in\mathbb{N}$. Let $G=\bdg{}{}{I}{\grs}$ be 
the reduced $\grs$-Gröbner basis of the ideal $I = \langle
\seqpol \rangle$. 
We denote by $E_{e}$ the row reduced echelon form of the
matrix  \[
\mac{}{}{\enspol{e}{n}{\seqpol}}{\mon{e,}{n}}{\grs}{\succ_{\seqpol}}.
\] Then, the following assertions hold.
\begin{enumerate}
\item
The matrix $E_{e}$ contains, as rows, all polynomials $g$ in $G$ such that $\deg(g) = e$.
\item Let $g$ be a nonzero polynomial represented by a row
    of $E_{e}$.
    The polynomial $g$ lies in $G$ if and only
    if there is no polynomial $h$ which is represented by a
    row of $E_{d}$ with $d<e$ such that $\lmdrl{h}$ 
    divides $\lmdrl{g}$.
\end{enumerate}
\end{lemma}

\begin{proof}
The first assertion is already claimed in 
\cite[Section 2]{lazard83}; for completeness, we provide 
an explicit proof in the appendix.

 We prove now the second assertion. 
 Let $g$ be a polynomial represented by a row of $E_{e}$. 
 Note that by \cref{lemma_generator_homog_space} and \cite[Section
 1.4, Theorem 5]{hoffman1971linear}, for all $d\in\mathbb{N}$, the
 rows of $E_d$ generate the $\K$-vector space of homogeneous
 polynomials of degree $d$ in $I$. Hence, for any \(h\) represented by a
row of  \(E_d\) for  \(d\le e\), it holds that  \(h\in I\).

$\left( \Longrightarrow \right)$ We assume that $g$ lies in $G$.
Now, we consider a polynomial \(h\) represented by a
row of  \(E_d\) for  \(d<e\).
Suppose, by contradiction, that 
$\lmdrl{h}$ divides $\lmdrl{g}$. As previously established,
the polynomial $h$ lies in $I$.
By the definition of Gröbner bases, there exists a polynomial $p$ in $G$
such that $\lmdrl{p}$ divides $\lmdrl{h}$. By the definition of the 
grevlex ordering, it holds that $\deg(p) \leq d$. We deduce that
the two polynomials $p$ and $g$ can not be equal and that 
$\lmdrl{p}$ divides $\lmdrl{g}$. This is a contradiction since 
both $p$ and $g$ lie in $G$, which is assumed to be reduced. Hence, such a polynomial $h$ can not
exist.

$\left( \Longleftarrow \right)$ Assume that there is no polynomial $h$ 
which is represented by a row of $E_d$ with $d< e$ such that  
$\lmdrl{h}$ divides $\lmdrl{g}$. Since 
the $\K$-vector space of homogeneous 
polynomials of degree $d$ in $I$ is exactly the one generated by
the rows of $E_d$, all the monomials that index the columns of pivots
of $E_d$ are all the possible leading monomials of degree $d$ in $I$.
By assumption, none of these leading monomials divide $\lmdrl{g}$ for all $d
< e$, while \(g\) lies in the  \(\K\)-vector space of all
homogeneous polynomials of degree \(e\). This implies that no leading monomial of a 
polynomial in $I$ divides $\lmdrl{g}$ other than $\lmdrl{g}$ itself.
We  
deduce that there exists a polynomial $f$ in $G$ with 
$\lmdrl{f} = \lmdrl{g}$. By the first assertion of the lemma, \(f\) is
represented by a row of $E_e$.
This row is the same as the one that represents $g$ since their leading monomials
are the same and $E_e$ is in row echelon form. Thus, $f = g$ and $g$ lies in $G$ as claimed.
\end{proof}

Let $\seqpol = (f_1, \ldots , f_\nbpol)$ be a sequence of polynomials.
For $d\in\mathbb{N}$ and $j\in\{1, \ldots , n\}$, we define the set
\[ \enspoltr{d}{j}{\seqpol} = \{ m\ \truncfu{f_i}{j} \mid i \in \{1, \ldots , \nbpol \}, m \in \mon{d-\delta,}{j} \}, \] 
which we order as follows:
\[ m_1\ \truncfu{f_{i_1}}{j} \succ_{\seqpol} m_2\ \truncfu{f_{i_2}}{j} \]
 if $m_1 \grs m_2$ or
$m_1 = m_2$ and $i_1 > i_2$.

For the remainder of the section, we use the ordering defined above to order the rows of Macaulay matrices.

\begin{lemma}
\label{lemma_sub_Macaulay_matrix}
Let $\seqpol= (f_1, \ldots , f_\nbpol)$ be a sequence of polynomials of degree \( \delta \) in $\Kp{n}$.
For fixed $d \geq 0$ and $j\in\{2, \ldots , n\}$, the Macaulay matrix $\mac{}{}{\enspoltr{d}{j-1}{\seqpol}}{\mon{d,}{j-1}}{\grs}{\succ_{\truncfu{\seqpol}{j-1}}}$ is a submatrix of $\mac{}{}{\enspoltr{d}{j}{\seqpol}}{\mon{d,}{j}}{\grs}{\succ_{\truncfu{\seqpol}{j}}}$
 in the following way:

\bigskip
\[
  \mac{}{}{\enspoltr{d}{j}{\seqpol}}{\mon{d,}{j}}{\grs}{\succ_{\truncfu{\seqpol}{j}}} =
  \begin{pNiceArray}[margin]{ccc|cccc}
   \Block{2-3}{\mac{}{}{\enspoltr{d}{j-1}{\seqpol}}{\mon{d,}{j-1}}{\grs}{\succ_{\truncfu{\seqpol}{j-1}}}} & & &  & & & \\
  & & & & & &\\
  0 ~~~~~~~~~~~~~& \cdots & ~~~~~~~~~~~~~0 & \Block{1-4}{  \ast   } & & & \\
  \vdots ~~~~~~~~~~~~~ &  & ~~~~~~~~~~~~~\vdots & & & & \\
  0 ~~~~~~~~~~~~~ & \cdots & ~~~~~~~~~~~~~0 & & & &
  \CodeAfter
   \OverBrace{1-1}{3-3}{\mon{d,}{j-1}}
   \UnderBrace{5-1}{5-7}{\mon{d,}{j}}
   \tikz \draw [black] (3-|1) -| (3-|4);
\end{pNiceArray} . 
 \]
\end{lemma}

\bigskip

\begin{proof}First, the columns of the matrix $\mac{}{}{\enspoltr{d}{j}{\seqpol}}{\mon{d,}{j}}{\grs}{\succ_{\truncfu{\seqpol}{j}}}$ are indexed
by the elements of the set $\mon{d,}{j}$ and ordered according to the grevlex ordering.
Let $m_1$ be a monomial in $\mon{d,}{j-1}$ and $m_2$ be a monomial in $\mon{d,}{j} \setminus \mon{d,}{j-1}$.
The monomials $m_1$ and $m_2$ have same degree $d$ and the equality $\deg_{x_\ell}(m_1) = \deg_{x_\ell}(m_2) = 0$
holds for all $\ell>j$. Moreover, the inequality $\deg_{x_j}(m_1) < \deg_{x_j}(m_2)$ holds since $\deg_{x_j}(m_1) = 0$ and $\deg_{x_j}(m_2)>0$.
Thus, we obtain $m_1 \grs m_2$.
We deduce that $\mon{d,}{j-1}$ is a subset of $\mon{d,}{j}$ whose elements
are greater than any of those in $\mon{d,}{j} \setminus \mon{d,}{j-1}$.

Next, we define the following set of polynomials
\[ S = \{ m\ \truncfu{f_i}{j} \mid i \in \{1, \ldots , s \},\ m \in \mon{d-\delta,}{j-1} \}. \]
Let $m_1\ \truncfu{f_{i_1}}{j}$ be in $S$ and $m_2\ \truncfu{f_{i_2}}{j}$ be in 
$\enspoltr{d}{j}{\seqpol} \setminus S$. Since the monomial $m_1$ lies in $\mon{d-\delta,}{j-1}$
and $m_2$ lies in $\mon{d-\delta,}{j} \setminus \mon{d-\delta,}{j-1}$, we already established that $m_1 \grs m_2$. It follows that $m_1\ \truncfu{f_{i_1}}{j}$ is greater 
than $m_2\ \truncfu{f_{i_2}}{j}$ with respect to the order we use on $\enspoltr{d}{j}{\seqpol}$.
We deduce that $S$ is a subset of $\enspoltr{d}{j}{\seqpol}$ whose elements
are greater than any of those in $\enspoltr{d}{j}{\seqpol} \setminus
S$.

It remains to prove two more facts.

First, we prove that the submatrix defined by the rows indexed
by elements of $S$
and the columns indexed by element of $\mon{d,}{j-1}$ is the matrix $\mac{}{}{\enspoltr{d}{j-1}{\seqpol}}{\mon{d,}{j-1}}{\grs}{\succ_{\truncfu{\seqpol}{j-1}}}$.
As proved above, the sets $S$ and $\mon{d,}{j-1}$ are sets whose elements index the first
rows and columns of $\mac{}{}{\enspoltr{d}{j}{\seqpol}}{\mon{d,}{j}}{\grs}{\succ_{\truncfu{\seqpol}{j}}}$ respectively.
Let $m\ \truncfu{f_i}{j}$ be in $S$.
The polynomial defined by the row indexed by $m\ \truncfu{f_i}{j}$ and the 
columns indexed by elements of $\mon{d,}{j-1}$ is:
\begin{align*}
(m \times \truncfu{f_i}{j})(x_1,\ldots , x_{j-1}, 0 ) &= m \times
((\truncfu{f_i}{j})(x_1,\ldots , x_{j-1}, 0 ))\qquad  \text{(because } 
m \text{ lies in } \Kp{j-1}) \\
&= m \times \truncfu{f_i}{j-1}.
\end{align*}
Therefore, the matrix $\mac{}{}{\enspoltr{d}{j-1}{\seqpol}}{\mon{d,}{j-1}}{\grs}{\succ_{\truncfu{\seqpol}{j-1}}}$ appears as a submatrix
in the top-left corner of the Macaulay matrix $\mac{}{}{\enspoltr{d}{j}{\seqpol}}{\mon{d,}{j}}{\grs}{\succ_{\truncfu{\seqpol}{j}}}$.

Finally, we establish that all the rows below this submatrix
have zeroes in the entries corresponding to the columns indexed by $\mon{d,}{j-1}$.
Let $m\ \truncfu{f_i}{j}$ be one of these rows. Then, the monomial $m$ lies in the set $\mon{d-\delta,}{j} \setminus \mon{d-\delta,}{j-1}$.
This implies that every monomial in the support of $m\ \truncfu{f_i}{j}$ is divisible by $x_j$.
Therefore, its support does not intersect $\mon{d,}{j-1}$, and the corresponding entries are zero.
\end{proof}

Now, we can prove \cref{thm_semieregseq_macmat}.

\begin{proof}[Proof of \cref{thm_semieregseq_macmat}]
Let $j$ be in $\{1,\ldots ,n\}$ and $d$ in $\mathbb{N}$.
Using \cite[Section 1.4, Theorem 5]{hoffman1971linear}, there exists a square invertible matrix $U_{d,j}$ such that 
\[ U_{d,j}\mac{}{}{\enspoltr{d}{j}{\seqpol}}{\mon{d,}{j}}{}{} = E_{d,j} \] 
with $E_{d,j}$ being in row reduced echelon form.
Using the inclusion of matrices stated in
\cref{lemma_sub_Macaulay_matrix}, we consider the following matrix
product:

\[
 \begin{pNiceArray}[margin]{c|ccccc}
   \Block{2-1}{U_{d,j-1}}  & \Block{2-5}{\ast}&   & & & \\
  & & & & & \\
  \hline
  0 \cdots 0 & \Block{3-5}{  Id   } &  & & & \\
  \vdots~~~~~\vdots &  & & & & \\
  0 \cdots 0 & & & & &
\end{pNiceArray}
\begin{pNiceArray}[margin]{ccc|cccc}
   \Block{2-3}{\mac{}{}{\enspoltr{d}{j-1}{\seqpol}}{\mon{d,}{j-1}}{}{}} & & &  & & & \\
  & & & & & &\\
  0 ~~~~~~~~~~~~~& \cdots & ~~~~~~~~~~~~~0 & \Block{1-4}{  \ast   } & & & \\
  \vdots ~~~~~~~~~~~~~ &  & ~~~~~~~~~~~~~\vdots & & & & \\
  0 ~~~~~~~~~~~~~ & \cdots & ~~~~~~~~~~~~~0 & & & &
  \CodeAfter
   \tikz \draw [black] (3-|1) -| (3-|4);
\end{pNiceArray} \]
\[
 = \begin{pNiceArray}[margin]{ccc|cccc}
   \Block{2-3}{E_{d,j-1}} & &  &  & & & \\
  & & & & & &\\
  0 ~~~~~& \cdots & ~~~~~0 & \Block{1-4}{  \ast   } & & & \\
  \vdots ~~~~~ &  & ~~~~~\vdots & & & & \\
  0 ~~~~~ & \cdots & ~~~~~0 & & & &
  \CodeAfter
   \tikz \draw [black] (3-|1) -| (3-|4);
\end{pNiceArray} \]
\[
 = \begin{pNiceArray}[margin]{ccc|cccc}
   \Block{1-3}{E} & & & \Block{1-4}{  \ast   } & & & \\
  \hline
  \Block{1-3}{0 ~~~~~~ \cdots  ~~~~~~~0} & & & \Block{4-4}{  B   } & & & \\
  0 ~~~~~& \cdots & ~~~~~0 &  & & & \\
  \vdots ~~~~~ &  & ~~~~~\vdots & & & & \\
  0 ~~~~~ & \cdots & ~~~~~0 & & & &
  \CodeAfter
   \tikz \draw [black] (3-|1) -| (3-|4);
\end{pNiceArray} \]
where $E$ is the matrix formed by the nonzero rows of $E_{d,j-1}$.

Therefore, the row reduced echelon form $E_{d,j}$ 
of $\mac{}{}{\enspoltr{d}{j}{\seqpol}}{\mon{d,}{j}}{}{}$ can be expressed as follows:
\begin{align*}
 E_{d,j} = \begin{pNiceArray}[margin]{ccc|cccc}
   \Block{2-3}{E_{d,j-1}} & & &  & & & \\
  & & & & & &\\
  0 ~~~~~& \cdots & ~~~~~0 & \Block{1-4}{  \ast   } & & & \\
  \vdots ~~~~~ &  & ~~~~~\vdots & & & & \\
  0 ~~~~~ & \cdots & ~~~~~0 & & & &
  \CodeAfter
  \tikz \draw [black] (3-|1) -| (3-|4);
\end{pNiceArray}. \tag*{$(\ast)$}
\end{align*}

Let \(G_j\) denote the 
reduced $\grs$-Gröbner basis $\bdg{}{}{I_j(\seqpol)}{\grs}$. Now, we
establish that 
\[G_j =  \{ \truncfu{g}{j} \mid g \in \bigsqcup_{i = 0}^{j}
\bdg{\ast,}{i}{I}{\grs} \}. \]
$(\supseteq)$ Let $f = \truncfu{g}{j} $ with $g$ in $G = \bigsqcup_{i =
0}^{j} \bdg{\ast,}{i}{I}{\grs}$. We prove that $f$ lies in the set
$G_j$.  Consider the integer $D$ which is the
maximal degree of the elements in the reduced $\grs$-Gröbner basis of
$I$.  Since $g$ lies in the reduced $\grs$-Gröbner basis of $I$, it
follows from \cref{lemma_Macmat_echlo_contains_GB_homog} that there
exists  $e$ in $\{ \delta , \ldots , D \}$ such that $g$ is
represented by a row of $E_{e,n}$.
Since $g$ lies in
$\bigsqcup_{i = 0}^{j} G_{\ast,i}(I)$, then its leading monomial lies in $\Kp{j}$,
 which implies that $f$ cannot be zero.
Using $(\ast)$, we deduce that $f$ is a nonzero row of $E_{e,j}$.

By \cref{lemma_Macmat_echlo_contains_GB_homog}, $f$ lies in 
$G_j$ if and only if there is no
polynomial $h$ which is represented by a row of $E_{d,j}$ with $d< e$ such 
that $\lm{\grs}{h}$ divides $\lm{\grs}{f}$.
Assume by contradiction that such a $h$ exists. Using $(\ast)$, there
exists a polynomial $\eta$ which is represented by a row of $E_{d,n}$ and
which is such that $h = \truncfu{\eta}{j}$.  Since $\lm{\grs}{\eta}$ and
$\lm{\grs}{g}$ lie $\Kp{j}$, remark that 
\[ \lm{\grs}{\eta} = \lm{\grs}{h} \quad  \text{and that} \quad \lm{\grs}{g} = \lm{\grs}{f}. \]

Since $\lm{\grs}{h}$ divides $\lm{\grs}{f}$, we deduce that
$\lm{\grs}{\eta}$ divides $\lm{\grs}{g}$. 
Since $d< e$, by the second assumption of \cref{lemma_Macmat_echlo_contains_GB_homog}
 this contradicts the fact that $g$ lies in $G$. Thus, such
 a polynomial $h$
 can not exist. It follows that $f$ lies in $G_j$.  

\medskip

$(\subseteq)$ Let $f$ be in $G_j$. We prove
 that there exists a polynomial $g$ in $\bigsqcup_{i = 0}^{j} \bdg{\ast,}{i}{I}{\grs}$ such that $f = \truncfu{g}{j}$.

Consider the integer $D_j$ which is the maximal degree of the elements 
in \(G_j\).
Using \cref{lemma_Macmat_echlo_contains_GB_homog}, there exists 
$e$ in $\{ \delta , \dots , D_j \}$ such that $f$ is a row of $E_{e,j}$.
By $(\ast)$, $f$ is of the form $\truncfu{g}{j}$ where $g$ is a nonzero row of $E_{e,n}$.
Using \cref{remark_pivot_equal_leadingmonomial} and $(\ast)$, we have that
 \[ \lm{\grs}{g} = \lm{\grs}{\truncfu{g}{j}} = \lm{\grs}{f}.\]

In order to prove that $g$ lies in $G$, we show that there does not exist a polynomial $h$
which is represented by a row of $E_{d,n}$ with $d< e$ such that
$\lm{\grs}{h}$ divides $\lm{\grs}{g}$.

Assume by contradiction that there exists $d$ in $\{ \delta , \dots ,
e - 1 \}$ such that the leading monomial of some polynomial $h$
represented by a row of $E_{d,n}$
divides $\lm{\grs}{g}$.  Since $\lm{\grs}{g}$ is in $\Kp{j}$, the
monomial $\lm{\grs}{h}$ lies in $\Kp{j}$.  Therefore, the polynomial
$\truncfu{h}{j}$ is a nonzero row of $E_{d,j}$ with
$\lm{\grs}{\truncfu{h}{j}}$ dividing $\lm{\grs}{f}$. Since $f\in G_j$, by
\cref{lemma_Macmat_echlo_contains_GB_homog} this is a contradiction. 
Thus, such a polynomial $h$ does no exist. It follows by
\cref{lemma_Macmat_echlo_contains_GB_homog}
that $g$ lies in $G$. Moreover, since $\lm{\grs}{g}$ lies in $\Kp{j}$,
the polynomial $g$ lies in $\bigsqcup_{i = 0}^{j} \bdg{\ast,}{i}{I}{\grs}$. 
\end{proof}

The following corollary is a direct consequence of \cref{thm_semieregseq_macmat}. 

\begin{cor}
\label{cor_lien_mono_base_quotient_tronqué}
Let $\seqpol=(f_1, \ldots , f_\nbpol)$ be a sequence of homogeneous polynomials of degree $\delta$ in $\Kp{n}$. Let $d$ be in $\mathbb{N}$ and $j$ be in $\{1, \ldots , n\}$.
We set $I = \langle \seqpol \rangle$. Then we have
\[ \bdgb{d,}{j}{I}{\grs}  = \bdgb{d,}{j}{I_j(\seqpol)}{\grs}. \]
\end{cor}

\begin{proof}
By definition we have that 
\begin{align*}
\bdgb{d,}{j}{I}{\grs} &=  \mon{d,}{j} \setminus \left( \mon{d,}{j} \cap \langle \bigsqcup_{i = 0}^j\bdgh{\ast,}{i}{I}{\grs} \rangle \right) \\
 &= \mon{d,}{j} \setminus \left( \mon{d,}{j} \cap \langle \bigsqcup_{i = 0}^j\bdgh{\ast,}{i}{I_{j}(\seqpol)}{\grs} \rangle \right) \tag{Using \cref{thm_semieregseq_macmat}}\\
 &= \bdgb{d,}{j}{I_j(\seqpol)}{\grs}. \tag{\text{by definition}}
\end{align*}
\end{proof}

\subsection{The last dimension of the generic monomial staircase}
\label{subsectionThelastdimensionofthegeneriquemonomialstaircase}

This subsection is devoted to prove \cref{thm_card_bdg_deg_moitMB}.  This
theorem is a specific instance of \cite[Corollary
3.2]{moreno2003degrevlex}, which characterizes the generic shape of
the monomial staircase of an ideal generated by a sequence of $n$
homogeneous polynomials.  We restate this result in a sharper form,
since the original statement is more general than the situation we
consider here. Moreover, the original theorem has 
several conclusions, which are presented in two different versions in
\cite[Corollary 3.2]{moreno2003degrevlex}. One of the key differences
is that in our setting all polynomials have the same degree. Without
this assumption, further subtleties arise, as discussed in
\cite{moreno2003degrevlex}. Moreover, we consider specifically the
case where the number of generating
polynomials equals the number of variables.
Note that this subsection only concerns homogeneous polynomials.
 For a given pair
$(\delta, n)$ of integers, we denote by $\zarimore{\delta}{n}$ the Zariski open set
of sequences $(f_1, \ldots, f_n)$ of homogeneous polynomials of degree
$\delta$ in $\Kp{n}$ that satisfy the properties stated in \cite[Corollary
3.2]{moreno2003degrevlex}.  The fact that the set
$\zarimore{\delta}{n}$ is Zariski dense relies on \cite[Conjectures 1.2
and 1.6]{moreno2003degrevlex}.  We first recall some notions.

As in \cite[Definition 2.1]{pardue2009syzygies}, for a series \(\sum_{d = 0}^{+\infty} a_d\),  
we define  
\[
\left[ \sum_{d = 0}^{+\infty} a_d \right]_+ = \sum_{d = 0}^{+\infty} b_d,
\]
where
\[
\begin{cases}
b_d = a_d & \text{if } a_i > 0 \text{ for all } i \leq d, \\
b_d = 0 & \text{otherwise}.
\end{cases}
\]

Let $n\in \mathbb{N}$ and let $I$ be an ideal of $\Kp{n}$ generated by homogeneous polynomials. 
We recall that its Hilbert series is defined by
\[
    \hilse{I}{z} = \sum_{d = 0}^{+ \infty} \beta_d z^d 
\]
where $\beta_d = \card{\bdgb{d,}{n}{I}{\grs}}$.
Strictly speaking, this is the Hilbert series of the quotient ring $\Kp{n}/I$, 
but we use the above notation and the expression ``Hilbert series of $I$'' by abuse of language.

We recall that for an integer \( d \), we denote by \( \Khp{n}{d} \) the \( \K \)-vector space generated by all homogeneous polynomials of degree \( d \)
in \( \Kp{n} \).
Let $n,\delta$ and $\nbpol$ be three integers, a sequence $\seqpol = (f_1, \ldots ,f_{\nbpol})$ 
of polynomials in $\Khp{n}{\delta}$ is said to be semi-regular if 
\[ \hilse{\langle \seqpol \rangle}{z} = \left[ \frac{(1-z^{\delta})^{\nbpol}}{(1-z)^n} \right]_+ .\]

We denote by $\zarisemi{\delta}{n}{\nbpol}$ the set of all semi-regular sequences $(f_1, \ldots ,f_{\nbpol})$
 of polynomials in $\Khp{n}{\delta}$.
If $\nbpol \leq n$, then the set $\zarisemi{\delta}{n}{\nbpol}$ is 
the set of all regular sequences $(f_1, \ldots ,f_{\nbpol})$
 of polynomials in $\Khp{n}{\delta}$ \cite[Proposition 1]{pardue2010generic}. In that case,
the set $\zarisemi{\delta}{n}{\nbpol}$ is known to be Zariski dense. Conversely,
if $\nbpol > n$, the statement that $\zarisemi{\delta}{n}{\nbpol}$ is Zariski dense
is \cite[Conjecture 1.2]{moreno2003degrevlex}, which was first exposed as 
\cite[Conjecture 1.1]{froberg1994hilbert}. 

\begin{thm}[Variant of {\cite[Corollary 3.2]{moreno2003degrevlex}}]
\label{thm_card_bdg_deg_moitMB}
Let $\delta \geq 1$ and $n \geq 2$ be integers. Let $\seqpol$ be in $\zarimore{\delta}{n}$
and set $I = \langle \seqpol \rangle$. 
Assume that $\truncfu{\seqpol}{n-1}$ lies in $\mathcal{S}_{\delta,n-1,n}$. 
We define $\dregmb = n(\delta-1)+1$ and $\dregprev = \left\lfloor \frac{\dregmb-1}{2} \right\rfloor + 1$. 
We also dfine the coefficient $\beta_d$ as $\left[ \frac{(1-z^{\delta})^n}{(1-z)^{n-1}} \right]_+ = \sum_{d = 0}^{\dregprev-1} \beta_d z^d$.
Then for all integer $d$ in $\{ \dregprev +1, \ldots , \dregmb \}$ we have that 
\begin{enumerate}[1)]
\item $\bdgh{d,}{n}{I}{\grs} = \{ x_n^{2d-\dregmb} m \mid m \in \bdgb{\dregmb-d,}{n-1}{I}{\grs} \}$
\item $\sharp (\bdg{d,}{n}{I}{\grs}) = \beta_{\dregmb-d}$.
\end{enumerate}
\end{thm}

To justify our formulation of this theorem, we relate our notations to those of  
\cite[Corollary 3.2]{moreno2003degrevlex}.  
An homogeneous ideal $I$ is zero-dimensional if and only if $\hilse{I}{z}$ 
is a polynomial in $z$. 
In this case, the degree of regularity of $I$ is equal to 
$\deg(\hilse{I}{z}) + 1$. 
The integers $\dregprev$ and $\dregmb$ correspond respectively to  
$\overline{\sigma}$ and $\overline{\delta}$ in \cite[Corollary 3.2]{moreno2003degrevlex}. As explained in the aforementioned corollary,
the integer $\dregprev$ denotes the degree of regularity 
of an ideal generated by a sequence in $\zarisemi{\delta}{n-1}{n}$.

Furthermore, we make explicit the definition of the integer $\mu$ defined in \cite[Section 2]{moreno2003degrevlex}, which is given by  
\[
\mu = \dregmb-1 - 2(\dregprev-1) = \dregmb - 2\left\lfloor \frac{\dregmb-1}{2} \right\rfloor -1.
\]
This implies, in our case, that $\mu \in \{0,1\}$.

The second formulation of \cite[Corollary 3.2]{moreno2003degrevlex} states that for any degree 
$d$ which lies in $\{ \dregprev + 1, \ldots , \dregmb\}$, the elements of $\bdgh{d,}{\ast}{I}{\grs}$ have degree in $x_n$ equal to $2d-\dregmb$, and moreover  
\[
\card{\bdgh{d,}{n}{I}{\grs}} = \beta_{\dregmb-d}.
\]  
It also follows that for such integers $d$, we have $\bdgh{d,}{\ast}{I}{\grs} = \bdgh{d,}{n}{I}{\grs}$ since $2d-\dregmb \geq 2$.  
This proves the second assertion 
of \cref{thm_card_bdg_deg_moitMB}. We now turn to the first assertion.  

We have already established that  
\[
\bdgh{d,}{n}{I}{\grs} \subseteq \{ x_n^{2d-\dregmb} m \;\mid\; m \in \bdgb{\dregmb-d,}{n-1}{I}{\grs} \}.
\]  
In addition, we know that $\card{\bdgh{d,}{n}{I}{\grs}} = \beta_{\dregmb-d}$.  
Thus, it remains to show that  
\[
\card{\{ x_n^{2d-\dregmb} m \;\mid\; m \in \bdgb{\dregmb-d,}{n-1}{I}{\grs} \}} = \beta_{\dregmb-d}.
\]
  
It appears that the condition $\seqpol \in \zarimore{\delta}{n}$ already implies that  
$\truncfu{\seqpol}{n-1} \in \mathcal{S}_{\delta,n-1,n}$.  
Nevertheless, we provide an additional explanation using \cref{cor_lien_mono_base_quotient_tronqué}.  

\begin{align*}
\card{\{ x_n^{2d-\dregmb} m \;\mid\; m \in \bdgb{\dregmb-d,}{n-1}{I}{\grs} \}} &= \card{\bdgb{\dregmb-d,}{n-1}{I}{\grs}} \\
&= \card{\bdgb{\dregmb-d,}{n-1}{I_{n-1}(\seqpol)}{\grs}} \tag{using \cref{cor_lien_mono_base_quotient_tronqué}} \\ 
  &= \beta_{\dregmb-d}. \tag{since $\truncfu{\seqpol}{n-1}$ lies in $\mathcal{S}_{\delta,n-1,n}$}
\end{align*}
It follows that  
\[
\bdgh{d,}{n}{I}{\grs} = \{ x_n^{2d-\dregmb} m \;\mid\; m \in \bdgb{\dregmb-d,}{n-1}{I}{\grs} \},
\]
as claimed.

\section{Algorithms and their properties}
\label{sec:Algorithmsandtheirproperties}

In \cref{subsection_Overview_of_the_algorithm}, we provide an overall description of a variant of the F4 algorithm.
As emphasized in the introduction, this variant combines Buchberger's criterion with row echelon form computations 
of Macaulay matrices to compute Gr\"obner bases. Note that these matrices may not have full rank. After computing 
their row echelon forms, the polynomials corresponding to the nonzero rows are added to the current basis. The set
of all leading monomials obtained during the computation forms what we call the Gröbner trace 
associated with the Gr\"obner basis computation.

\cref{sec:F4B_F4T} introduces first 
\algoName{algo:F4B} which is essentially the
F4 algorithm augmented with steps that store some data in some object
\(\mathcal{T}\), which we call a \emph{trace}.
This trace gives in particular some rank information that
enables us to locate full rank submatrices of the matrices considered by F4. 
This allows us to derive \algoName{algo:F4T} which can then be used to
compute Gr\"obner bases for systems of polynomial equations that 
are compatible with the trace. 
Finally, \cref{ssec:algo:properties} establishes structural properties
of these algorithms. 

\subsection{Overall description of the F4 algorithm}
\label{subsection_Overview_of_the_algorithm}

We start by recalling the basic principles of the F4 algorithm from
\cite{faugere1999new}.
Let $(f,g)$ be a pair of nonzero polynomials in $\Kp{n}$, and let 
$t = \operatorname{lcm}(\lm{\grs}{f},\lm{\grs}{g})$.
 As in \cite[Definition~2.5]{faugere1999new}, we define 
 the pair associated to $(f,g)$ by
\[ \pair{f}{g} = \left( t,\frac{t}{\lm{\grs}{f}},f,\frac{t}{\lm{\grs}{g}},g \right). \]
The degree of the pair $\pair{f}{g}$ is $\deg(t)$ by definition.
Moreover, we denote the associated S-polynomial by
\[ \Spol{f}{g}{} = \frac{t}{\lt{\grs}{f}}f - \frac{t}{\lt{\grs}{g}}g. \]

The termination of the algorithm is based on the following statement.
Let $I$ be an ideal of $\Kp{n}$ generated by the set 
$G = \{g_1, \ldots , g_{\nbgb}\}$. 
As shown in \cite[Theorem~6, Section~6, Chapter~2]{cox94} and 
\cite{buchberger1965algorithmus}, 
$G$ forms a $\grs$-Gröbner basis of $I$  
if and only if for every pair $\{i,j\}$, 
the remainder obtained when dividing $\Spol{g_i}{g_j}$ by $G$ is zero.
This division is also called reduction, and the polynomials of \(G\) used
in this division are called reductors.
This division is a multivariate version of the Euclidean division using 
a monomial ordering. It is described in \cite[Section~3, Chapter~2]{cox94}.  
However, the construction of the set of reductors involves some choices.

Starting from a sequence $\seqpol$ of polynomials of $\Kp{n}$, the F4
algorithm constructs 
the set $\explpairs$ of all pairs, and
it reduces these pairs.
In Buchberger's algorithm, these reductions are performed sequentially, 
handling one pair at a time. 
By contrast, the F4 algorithm performs several reductions simultaneously. 
To determine which pairs will be reduced at the same step, 
a \emph{selection function} $\Sel{\cdot}$ is introduced. 
Given a set of pairs $\setpair$, the function $\Sel{\setpair}$ 
returns a subset of $\setpair$, with
$\Sel{\setpair} = \emptyset$ if and only if $\setpair = \emptyset$.
Further, we choose a specific selection function that filters pairs by degree:
\[ \Sel{\setpair} = \{ p \in \setpair \mid \deg(p) = \min_{\rho \in
\setpair} \left( \deg(\rho) \right) \}.\] 

The algorithm then selects the pairs in $\Sel{\explpairs}$ and removes
them from $\explpairs$. Next it reduces them using Macaulay matrices
combined with Gaussian elimination. The resulting remainders are
polynomials that may be zero: the nonzero ones are added to $\seqpol$
and all newly generated pairs are added to $\explpairs$.  When
$\explpairs$ is empty, $\seqpol$ is a $\grs$-Gröbner basis of the
ideal $I$.

Let us now examine how the $S$-polynomials are reduced. Consider the set $\explpairs_1 = \Sel{\explpairs}$ at a given step in the algorithm. Define the set of $S$-polynomials:
\[
\explSp = \{ \Spol{f}{g} \mid p_{f,g} \in \explpairs_1 \}.
\]
A subroutine called $\mathsf{SymbolicPreprocessing}$ constructs a set $\explred$ of reductors,
 each of the form $mh$ where $m$ is a monomial and $h$ is in the current basis. 
First introduced in \cite[Section~2.3]{faugere1999new}, this subroutine checks, for each monomial $\widetilde{m} \in \supo{\explSp}$, whether $\widetilde{m}$ lies in the monomial ideal $\langle \lm{\grs}{\seqpol} \rangle$. 
If it does, then there exist a monomial $m$ and a polynomial $h \in \seqpol$ such that 
$\lm{\grs}{mh} = \widetilde{m}$. Thus, the polynomial $mh$ is inserted in $\explred$
as a reductor.
Moreover, for each reductor $mh$ added to $\explred$, the subroutine 
$\mathsf{SymbolicPreprocessing}$ goes through the same procedure as with $\widetilde{m}$ 
with any monomial in $\supo{mh}$.
This process is guaranteed to terminate \cite[Theorem 2.2]{faugere1999new}.
One can notice that there is flexibility in choosing the reductors. In our 
analysis, we do not favor any particular selection strategy.

At this stage of the algorithm, we thus have
a set $\explSp$ of polynomials to reduce and a set $\explred$ of reductors.
The reduction step is carried out using Gaussian elimination on the
Macaulay matrix associated with the polynomials in $\explSp$ and
$\explred$.  Here, we consider Macaulay matrices whose columns are
indexed by monomials ordered with respect to the grevlex ordering. 
All the sets of monomials that index the columns of matrices are
ordered with respect to the grevlex ordering; thus, we call them
sequences.
We do not make an explicit specification of the row ordering, as it has no influence on
our complexity analysis.

We define the sequence $\expluni = \explSp \cup \explred$.
Recall that $\mac{}{}{\expluni}{\supo{\expluni}}{}{}$ may not have full
rank (and often does not).
The ordering of the
polynomials in $\explSp$ is irrelevant. But note that \emph{one can order the polynomials in
$\explred$ so that $\mac{}{}{\explred}{\supo{\expluni}}{}{}$
is in row-echelon form}.  Such an ordering exists, since
$\mathsf{SymbolicPreprocessing}$ does not deal with the same monomial
twice. This latter property will play a crucial role in the complexity 
analysis we provide in \cref{sec:complexity_analysis}. 

In \cref{sec:reduction}, we will describe \algoName{algo:Reduction} which
takes as input a matrix $\mac{}{}{\explSp}{\supo{\expluni}}{}{}$,
a matrix $\mac{}{}{\explred}{\supo{\expluni}}{}{}$ in row echelon form and a
sequence of monomials $\seqmono$, and returns a row echelon form of
$\mac{}{}{\expluni}{\supo{\expluni}}{}{}$, up to a permutation of the rows.
The two input matrices must have a number of columns equal to the cardinality
of $\seqmono$ such that each column is indexed by a monomial of $\seqmono$.
\algoName{algo:Reduction} is a
specific algorithm that computes echelon forms in a certain way in order to 
obtain our complexity result. It itself uses a basic linear algebra
subroutine called $\ech{\cdot}$ that takes a matrix $\mathscr{M}$ with entries
in $\K$ as input and that returns a row echelon form of $\mathscr{M}$.
Also, as $\mac{}{}{\explred}{\supo{\expluni}}{}{}$ is already in row echelon form,
the output of \cref{sec:reduction} is of the form
\[  \begin{pNiceArray}[margin]{c}
    \mac{}{}{\explred}{\supo{\expluni}}{}{}\\
    \hline
    \mathscr{N}
    \end{pNiceArray}. \]

Consider a matrix $\mathscr{M}$ with coefficients
in $\K$ and a sequence of monomials $\seqmono$
such that each column of $\mathscr{M}$ is indexed by a monomial.
We set $\polrows{\seqmono}{\mathscr{M}}$ a subroutine that
takes $\mathscr{M}$ and $\seqmono$ as input, and that returns
the set of nonzero polynomials that are
represented by the rows of $\mathscr{M}$.
Denoting by \(\mathscr{E}\) the matrix obtained
after running \algoName{algo:Reduction}, the polynomials in the set 
\(\polrows{\supo{\expluni}}{\mathscr{E}}\) that 
have their leading monomial not divisible by the leading monomial of a polynomial
in the current basis are added to it.

\smallskip
\emph{In the sequel, we assume that the reduction algorithm performs
top-reduction}, i.e.\ the row echelon form computed is not necessarily reduced. 

\smallskip
Also, before each step of reduction,  we use the criterion introduced
in \cite[section 3.4]{gebauer1988installation} (see \cref{algo:F4-GM})
in order to decrease the number of pairs to reduce.  Once the set
$\explSp = \Sel{\explpairs}$ has been constructed, all pairs
$\pair{g_i}{g_j}$ such that $(i,j)$ lies in $\mathcal{B}(G)$ defined
in \cref{subsection_Onthesyzygiesmodule} are removed from it. 

We can now provide a description of F4 (\cref{algo:F4}), which, overall,
at each step (denoted by \(\iota\)): 
 \begin{itemize}
     \item selects pairs as described above, 
     \item computes a row reduction echelon form of the Macaulay
         matrix defined by the $S$-polynomials associated to these
         pairs and their reductors identified by the routine
         $\mathsf{SymbolicPreprocessing}$, 
     \item stores the new nonzero polynomials encoded by the
         echelonized matrix that have new leading monomials, 
     \item updates the set of pairs accordingly.
\end{itemize}
This is repeated until the set of pairs becomes empty. 

\begin{algorithm}[ht]
  \algoCaptionLabel{F4}{f_1,\ldots,f_{\nbpol}}
  \begin{algorithmic}[1]
  \Require{A sequence $\seqpol = (f_1, \ldots , f_{\nbpol})$ of polynomials in $\Kp{n}$.}
  \Ensure{A $\grs$-Gröbner basis of $\langle \seqpol \rangle$ .}
  \State $\indalgo = 0$
  \State $\mathscr{E} = \ech{ \mac{}{}{\seqpol}{\supo{\seqpol}}{}{}}$
  \State $G = \polrows{\supo{\seqpol}}{\mathscr{E}}$
  \State $\explpairs = \{ \pair{f}{g} \mid f,g \in G \text{ with } f \ne g \}$
  \While{$\explpairs \ne \emptyset$} 
    \State $ \indalgo = \indalgo + 1 $  
    \State $\explpairs_\indalgo = \Sel{\explpairs}$
    \State $\explpairs = \explpairs \setminus \explpairs_\indalgo$
    \State $\explpairs_\indalgo = \{ \pair{g_i}{g_j} \in \explpairs_\indalgo \mid (i,j) \notin \mathcal{B}(G) \}$ 
    \label{algo:F4-GM}
    \State $S_\indalgo = \{ \Spol{f}{g} \mid \pair{f}{g}  \in \explpairs_\indalgo \}$
    \State $R_\indalgo =\mathsf{SymbolicPreprocessing}(S_\indalgo,G)$
    \label{algo:F4:reducers}
    \State $ \mathscr{S} = \mac{}{}{S_\indalgo}{\supo{S_\indalgo \cup R_\indalgo}}{}{}$
    \State $ \mathscr{R} =  \mac{}{}{R_\indalgo}{\supo{S_\indalgo \cup R_\indalgo}}{}{}$
    \State $\seqmono = \supo{S_\indalgo \cup R_\indalgo}$
    \State  $\begin{pNiceArray}[margin]{c}
    \mathscr{R}\\
    \hline
    \mathscr{N}
    \end{pNiceArray} =
    \mathsf{Reduction}(\mathscr{S},\mathscr{R},\seqmono)$ 
    \For{$f$ in $\polrows{\seqmono}{\mathscr{N}}$}
      \State $G = G \cup \{ f\}$ \label{f4:step:addpolys}
      \State $\explpairs = \explpairs \cup \{ \pair{f}{g} \mid g \in G \}$
    \EndFor 
  \EndWhile
  \State \textbf{return $G$}
  \end{algorithmic}
\end{algorithm}

As explained above, the leading monomials of the newly added
polynomials  at step \(\iota\) (Line~\ref{f4:step:addpolys}) 
form what we call the trace of this Gr\"obner basis computation. 

\subsection{Descriptions of \algoName{algo:F4B} and \algoName{algo:F4T}} 
\label{sec:F4B_F4T}

The goal of \algoName{algo:F4B} is to slightly modify the F4 algorithm
in order to build an additional data structure called the Gr\"obner trace, depending on the
input sequence of polynomials \(\seqpol = (f_1, \ldots, f_s)\).
This trace allows one to obtain \algoName{algo:F4T} which computes
Gr\"obner bases for sequences of polynomials 
that are compatible with this trace, by building full rank Macaulay matrices only. 
This trace is similar to the one introduced in \cite{traverso1989grobner}, and is also
hinted at or used in \cite{berthomieu,magma,maple}.

Formally, a Gröbner trace $\mathcal{T} = (\Gamma, \Sigma, \Theta)$ is defined as
follows:
\begin{enumerate}[(1)]
\item $\Gamma$ is a sequence of sets of monomials $\Gamma_\indalgo$.
The sequence $\Gamma$ is intended to contain all leading monomials newly
obtained during the computation.
\item $\Sigma$ is a sequence of sets of positive integers $\Sigma_{\indalgo}$.
The sequence $\Sigma$ stores the pairs that do not reduce to zero during the 
computation.
\item $\Theta$ is a sequence of sets of tuples of monomials $\Theta_{\indalgo}$.
Each element in $\Theta_{\indalgo}$ is of the form $(m_1,m_2)$.
The sequence $\Sigma$ keeps track of the reductors used during
the computation.
\end{enumerate}

We now introduce \algoName{algo:F4B} whose purpose is to build such a
tracer, during the execution of the F4 algorithm.

\begin{algorithm}[ht]
  \algoCaptionLabel{F4B}{f_1,\ldots,f_{\nbpol}}
  \begin{algorithmic}[1]
  \Require{A sequence $\seqpol = (f_1, \ldots , f_{\nbpol})$ of polynomials in $\Kp{n}$.}
  \Ensure{A $\grs$-Gröbner basis of $\langle \seqpol \rangle$ and a Gröbner trace $\mathcal{T}$.}
  \State $\indalgo = 0$ and \(\mathcal{T} = \left(  \right) \)
  \State $\mathscr{E} = \ech{ \mac{}{}{\seqpol}{\supo{\seqpol}}{}{}}$
  \State $G = \polrows{\supo{\seqpol}}{\mathscr{E}}$
  \State $\mathcal{T} = \mathsf{Update}\Gamma(\mathcal{T}, \lm{}{G})$
  \label{algo:F4-BuildTracer_build_gamma0}
  \State $\explpairs = \{ \pair{f}{g} \mid f,g \in G \text{ with } f \ne g \}$
  \While{$\explpairs \ne \emptyset$} \label{algo:F4B:while}
    \State $ \indalgo = \indalgo + 1 $ \label{algo:F4-BuildTracer_index_d} 
    \State $\explpairs_\indalgo = \Sel{\explpairs}$ 
    \State $\explpairs = \explpairs \setminus \explpairs_\indalgo$ \label{algo:F4-BuildTracer_rem_pairs} 
    \State $\explpairs_\indalgo = \{ \pair{g_i}{g_j} \in \explpairs_\indalgo \mid (i,j) \notin \mathcal{B}(G) \}$ \label{algo:F4-BuildTracer_critere}
    \State $S_\indalgo = \{ \Spol{f}{g} \mid \pair{f}{g}  \in \explpairs_\indalgo \}$
    \State \label{algo:F4-BuildTracer_SymPT} $(R_\indalgo, \Theta_{\indalgo}) =\mathsf{SymbolicPreprocessingTracer}(S_\indalgo,G)$
    \State \label{algo:F4-BuildTracer_updT} $\mathcal{T} = \mathsf{Update}\Theta(\mathcal{T}, \Theta_{\indalgo})$
    \State \label{algo:F4-BuildTracer_S} $ \mathscr{S} = \mac{}{}{S_\indalgo}{\supo{S_\indalgo \cup R_\indalgo}}{}{}$
    \State \label{algo:F4-BuildTracer_R} $ \mathscr{R} =  \mac{}{}{R_\indalgo}{\supo{S_\indalgo \cup R_\indalgo}}{}{}$
    \State $\seqmono = \supo{S_\indalgo \cup R_\indalgo}$
    \State \label{algo:F4-BuildTracer_Reduction}  $\begin{pNiceArray}[margin]{c}
    \mathscr{R}\\
    \hline
    \mathscr{N}
    \end{pNiceArray} =
    \mathsf{Reduction}(\mathscr{S},\mathscr{R},\seqmono)$
    \State \label{algo:F4-BuildTracer_updS} $\mathcal{T} = \mathsf{Update}\Sigma(\mathcal{T}, \mathscr{N})$
    \State $\Gamma_{\indalgo} = \{\}$
    \For{$f$ in $\polrows{\seqmono}{\mathscr{N}}$}
      \State $\Gamma_{\indalgo} = \Gamma_{\indalgo} \cup \{ \lm{}{f} \}$ 
      \State $G = G \cup \{ f\}$
      \State $\explpairs = \explpairs \cup \{ \pair{f}{g} \mid g \in G \}$
    \EndFor 
    \State \label{algo:F4-BuildTracer_updG} $\mathcal{T} = \mathsf{Update}\Gamma(\mathcal{T}, \Gamma_{\indalgo})$
  \EndWhile
  \State \textbf{return $G$ and $\mathcal{T}$}
  \end{algorithmic}
\end{algorithm}

Note that the rows of the matrices \(\mathscr{S}\) (at
\cref{algo:F4-BuildTracer_S}) and \(\mathscr{N}\)  (at
\cref{algo:F4-BuildTracer_Reduction}) are in one-to-one correspondence
with the set of pairs $\explpairs_\indalgo $ at
\cref{algo:F4-BuildTracer_critere}. Further, we make implicit this
mapping. 

The Gr\"obner trace construction is based on three subroutines: 
\begin{itemize}
    \item $\mathsf{Update}\Gamma$ which takes as input the current
        trace $\mathcal{T} = (\Gamma, \Sigma, \Theta)$ (with $\Gamma =
        \Gamma_0, \ldots, \Gamma_\indalgo$) and a set of
        monomials \( \Gamma_{\indalgo+1} \) and replaces \(\Gamma\) by
        $\Gamma_0, \ldots, \Gamma_\indalgo, \Gamma_{\indalgo+1}$ in $\mathcal{T}$.
         It is used in
        \cref{algo:F4-BuildTracer_build_gamma0} and 
        \cref{algo:F4-BuildTracer_updG} 
        to store the
        leading monomials of the newly polynomials which are added to 
        the current basis. 
    \item \(\mathsf{Update}\Sigma\) which takes as input the current
        trace  \(\mathcal{T} = \left( \Gamma, \Sigma, \Theta \right)
        \) (with \(\Sigma =  \Sigma_0, \ldots, \Sigma_\indalgo 
        \)) and a matrix \(\mathscr{N}\) in row echelon form. It identifies the
        set \(\Sigma_{\indalgo+1}\) of indices of \(\mathscr{N}\) which encode nonzero
        polynomials and replaces $\Sigma$
        by \(\left( \Sigma_0, \ldots, \Sigma_\indalgo, \Sigma_{\indalgo+1} \right) \)
        in \(\mathcal{T}\). This routine is used in \cref{algo:F4-BuildTracer_updS}. 
    \item \(\mathsf{Update}\Theta\) which takes as input the current
        trace  \(\mathcal{T} = \left( \Gamma, \Sigma, \Theta \right)
        \) (with \(\Theta =  \Theta_1, \ldots, \Theta_\indalgo 
        \)) and a set of tuples of monomials $\Theta_{\indalgo+1}$ and replaces 
        $\Theta$ by $\Theta_0, \ldots, \Theta_\indalgo,\Theta_{\indalgo+1}$ in $\mathcal{T}$. 
        This routine is used in \cref{algo:F4-BuildTracer_updT}.
\end{itemize}

In order to build the sequence $\Theta$, we introduce the subroutine
$\mathsf{SymbolicPreprocessingTracer}$, which is a variant of
$\mathsf{SymbolicPreprocessing}$ that returns additional information.
This routine is used in \cref{algo:F4-BuildTracer_SymPT}.
It takes as input two sets of polynomials: $S_\indalgo$ (the $S$-polynomials)
and $G$ (the current elements of the basis).
It returns a set $R_\indalgo$ of reductors and a set $\Theta_{\indalgo}$ of
tuples of monomials.
The output $\Theta_{\indalgo}$ is constructed as follows:
\begin{itemize}
\item each reductor in $R_\indalgo$ is of the form $m_1 g$, where $g \in G$;
\item for each reductor $m_1 g$ in $R_\indalgo$, we add the tuple
$(m_1, \lm{}{g})$ to $\Theta_{\indalgo}$.
\end{itemize}
The set $\Theta_{\indalgo}$ thus precisely describes the choices of reductors
made by $\mathsf{SymbolicPreprocessingTracer}$.

Note that applying \(\mathsf{Update}\Sigma\) identifies which
rows of $\mathscr{S}$
(at \cref{algo:F4-BuildTracer_S} of \algoName{algo:F4B}) 
will not be reduced to \(0\). These rows were
coming from the set of pairs \(\explpairs_\indalgo\) 
at \cref{algo:F4-BuildTracer_critere} of
\algoName{algo:F4B}. Hence, we define the following subroutine:
\begin{itemize}
    \item \(\mathsf{RemoveUselessPairs}\) which takes as input a
        set 
        of pairs \(P\), a Gr\"obner trace \(\mathcal{T} = \left(
        \Gamma, \Sigma, \Theta \right) \) with \(\Sigma = \left(
        \Sigma_0, \ldots, \Sigma_\indalgo \right) \) and an integer \(\indalgo\)
        and which returns the subset of \(P\) that corresponds to \(\Sigma_\indalgo\). 
\end{itemize}
Also, once useless pairs have been removed, some reductors may become
useless.
The purpose of applying \(\mathsf{Update}\Theta\) is to
avoid considering such reductors, which are introduced at 
\cref{algo:F4:reducers}.
Hence, we define the following subroutine:
\begin{itemize}
    \item $\mathsf{ConstructReductors}$ which takes as input
    a Gr\"obner trace \(\mathcal{T} = \left(
        \Gamma, \Sigma, \Theta \right) \) with \(\Gamma = \left(
        \Gamma_0, \ldots, \Gamma_\indalgo \right) \) and \(\Theta = \left(
        \Theta_1, \ldots, \Theta_\indalgo \right) \)
         and an integer \( \indalgo \)
     and a set $G$ of polynomials.
    It returns a set of $R$ of polynomials.
    First, it considers each element $(m_1,m_2)$ of $\Theta_{\indalgo}$.
    Since we are in a top-reduction context, it verifies if 
    $m_1m_2 \grs \min_{\grs}\left( \Gamma_{\indalgo} \right)$.
    If it is the case, then it adds to $R$ the element $m_1g$
    where $g$ lies in $G$ and $\lm{}{g} = m_2$. 
\end{itemize}

\begin{algorithm}[ht]
  \algoCaptionLabel{F4T}{f_1, \ldots , f_{\nbpol}, \mathcal{T}}
  \begin{algorithmic}[1]
  \Require{A sequence $\seqpol = (f_1, \ldots , f_{\nbpol})$ of polynomials in $\Kp{n}$ and a Gröbner trace $\mathcal{T}$.}
  \Ensure{A set of polynomials $G$ if it does not fail.}
  \State $\indalgo = 0$  
  \State $\mathscr{E} = \ech{ \mac{}{}{\seqpol}{\supo{\seqpol}}{}{}}$ \label{algo:F4-UsingTracer_before_for}
  \State $G = \polrows{\supo{\seqpol}}{\mathscr{E}}$  
  \State $\indalgo_{\max} = \max \left\{ \indalgo \in \mathbb{N} \mid \Sigma_\indalgo \neq \emptyset \right\}$ 
  \For{$\indalgo = 1$ to $\indalgo = \indalgo_{\max}$}\label{algo:F4-UsingTracer_for}   
    \If{$\Sigma_\indalgo \neq \emptyset$} 
      \State $\explpairs_\indalgo = \Sel{\explpairs}$ 
      \State $\explpairs = \explpairs \setminus \explpairs_\indalgo$ 
      \State $\explpairs_\indalgo = \{ \pair{g_i}{g_j} \in \explpairs_\indalgo \mid (i,j) \notin \mathcal{B}(G) \}$ \label{algo:F4-UsingTracer_critere}
      \State $\explpairs_\indalgo = \mathsf{RemoveUselessPairs}(\mathcal{T},\explpairs_\indalgo,\indalgo)$
      \State $S_\indalgo = \{ \Spol{f}{g} \mid \pair{f}{g}  \in \explpairs_\indalgo \}$ \label{algo:F4-UsingTracer_constpairs}
      \State $R_\indalgo = \mathsf{ConstructReductors}(\mathcal{T},G,\indalgo)$
      \State \label{algo:F4-UsingTracer_S} $ \mathscr{S} = \mac{}{}{S_\indalgo}{\supo{S_\indalgo \cup R_\indalgo}}{}{}$
      \State $ \mathscr{R} =  \mac{}{}{R_\indalgo}{\supo{S_\indalgo \cup R_\indalgo}}{}{}$
      \State $\seqmono = \supo{S_\indalgo \cup R_\indalgo}$
      \State \label{algo:F4-UsingTracer_Reduction}  $\begin{pNiceArray}[margin]{c}
      \mathscr{R}\\
      \hline
      \mathscr{N}
      \end{pNiceArray} =
      \mathsf{Reduction}(\mathscr{S},\mathscr{R},\seqmono)$
    \State $\mathbb{L} = \{\}$
    \For{$f$ in $\polrows{\seqmono}{\mathscr{N}}$}
      \State $\mathbb{L} = \mathbb{L} \cup \{ \lm{}{f} \}$
      \State $G = G \cup \{ f\}$
    \EndFor
    \If{ $\mathbb{L} \ne \mathcal{T}_{\indalgo}$}
      \State \Return False
    \EndIf
    \EndIf 
    \EndFor
  \State \Return $G$
  \end{algorithmic}
\end{algorithm}

We can now introduce \algoName{algo:F4T} that uses the Gr\"obner trace 
associated to some polynomial sequence \(\seqpol\), in order
to compute Gr\"obner bases for polynomial systems that are compatible with the trace, but avoiding 
all reductions to zero and useless reductors.  

The correctness of \algoName{algo:F4T} is immediate and already
commented in \cite{berthomieu}. 

\subsection{Details on the reduction algorithm}
\label{sec:reduction}

Consider the following matrix:
\[
M =
\begin{pNiceArray}[margin]{c}
    \mathscr{R} \\
   \cline{1-1}
   \mathscr{S}
\end{pNiceArray}
=
\begin{pNiceArray}[margin]{*{3}{>{\hspace{0.8cm}}c<{\hspace{0.8cm}}}}
   \Block{3-3}{\mac{}{}{\explred}{\supo{\expluni}}{}{}} & & \\
   & & \\
   & & \\
   \cline{1-3}
   \Block{3-3}{\mac{}{}{\explSp}{\supo{\expluni}}{}{}} & & \\
   & & \\
   & &
\end{pNiceArray}.
\]
\algoName{algo:Reduction} in a top-reduction context consists in computing a row echelon
form of the matrix $M$. We recall that the matrix $\mathscr{R}$ is already
in row echelon form. Hence, $\mathscr{R}$ will be a submatrix of the
output of this subroutine.

For a matrix $M$, a set $J_1$ of row indices and a set $J_2$ of column
indices, we denote by
$M_{J_1,J_2}$ the associated submatrix. We also denote by 
$M_{\ast,J_2}$ (resp.\ $M_{J_1,\ast}$) the submatrix with all columns 
(resp.\ rows) indexed by \(J_2\) (resp.\ \(J_1\)).  
Let $\mathscr{S}$ and $\mathscr{R}$ be two matrices with entries in $\K$ and the same number of columns.
Suppose that $\mathscr{R}$ is in row-echelon form with no rows of zeroes. 
Then, there exist unique matrices $\mathscr{N}$ and $\mathscr{Q}$ such that
\[\mathscr{S} = \mathscr{Q}\mathscr{R} + \mathscr{N}\] 
and the entries in $\mathscr{N}$ in the same columns of pivots of
$\mathscr{R}$ are zeros.
This matrix $\mathscr{N}$ is called the \emph{normal form} of
$\mathscr{S}$ with respect 
to $\mathscr{R}$ and is denoted by 
\[ \mathscr{N} = \NF{\mathscr{S}}{\mathscr{R}}.\]
Write $J$ for the set of indices of the columns in $\mathscr{R}$ that contain
the pivots used for the reduction.
If we assume that \(\mathscr{R}\) has full row rank, the submatrix
\(\mathscr{R}_{\ast,J}\) is invertible.
One then easily verifies that $\mathscr{Q}$ equals $\mathscr{S}_{\ast,J}
(\mathscr{R}_{\ast,J})^{-1}$.
We use the subroutine $\mathsf{NormalForm}$ in \algoName{algo:Reduction},
which takes as input two matrices $\mathscr{S}$ and $\mathscr{R}$ having
the same number of columns, with $\mathscr{R}$ in row echelon form
and which returns $\NF{\mathscr{S}}{\mathscr{R}}$.

We describe now \algoName{algo:Reduction} which computes 
such a normal form \(\mathscr{N}\). We start by introducing some
notations.
Consider a matrix $M$, a set $\seqmono$ of monomials whose cardinality is
the number of columns of $M$, and a subset $\seqmonot{1}$ of $\seqmono$.
The set $\seqmono$ is ordered with respect to the grevlex ordering, and
each column of the matrix $M$ is indexed by an element of $\seqmono$ in
decreasing order.
We denote by $\selered{M}{\seqmono}{\seqmonot{1}}$ the submatrix of $M$
consisting of the rows whose pivots lie in columns indexed by elements
of $\seqmonot{1}$.

The goal is now to perform a reduction using row submatrices whose pivots
are selected from certain subsets of $\seqmono$.
We introduce the following notation to describe these subsets.
Let $i$ be an integer in $\{1, \ldots, n\}$, and let
$d = \max_{m \in \seqmono} \deg(m)$
be the maximal degree of the elements of $\seqmono$.
We denote by $\redensmon{\seqmono}{i}$ the subset
$\seqmono \cap \left( \mon{d,}{i} \setminus \mon{d,}{i-1} \right)$.
Moreover, we denote by $\redensmon{\seqmono}{0}$ the subset
$\seqmono \cap \mon{\leq d-1}{}$.
Observe that
$\seqmono = \bigsqcup_{i=0}^{n} \redensmon{\seqmono}{i}$,
except in the case $\seqmono = \{1\}$, which never occurs in our setting.
We define the subroutine $\mathsf{Modify}$ that takes as input a matrix $M$,
a set $\seqmono$ of monomials whose cardinality is
the number of columns of $M$, a subset $\seqmonot{1}$ of $\seqmono$,
and a matrix $N$ that has the same dimensions as 
$\selered{M}{\seqmono}{\seqmonot{1}}$. It returns the matrix $M$ where 
its submatrix $\selered{M}{\seqmono}{\seqmonot{1}}$ has been replaced by $N$.

\begin{algorithm}[ht]
    \algoCaptionLabel{Reduction}{} 
    \begin{algorithmic}[1]
        \Require{Two matrices $\mathscr{S}$ and $\mathscr{R}$ 
            with entries in $\K$, the same number of columns and $R$ is in row 
  echelon form. A set $\seqmono$ of  monomials which cardinality is equal to
  the number of columns of $\mathscr{S}$ and $\mathscr{R}$.}
  \Ensure{A row echelon form of the matrix obtained by combining the rows of $S$ and $R$.}
  \State $\mathscr{S}' = \mathscr{S}$
  \For{$i = 1$ to $i = n$}\label{algorithm_reduction_for}  
    \State $\seqmonot{i} = \redensmon{\seqmono}{i}$
    \State $N = \normfsubr{\selered{\mathscr{S}'}{\seqmono}{\seqmonot{i}}}{\selered{\mathscr{R}}{\seqmono}{\seqmonot{i}}}$ \label{algorithm_reduction_red}
    \State $ \mathscr{S}' = \modify{\mathscr{S}'}{\seqmono}{\seqmonot{i}}{N}$
    \State $E = \ech{\selered{\mathscr{S}'}{\seqmono}{\seqmonot{i}}}$ \label{algorithm_reduction_ech}
    \State $ \mathscr{S}' = \modify{\mathscr{S}'}{\seqmono}{\seqmonot{i}}{E}$ 
  \EndFor
  \State $\seqmonot{0} = \redensmon{\seqmono}{0}$
    \State $N = \normfsubr{\selered{\mathscr{S}'}{\seqmono}{\seqmonot{0}}}{\selered{\mathscr{R}}{\seqmono}{\seqmonot{0}}}$ \label{algorithm_reduction_redap}
    \State $ \mathscr{S}' = \modify{\mathscr{S}'}{\seqmono}{\seqmonot{0}}{N}$
    \State $E = \ech{\selered{\mathscr{S}'}{\seqmono}{\seqmonot{0}}}$ \label{algorithm_reduction_echap}
    \State $ \mathscr{S}' = \modify{\mathscr{S}'}{\seqmono}{\seqmonot{0}}{E}$
  \State \Return $\begin{pNiceArray}[margin]{c}
   \mathscr{R}\\
   \hline 
   \mathscr{S}'
\end{pNiceArray}.$
  \end{algorithmic}
\end{algorithm}

In order to prove the correctness of \algoName{algo:Reduction}, we establish the 
following result.

\begin{lemma}\label{lemma_red_redech}
    Let $\mathscr{S}$ and $\mathscr{R}$ be two matrices with entries in $\K$
    and the same number of columns. Suppose that $\mathscr{R}$ is in 
    row echelon form. Consider the matrix $\mathscr{E}$ to be a row echelon 
    form of $\mathscr{N} = \NF{\mathscr{S}}{\mathscr{R}}$.

    Then the matrix $\begin{pNiceArray}[margin]{c}
   \mathscr{R}\\
   \hline 
   \mathscr{E}
\end{pNiceArray}$
is a row echelon form of $\begin{pNiceArray}[margin]{c}
   \mathscr{R}\\
   \hline 
   \mathscr{S}
\end{pNiceArray}$ up to a permutation of its rows. 
\end{lemma}

\begin{proof}
Let $J_{\mathscr{R}}$ (resp.\ $J_{\mathscr{E}}$) denote the set of indices of
columns containing pivots of $\mathscr{R}$ (resp.\ $\mathscr{E}$).
We prove that the matrix $\begin{pNiceArray}[margin]{c}
   \mathscr{R}\\
   \hline 
   \mathscr{E}
\end{pNiceArray}$
is in row echelon form up to a permutation of its rows.
First, note that the matrices $\mathscr{R}$ and $\mathscr{E}$ are in row
echelon form.
Thus, it suffices to prove that for each $j \in J_{\mathscr{R}}$, the
$j$-th column of $\mathscr{E}$ does not contain a pivot.

By definition, the $j$-th column of $\mathscr{N}$ is a column of zeroes.
Since the rows of the matrix $\mathscr{E}$ are linear combinations of the
rows of $\mathscr{N}$, the $j$-th column of $\mathscr{E}$ is also a
column of zeroes.
We deduce that the matrix $\begin{pNiceArray}[margin]{c}
   \mathscr{R}\\
   \hline 
   \mathscr{E}
\end{pNiceArray}$
is in row echelon form up to a permutation of its rows.

Since the two matrices $\mathscr{S}$ and $\mathscr{E}$ have the same number of rows, 
it remains to prove that the rows of the matrices $\begin{pNiceArray}[margin]{c}
   \mathscr{R}\\
   \hline 
   \mathscr{E}
\end{pNiceArray}$ and $\begin{pNiceArray}[margin]{c}
   \mathscr{R}\\
   \hline 
   \mathscr{S}
\end{pNiceArray}$
generate the same $\K$-vector space.
Since the rows of the matrices $\mathscr{N}$ and $\mathscr{E}$
generate the same $\K$-vector space, rows of the matrices
$\begin{pNiceArray}[margin]{c}
   \mathscr{R}\\
   \hline 
   \mathscr{N}
\end{pNiceArray}$ and $\begin{pNiceArray}[margin]{c}
   \mathscr{R}\\
   \hline 
   \mathscr{E}
\end{pNiceArray}$
generate the same $\K$-vector space.
We consider the transformation matrix $\mathscr{Q}$ such that 
$\mathscr{S} = \mathscr{Q}\mathscr{R} + \mathscr{N}$.
We then have the following expression 
\[
\begin{pNiceArray}[margin]{c}
   \mathscr{R}\\
   \hline 
   \mathscr{N}
\end{pNiceArray} =
\begin{pNiceArray}[margin]{c|c}
   \mathscr{I} &  0\\
   \hline 
   -\mathscr{Q} & \mathscr{I}
\end{pNiceArray} 
\begin{pNiceArray}[margin]{c}
   \mathscr{R}\\
   \hline 
   \mathscr{S}
\end{pNiceArray}
\]
with $\mathscr{I}$ designates the identity matrix of the 
conformal size. 
As this transformation matrix is invertible, this implies that the rows of the matrices 
$\begin{pNiceArray}[margin]{c}
   \mathscr{R}\\
   \hline 
   \mathscr{S}
\end{pNiceArray}$
and 
$\begin{pNiceArray}[margin]{c}
   \mathscr{R}\\
   \hline 
   \mathscr{N}
\end{pNiceArray}$
generate the same $\K$-vector space.
This concludes the proof.
\end{proof}

\begin{lemma}\label{lemma:reduction:correctness}
    On input $\mathscr{R}$ and $\mathscr{S}$,
    \algoName{algo:Reduction} returns the matrix $\begin{pNiceArray}[margin]{c}
   \mathscr{R}\\
   \hline 
   \mathscr{E}
\end{pNiceArray}$
which is a row echelon form of $\begin{pNiceArray}[margin]{c}
   \mathscr{R}\\
   \hline 
   \mathscr{S}
\end{pNiceArray}$.   
\end{lemma}

\begin{proof}
There are $n$ steps in the loop of \algoName{algo:Reduction} in \cref{algorithm_reduction_for}.
The step after the loop corresponds to
\cref{algorithm_reduction_redap} and \cref{algorithm_reduction_echap} and
consists of the same operations as the previous steps but is out of the loop 
for notational reasons.

Let us denote $\selered{\mathscr{S}}{\seqmono}{\seqmonot{1}}$ by
$\mathscr{S}_1$ and $\selered{\mathscr{R}}{\seqmono}{\seqmonot{1}}$ by
$\mathscr{R}_1$. We proceed through the first step and then conclude by
iteration.

Consider the input matrix formed by the rows of $\mathscr{S}$ and
$\mathscr{R}$. We can reorder these rows as needed to better visualize
what is happening in \algoName{algo:Reduction}
\bigskip

\[ \mathscr{M}_1 = \begin{pNiceArray}[margin]{ccc|c|c|c}[last-col]
    0 & \cdots & 0 & ~~~~\ast~~~~ & \cdots & ~~~~\ast~~~~ & \\
    \vdots & \ddots & \vdots & ~~~~\ast~~~~ & \cdots & ~~~~\ast~~~~ & \\
    0 & \cdots & 0 & ~~~~\ast~~~~ & \cdots & ~~~~\ast~~~~ & \\
   \hline
   & ~~~~\ast~~~~& & ~~~~\ast~~~~ & \cdots & ~~~~\ast~~~~ & \}\mathscr{R}_1 \\
   \hline
   & ~~~~\ast~~~~ & & ~~~~\ast~~~~ & \cdots & ~~~~\ast~~~~ & \}\mathscr{S}_1\\
   \CodeAfter
   \OverBrace{1-1}{3-3}{\seqmonot{1}}
   \OverBrace{1-4}{1-4}{\seqmonot{2}}
   \OverBrace{1-6}{1-6}{\seqmonot{0}}
\end{pNiceArray}.\]
During the first iteration of the loop in \cref{algorithm_reduction_for}, the initial matrix
$\mathscr{S}_1$ is transformed into the matrix $\mathscr{E}_1 = \ech{\NF{\mathscr{S}_1}{\mathscr{R}_1}}$.
Since $\mathscr{R}$ is already in row echelon form, the same holds for
$\mathscr{R}_1$. Therefore, one can apply \cref{lemma_red_redech} to deduce that the matrix
$\begin{pNiceArray}[margin]{c}
   \mathscr{R}_1\\
   \hline 
   \mathscr{E}_1
\end{pNiceArray}$ is a row echelon form of the matrix 
$\begin{pNiceArray}[margin]{c}
   \mathscr{R}_1\\
   \hline 
   \mathscr{S}_1
\end{pNiceArray}$. Hence, at the end of the first iteration of the loop at \cref{algorithm_reduction_for},
the matrix $\mathscr{M}_1$ becomes the following matrix $\widetilde{\mathscr{M}_1}$ up to a permutation of its rows
\bigskip

\[ \widetilde{\mathscr{M}_1} = \begin{pNiceArray}[margin]{ccc|c|c|c}[last-col]
    0 & \cdots & 0 & ~~~~\ast~~~~ & \cdots & ~~~~\ast~~~~ & \\
    \vdots & \ddots & \vdots & ~~~~\ast~~~~ & \cdots & ~~~~\ast~~~~ & \\
    0 & \cdots & 0 & ~~~~\ast~~~~ & \cdots & ~~~~\ast~~~~ & \\
   \hline
   & ~~~~\ast~~~~& & ~~~~\ast~~~~ & \cdots & ~~~~\ast~~~~ & \}\mathscr{R}_1 \\
   \hline
   & ~~~~\ast~~~~ & & ~~~~\ast~~~~ & \cdots & ~~~~\ast~~~~ & \}\mathscr{E}_1\\
   \CodeAfter
   \OverBrace{1-1}{3-3}{\seqmonot{1}}
   \OverBrace{1-4}{1-4}{\seqmonot{2}}
   \OverBrace{1-6}{1-6}{\seqmonot{0}}
\end{pNiceArray} \]
with $\begin{pNiceArray}[margin]{c}
   \mathscr{R}_1\\
   \hline 
   \mathscr{E}_1
\end{pNiceArray}$ being a row echelon form of
$\begin{pNiceArray}[margin]{c}
   \mathscr{R}_1\\
   \hline 
   \mathscr{S}_1
\end{pNiceArray}$.
Thus, the rows of $\mathscr{M}_1$ and $\widetilde{\mathscr{M}_1}$ generate the same
$\K$-vector space.

There are two types of rows in $\mathscr{E}_1$: the rows whose pivots lie in columns
indexed by elements of $\seqmonot{1}$, and the remaining rows. These two types of rows
correspond to the matrices $\mathscr{E}_{1,1}$ and $\mathscr{E}_{1,2}$, respectively.
By suitably permuting the rows of $\widetilde{\mathscr{M}_1}$, one obtains the following matrix:
\bigskip

\[ \begin{pNiceArray}[margin]{ccc|c|c|c}[first-col,last-col]
   \Block{4-1}{ \mathscr{M}_2 \scalebox{1.6}{\Bigg\{}  } &0 & \cdots & 0 & ~~~~\ast~~~~ & \cdots & ~~~~\ast~~~~ & \\
    &\vdots & \ddots & \vdots & ~~~~\ast~~~~ & \cdots & ~~~~\ast~~~~ & \\
    &0 & \cdots & 0 & ~~~~\ast~~~~ & \cdots & ~~~~\ast~~~~ & \\
   \hline
    && 0 & & ~~~~\ast~~~~ & \cdots & ~~~~\ast~~~~ & \}\mathscr{E}_{1,2}\\
   \hline
    && ~~~~\ast~~~~& & ~~~~\ast~~~~ & \cdots & ~~~~\ast~~~~ & \}\mathscr{R}_1 \\
   \hline
   & & ~~~~\ast~~~~ & & ~~~~\ast~~~~ & \cdots & ~~~~\ast~~~~ & \}\mathscr{E}_{1,1}\\
   \CodeAfter
   \OverBrace{1-1}{3-3}{\seqmonot{1}}
   \OverBrace{1-4}{1-4}{\seqmonot{2}}
   \OverBrace{1-6}{1-6}{\seqmonot{0}}
\end{pNiceArray} \]
with $\begin{pNiceArray}[margin]{c}
   \mathscr{R}_1\\
   \hline 
   \mathscr{E}_{1,1}
\end{pNiceArray}$ being in row echelon form with its pivots in columns indexed
by elements of $\seqmonot{1}$.
The second iteration of the loop in \cref{algorithm_reduction_for} transforms only
the matrix $\mathscr{M}_2$ and considers the columns indexed by $\seqmonot{2}$. 
Since each step of \cref{algorithm_reduction_for} deals with a group of columns, 
and since the union of all these groups eventually covers all the columns of $\mathscr{M}_1$, 
the algorithm returns a row echelon form of $\mathscr{M}_1$.
\end{proof}

\begin{remark}
\label{remark_degree_last_lines}
Let $d$ denote the highest degree of the elements in the set $\seqmono$.
Suppose that all the pivots of the matrix $\mathscr{R}$ lie in columns
indexed by monomials of degree $d$. Then
\cref{algorithm_reduction_redap} and \cref{algorithm_reduction_echap} of
\algoName{algo:Reduction} do not occur.
\end{remark}

The next statement provides a complexity bound for the computation of such a
normal form. Its proof relies on 
\cite[Lemma~8.8]{Gopalakrishnan2025} and indirectly on
\cite[Lemma~3.1]{DumasGiorgiPernet2004}.
Recall that $\omega$ denotes the constant of matrix multiplication
over \(\K\).

\begin{prop}
\label{prop_complexity_normal_form}
Let $M_1,M_2,C_1,C_2$ be matrices that form the block matrix
\[ \begin{pNiceArray}[margin]{c|c}
    M_1& M_2   \\ 
    \hline
    C_1 & C_2 
    \end{pNiceArray}, \]
where $M_1$ is an upper triangular invertible matrix in $\K^{ a \times a}$, the matrix $M_2$ is in $\K^{ a \times b}$,
the matrix $C_1$ is in $\K^{ c \times a}$ and the matrix $C_2$ is in $\K^{ c \times b}$.
Computing the Schur complement
\[ \NF{\begin{pNiceArray}[margin]{c|c}
    C_1 & C_2 
    \end{pNiceArray}}{\begin{pNiceArray}[margin]{c|c}
    M_1& M_2
    \end{pNiceArray}} = \begin{pNiceArray}[margin]{c|c}
    0 & C_2 - C_1M_1^{-1}M_2  
    \end{pNiceArray} \]
takes:
\begin{itemize}
\item $\mathcal{O} \left( a^2c^{\omega - 2} + \frac{cab}{\min(c,b)^{3 - \omega}} \right)$ operations in $\K$ if $c \leq a$,
\item $\mathcal{O} \left( ca^{\omega - 1} + \frac{cab}{\min(a,b)^{3 - \omega}} \right)$ operations in $\K$ if $a \leq c$.
\end{itemize}
\end{prop}

\begin{proof}
The arithmetic complexity of computing such a normal form is the same
as the one of computing the product $C_1M_1^{-1}M_2$.
We decompose this product as $(C_1M_1^{-1})M_2$.
The complexity to compute the product of $(C_1M_1^{-1})$ by $M_2$
is equal to $\mathcal{O} \left( \frac{acb}{\min(a,c,b)^{\omega}} \right)$.
Moreover, since $M_1$ is an upper triangular invertible matrix, one can apply
\cite[Lemma~8.8]{Gopalakrishnan2025} to establish that the arithmetic complexity 
of computing $C_1M_1^{-1}$ is $\mathcal{O} \left( a^2c^{\omega - 2} \right)$
if $c \leq a$ and $\mathcal{O} \left( ca^{\omega - 1} \right)$ if $a \leq c$.
\end{proof}

\subsection{Structural properties}
\label{ssec:algo:properties}

Further, we assume that both the number of variable $n$ and the 
maximum degree $\delta$ of the input
polynomials are greater than or equal to $2$.
We denote by $\gbalgo{0}$ the set $G$ in \algoName{algo:F4B} before the while loop at 
\cref{algo:F4B:while}.

\begin{lemma}
\label{lemma_icreasing_degree_in_algorithm_vector_space}
Let $\seqpol = (f_1, \ldots ,f_\nbpol)$ be a sequence of nonzero
polynomials in the ring $\Kp{n}$ such that $(\hd{f_1}, \ldots ,
\hd{f_\nbpol})$ lies in the nonempty Zariski open set 
$\zarisemi{\delta}{n}{\nbpol}$ (introduced in
\cref{subsectionThelastdimensionofthegeneriquemonomialstaircase}) 
with $\nbpol
\leq n$. For all $\degalgo \geq \delta$, set 
\[ W_\degalgo = \vectorspace{\{mf_i \mid i \in \{1, \ldots , \nbpol\}, m \in \mon{\leq \degalgo-\delta,}{n} \}}{\K}. \]
Then, in \algoName{algo:F4B}, 
for all polynomials $g$ in $\gbalgo{0}$, it holds that $\deg(g) = \delta$ and
the following equality holds:
\[ W_\degalgo = \vectorspace{\{mg \mid g \in \gbalgo{0}, m \in \mon{\leq \degalgo-\delta,}{n} \}}{\K} \]
for all $\degalgo \geq \delta$.
\end{lemma}

\begin{proof}
Let $\mathscr{M} = \mac{}{}{\seqpol}{\supo{\seqpol}}{}{}$ be the Macaulay matrix
associated to $\seqpol$ with respect to the grevlex ordering, and
consider its row reduced echelon form $\mathscr{E} = \ech{M}$. Since by definition we have
that $\gbalgo{0} = \polrows{\supo{\seqpol}}{\mathscr{E}}$, it holds that
$W_{\delta} = \vectorspace{\gbalgo{0}}{\K}$.  

Let $g \in \gbalgo{0}$. 
Note that, by construction, $g \neq 0$.
Then, there exist coefficients $q_i \in \K$ such that
\[
  g = \sum_{i = 1}^\nbpol q_i f_i.
\]
Since $g$ is a linear combination of polynomials of degree $\delta$, it follows that 
$\deg(g) \leq \delta$. Suppose, for the sake of contradiction, that $\deg(g) < \delta$.
Because $(\hd{f_1}, \ldots, \hd{f_\nbpol})$ lies in $\zarisemi{\delta}{n}{\nbpol}$ with $\nbpol \leq n$ and 
$\deg(q_i) \leq 0$, we can apply \cref{prop_no_degree_fall}.  
This yields the existence of polynomials $(h_i)_{i \in \{1, \ldots , \nbpol\}}$ in $\Kp{n}$ 
with $\deg(h_i) < 0$ for all $i$ such that
\[
  g = \sum_{i = 1}^\nbpol h_i f_i.
\]
Clearly, this forces $h_i = 0$ for every $i$ in $\{1, \ldots ,
\nbpol\}$, which implies $g = 0$.  Since $g$ is nonzero, we reach a
contradiction. Hence, it holds that $\deg(g) = \delta$.  

Now, let $\degalgo \geq \delta$ be an integer and $m \in \mon{\leq \degalgo-\delta,}{n}$.  
We can write
\[
  mg = m \sum_{i=1}^\nbpol q_i f_i,
\]
so that $mg \in W_\degalgo$. The reverse inclusion follows by the same reasoning.  
Therefore, we conclude that
\[
  W_\degalgo = \vectorspace{\{\, mg \mid g \in \gbalgo{0}, \, m \in \mon{\leq \degalgo-\deg(g),}{n} \}}{\K}.
\qedhere
\]
\end{proof}

For any integer $\indalgo \geq 1$, let $\gbalgo{\indalgo}$ denote the value of $G$ at the beginning of the
$(\indalgo+1)$-th iteration of the while loop in
\cref{algo:F4B:while}, precisely the value in effect when
\cref{algo:F4-BuildTracer_index_d} of \algoName{algo:F4B} is executed.
As \algoName{algo:F4B} is only adding elements in $G$, then
for all $\indalgo$ in $\mathbb{N}$ we have that $\gbalgo{\indalgo}
\subset \gbalgo{\indalgo+1}$.  
We set by convention that $\gbalgo{-1} = \emptyset$.
Also, for any integer $\indalgo \geq 1$, let $\allpalgo{\indalgo}$ denote the value of $P$ during the
$\indalgo$-th iteration of the while loop in \cref{algo:F4B:while}, precisely the value in effect when step \cref{algo:F4-BuildTracer_index_d} is executed.
For all $\indalgo \geq 1$, we denote the minimal degree of the current pairs by
$\curdeg{\indalgo} = \min_{p \in \allpalgo{\indalgo}}(\deg(p))$. We
set $\curdeg{0} = \delta$ and $\curdeg{-1} = -1$.

\begin{lemma}
\label{lemma_icreasing_degree_in_algorithm}
Let $\seqpol = (f_1, \ldots ,f_{\nbpol})$ be a sequence of polynomials
in the ring $\Kp{n}$ such that its associated sequence  
$(\hd{f_1}, \ldots , \hd{f_\nbpol})$
lies in $\zarisemi{\delta}{n}{\nbpol}$ with $\nbpol \leq n$.
The following
assertions hold for all $\indalgo \geq 0$. 
\begin{enumerate}[leftmargin=!, labelwidth=0.8cm, align=parleft]
\item $\asseralgo{\indalgo}{1}$:
    \label{lemma_icreasing_degree_in_algorithm_item1}
    $\curdeg{\indalgo} > \curdeg{\indalgo-1}$.
\item $\asseralgo{\indalgo}{2}$:
    \label{lemma_icreasing_degree_in_algorithm_item2} For all $g \in
    G_\indalgo \backslash \gbalgo{\indalgo-1},~~ \deg(g) =
    \curdeg{\indalgo}$.
\item $\asseralgo{\indalgo}{3}$:
    \label{lemma_icreasing_degree_in_algorithm_item3} For all $ 
    (g_1,g_2) \in \gbalgo{\indalgo}\times \gbalgo{\indalgo} \text{ with } g_1 \ne g_2,
    \lm{\grs}{g_1} \text{ does not divide } \lm{\grs}{g_2}.$
\item $\asseralgo{\indalgo}{4}$:
    \label{lemma_icreasing_degree_in_algorithm_item4} For all $p$ in
    $\allpalgo{\indalgo+1}$, $\deg(p) > \curdeg{\indalgo}$.
\item $\asseralgo{\indalgo}{5}$:
    \label{lemma_icreasing_degree_in_algorithm_item5} The equality
    $V_{\indalgo} = \widetilde{V}_{\indalgo}$ holds where:
    \[
 V_{\indalgo} = \vectorspace{\{mf_i \mid i \in \{1, \ldots , \nbpol\},
m \in \mon{\leq \curdeg{\indalgo}-\delta,}{n} \}}{\K}\] 
and 
\[
    \widetilde{V}_{\indalgo} = \vectorspace{\{mg \mid g \in
    \gbalgo{\indalgo}, m \in \mon{\leq \curdeg{\indalgo}-\deg(g),}{n}
\}}{\K}.
\]
\end{enumerate}
\end{lemma}

\begin{proof}
Our proof is by induction on $\indalgo$.

We start to prove that $\asseralgo{0}{1}$, $\asseralgo{0}{2}$,
$\asseralgo{0}{3}$, $\asseralgo{0}{4}$, and $\asseralgo{0}{5}$, hold.
First, since $\delta$ is positive we have $\asseralgo{0}{1}$.  The set
$\gbalgo{0} \setminus \gbalgo{-1}$ is equal to $\gbalgo{0}$ by
definition. Since the sequence $(\hd{f_1}, \ldots , \hd{f_\nbpol})$
lies in $\zarisemi{\delta}{n}{\nbpol}$ with $\nbpol \leq n$, one can
apply \cref{lemma_icreasing_degree_in_algorithm_vector_space}.  It
follows that, for all $g$ in $\gbalgo{0}$, we have that $\deg(g) =
\delta$. This proves $\asseralgo{0}{2}$. 

Moreover, we know that there exists a matrix $\mathscr{E}$ with entries in $\K$ in row
echelon form such that $\gbalgo{0} = \polrows{\supo{\seqpol}}{\mathscr{E}}$. Since the leading
monomial of a polynomial represented by a row of $\mathscr{E}$ is determined by
the column of its pivot, we can not have twice the same leading
monomial in $\gbalgo{0}$.  In addition to that, using
$\asseralgo{0}{1}$, for all $g$ in $\gbalgo{0}$ we have $\lm{\grs}{g}
= \delta$. Thus, all the elements of $\lm{\grs}{\gbalgo{0}}$ are
different and have the same degree. This establishes $\asseralgo{0}{3}$.

Consider a pair $p$ in $\allpalgo{1}$, then we have that $p =
\pair{g_1}{g_2}$ where $g_1$ and $g_2$ lie in $\gbalgo{0}$.
 We write the monomial $t = \operatorname{lcm} \left(
 \lm{\grs}{g_1},\lm{\grs}{g_2} \right)$. We know that $\deg(p) =
 \deg(t)$.  Since $\lm{\grs}{g_1}$ and $\lm{\grs}{g_2}$ are different
 and have degree $\delta$, the monomial $t$ has degree greater than
 $\delta$. This gives $\asseralgo{0}{4}$.

Since $\curdeg{0} = \delta$, remark that $\asseralgo{0}{5}$ is a
particular case of
\cref{lemma_icreasing_degree_in_algorithm_vector_space} $(\degalgo =
\delta)$.

\smallskip

We proceed now with the induction. Let $\indalgo$ be a positive
integer and  
assume that $\asseralgo{\tilde{\indalgo}}{1}$,
$\asseralgo{\tilde{\indalgo}}{2}$,
$\asseralgo{\tilde{\indalgo}}{3}$,$\asseralgo{\tilde{\indalgo}}{4}$,
and $\asseralgo{\tilde{\indalgo}}{5}$ hold for all $\tilde{\indalgo} <
\indalgo$.  We prove that $\asseralgo{\indalgo}{1}$,
$\asseralgo{\indalgo}{2}$, $\asseralgo{\indalgo}{3}$,
$\asseralgo{\indalgo}{4}$, and $\asseralgo{\indalgo}{5}$ hold.

\medskip

First, we recall that $\curdeg{\indalgo} = \min_{p \in
\allpalgo{\indalgo}}(\deg(p))$. By $\asseralgo{\indalgo-1}{4}$, we
have that $\curdeg{\indalgo}> \curdeg{\indalgo-1}$ which is exactly
$\asseralgo{\indalgo}{1}$.

We prove now that $\asseralgo{\indalgo}{2}$ holds by showing that for
all $g$ in $\gbalgo{\indalgo} \setminus \gbalgo{\indalgo-1}$, we have
 $\deg(g) = \curdeg{\indalgo}$. Let $g'$ be in $\gbalgo{\indalgo}
\setminus \gbalgo{\indalgo-1}$.  By definition of
$\Sel{\allpalgo{\indalgo}}$, we have that all pairs in $P_\indalgo$
are all of degree $\curdeg{\indalgo}$.  We deduce that $\deg(g') \leq
\curdeg{\indalgo}$. Suppose by contradiction that $\deg(g') <
\curdeg{\indalgo}$. 

First, we prove that $g'$ can be expressed as follows $g' = \sum_{g \in
\gbalgo{\indalgo-1}} u_g g$ where $u_g$ is a polynomial with
$\deg(u_g) \leq \curdeg{\indalgo}-\deg(g)$.

Since $g'$ lies in $\gbalgo{\indalgo} \setminus
\gbalgo{\indalgo-1}$, it can be seen as a row of the echelon form of the
matrix $\begin{pNiceArray}[margin]{c}
   \mathscr{R}\\
   \hline 
   \mathscr{S}
\end{pNiceArray}$ built at \cref{algo:F4-BuildTracer_S} and \cref{algo:F4-BuildTracer_R}.
Thus, the polynomial $g'$ is a $\K$-linear combination of the polynomials 
that are represented by the rows of $\mathscr{S}$ and $\mathscr{R}$.
The rows of $\mathscr{S}$ represent $S$-polynomials from pairs of $P_{\indalgo}$.
As previously said, these pairs have degree $\curdeg{\indalgo}$ and are built
using two polynomials of $\gbalgo{\indalgo-1}$. We deduce that the polynomials 
that are represented by the rows of $\mathscr{S}$ lie in the 
$\K$-vector space generated by 
$\{ mg \mid g \in \gbalgo{\indalgo-1}, m \in \mon{\leq \curdeg{\indalgo} - \deg(g)}{} \}$.
Moreover, by definition of $\mathsf{SymbolicPreprocessing}$, the 
rows of $\mathscr{R}$ are of the form $mg$ with $g$ in $\gbalgo{\indalgo-1}$ and
$m$ in $\mon{\leq \curdeg{\indalgo} - \deg(g)}{}$.

We deduce that $g'$ can be expressed as follows $g' = \sum_{g \in
\gbalgo{\indalgo-1}} u_g g$ where $u_g$ is a polynomial with
$\deg(u_g) \leq \curdeg{\indalgo}-\deg(g)$.

Moreover, one can decompose $g'$ as follows:
\begin{align*}
g' &= \sum_{g \in \gbalgo{\indalgo-1}} \left( \sum_{m \in \supo{u_g}}
\lambda_{m,g} mg  \right) \\
&= \sum_{g \in \gbalgo{\indalgo-1}} \left(
\sum_{\substack{m \in \supo{u_g} \\ \deg(m) \leq
\curdeg{\indalgo-1}-\deg(g) }} \lambda_{m,g} mg + \sum_{\substack{ m
\in \supo{u_g} \\ \deg(m) > \curdeg{\indalgo-1}-\deg(g) }}
\lambda_{m,g} mg  \right).
\end{align*}

By $\asseralgo{\indalgo-1}{1}$, we have $\curdeg{\indalgo}>
\curdeg{\indalgo-1}$.  Then, we deduce that $\curdeg{\indalgo} - \deg(g)>
\curdeg{\indalgo-1} - \deg(g)$.  Thus, one can rewrite 
\begin{align*}
g' = p_1 + \sum_{\substack{m \in \mon{\leq
\curdeg{\indalgo}-\curdeg{\indalgo-1},}{n} \\ m \ne 1}} mp_m = \sum_{m
\in \mon{\leq \curdeg{\indalgo}-\curdeg{\indalgo-1},}{n}} mp_m 
\end{align*}
where $p_m$ lies in $\widetilde{V}_{\indalgo-1}$ for all $m$ in
$\mon{\leq \curdeg{\indalgo}-\curdeg{\indalgo-1},}{n}$.  Using
$\asseralgo{\indalgo-1}{5}$, the polynomials $p_m$ (for all $m$
in $\mon{\leq \curdeg{\indalgo}-\curdeg{\indalgo-1},}{n}$) lie in
$V_{\indalgo-1}$.  We deduce that $g'$ can be written as: \[ g' = \sum_{i
= 1}^{\nbpol} q_if_i \] with $\deg(q_i) \leq \curdeg{\indalgo} -
\delta$. 

Since the sequence $(\hd{f_1}, \ldots , \hd{f_{\nbpol}})$ lies in
$\zarisemi{\delta}{n}{\nbpol}$ and $\deg(g') < \curdeg{\indalgo}$ by
hypothesis, one can apply \cref{prop_no_degree_fall} to express $g'$
as follows \[ g' = \sum_{i = 1}^{\nbpol} h_if_i \] with $\deg(h_i) <
\curdeg{\indalgo} - \delta$.  By
\cref{lemma_icreasing_degree_in_algorithm_vector_space}, the
polynomial $g'$ lies in
\begin{align*}
 \vectorspace{\{mg \mid g \in \gbalgo{0}, m \in \mon{<
 \curdeg{\indalgo}-\delta,}{n} \}}{\K} \subset \vectorspace{\{mg \mid
g \in \gbalgo{\indalgo-1}, m \in \mon{< \curdeg{\indalgo}
-\deg(g),}{n} \}}{\K} 
\end{align*}
as $\gbalgo{0} \subset \gbalgo{\indalgo-1}$.

From this membership, one can write $g'$ as follows:
\begin{align*}
 g' = \sum_{\substack{g \in \gbalgo{\indalgo-1} \\ m \in \mon{<
 \curdeg{\indalgo} -\deg(g),}{n}} }  c_{m,g}mg \tag{$\ast$}  
\end{align*}
with \(c_{m, g}\in \K \setminus \{0\} \).  Since $g'$ lies in $\gbalgo{\indalgo}
\setminus \gbalgo{\indalgo-1}$, we have that $\lm{\grs}{g'} \notin
\langle \gbalgo{\indalgo-1} \rangle$.  
We use $(\ast)$ to sort the sum and express $g'$ as
follows

\begin{align*}
g' =  \sum_{\upsilon \in \Upsilon} \sum_{j = 0}^{\ell_\upsilon} p_{\upsilon,j} 
\end{align*}
with $p_{\upsilon,j}$ of the form $c_{m,g}mg$ with $g \in
\gbalgo{\indalgo-1}$, $m \in \mon{< \curdeg{\indalgo}
-\deg(g),}{n}$ and $m\lm{\grs}{g} = \upsilon$. Thus, the set of monomials $\Upsilon$
is a subset of $\mon{< \curdeg{\indalgo},}{n}$.
Moreover, we have that $\lm{\grs}{p_{\upsilon,j}} =
\upsilon$.  Also, the set $\Upsilon$ is chosen such that for all $\upsilon$ in $\Upsilon$, we have that $\{0,
\ldots, \ell_{\upsilon}\} \ne \emptyset$. As $g'$ is not zero, let us consider $\upsilon_1 \in
\Upsilon$, the greatest monomial of $\Upsilon$ with respect to the grevlex ordering.
First, using the last expression of $g'$ and the definition of $\upsilon_1$, 
it is obvious that $\upsilon_1 \grseq \lm{\grs}{g'}$. If
$\lm{\grs}{g'} = \upsilon_1$, then we have that $\lm{\grs}{g'}$ lies in
$\lm{\grs}{\gbalgo{\indalgo-1}}$ which is not possible. Thus, we must have
that $\upsilon_1 \grs \lm{\grs}{g'}$. This implies that $\upsilon_1 \grs
\lm{\grs}{\sum_{j = 0}^{\ell_{\upsilon_1} } p_{\upsilon_1,j}}$. Using
\cite[Chapter~2, Section~6, Lemma~5]{cox94}, we deduce that $\sum_{j =
0}^{\ell_{\upsilon_1} } p_{\upsilon_1,j}$ is a $\K$-linear combination of
S-polynomials formed by the $(p_{\upsilon_1,j})_{0 \leq j \leq
\ell_{\upsilon_1}}$. Moreover, each of these S-polynomials have a leading
monomial which is less than $\upsilon_1$ for the grevlex ordering. Let us consider one of
them.  Let $j_1$ and $j_2$ be in $\{ 0 , \ldots , \ell_{\upsilon_1} \}$.
We denote $p_{\upsilon_1,j_1} = c_{m_1,g_1}m_1g_1$ and
$p_{\upsilon_1,j_2} = c_{m_2,g_2}m_2g_2$. Then
$\Spol{p_{\upsilon_1,j_1}}{p_{\upsilon_1,j_2}}$ is of the form $\lambda \mu
\Spol{g_1}{g_2}$ with $\lambda \in \K$, $\mu \in \mon{}{n}$ and
$\mu\operatorname{lcm}(\lm{\grs}{g_1},\lm{\grs}{g_2}) = \upsilon_1$.
Moreover, we have that $\deg(\pair{g_1}{g_2}) < \curdeg{\indalgo}$.
Such a pair has already been reduced in the algorithm. We deduce that  
\begin{align*}
\Spol{g_1}{g_2} = \sum_{\substack{g \in \gbalgo{\indalgo-1}  \\ m \in \mon{}{n} }} b_{m,g}mg
\end{align*}
where $\operatorname{lcm}(\lm{\grs}{g_1},{\lm{\grs}{g_2}}) \grs \lm{\grs}{mg}$ and 
the $b_{m,g}$ lie in $\K$.
We deduce that $\Spol{p_{\upsilon_1,j_1}}{p_{\upsilon_1,j_2}}$ can be expressed as a linear
combination of elements $mg$ where $m$ lies in $\mon{}{n}$, $g$ lies in $\gbalgo{\indalgo-1}$
and where $\upsilon_1 \grs \lm{\grs}{mg}$.

We can now rewrite $g'$ using this new expression of $\sum_{j =
0}^{\ell_{\upsilon_1} } p_{\upsilon_1,j}$. This yields  
\begin{align*}
g'  = \sum_{\upsilon \in \widetilde{\Upsilon}} \sum_{j = 0}^{\ell_\upsilon} \tilde{p}_{\upsilon,j},
\end{align*}
where the greatest monomial of $\widetilde{\Upsilon}$ is less than
$\upsilon_1$ for the grevlex ordering. Knowing that the set of monomials
that are smaller than $\upsilon_1$ is finite, we deduce by repeating this
process that $g'$ must be zero.  Since it is not the case, this is a
contradiction. Thus, we have $\deg(g') = \curdeg{\indalgo}$. 
This proves $\asseralgo{\indalgo}{2}$.

Now, we prove that $\asseralgo{\indalgo}{3}$ holds. Let $g_1$ and $g_2$
in $\gbalgo{\indalgo}$.  If $(g_1,g_2)$ lies in
$\gbalgo{\indalgo-1}^2$ then, by
$\asseralgo{\indalgo-1}{3}$, we deduce that $\lm{\grs}{g_1}$ does not divide
$\lm{\grs}{g_2}$.  Suppose that $(g_1,g_2)$ lies in
$(\gbalgo{\indalgo} \backslash \gbalgo{\indalgo-1})^2$. Since $g_1$
and $g_2$ are represented by different rows of the same matrix 
which is in row echelon form, it
holds that $\lm{\grs}{g_1} \neq \lm{\grs}{g_2}$.  Since they have the
same degree by $\asseralgo{\indalgo}{2}$, the monomial
$\lm{\grs}{g_1}$ does not divide $\lm{\grs}{g_2}$. Finally, suppose
that $g_2$ lies in $\gbalgo{\indalgo} \backslash \gbalgo{\indalgo-1}$
and that $g_1$ lies in $\gbalgo{\indalgo-1}$. By
$\asseralgo{\tilde{\indalgo}}{1}$ and
$\asseralgo{\tilde{\indalgo}}{2}$ for all $\tilde{\indalgo}$ in $\{0,
\ldots , \indalgo\}$, we have that $\deg(g_2) > \deg(g_1)$. Then
$\lm{\grs}{g_2}$ does not divide $\lm{\grs}{g_1}$. Suppose by
contradiction that $\lm{\grs}{g_1}$ divides $\lm{\grs}{g_2}$. Then
the function $\mathsf{SymbolicPreprocessing}$ would have
selected a reductor $mg_1$ for $R_\indalgo$ which is such that
the equality $m\lm{\grs}{g_1} = \lm{\grs}{g_2}$ holds. Since $g_2$ is reduced by $mg_1$
during \cref{algo:F4-BuildTracer_Reduction} by
\algoName{algo:F4B}, this is a contradiction. Then
$\lm{\grs}{g_1}$ does not divide $\lm{\grs}{g_2}$. This proves 
$\asseralgo{\indalgo}{3}$.

In order to prove $\asseralgo{\indalgo}{4}$, consider a pair $p = \pair{g_1}{g_2}$ 
in $\allpalgo{\indalgo+1}$. Let us prove that $\deg(p) >
\curdeg{\indalgo}$. We discuss two different cases whether both $g_1$
and $g_2$ lie in $\gbalgo{\indalgo-1}$ or not.  Suppose that $g_1$ and
$g_2$ lie in $\gbalgo{\indalgo-1}$. Since all pairs of that type and
of degree $\curdeg{\indalgo}$ have been removed at
\cref{algo:F4-BuildTracer_rem_pairs} in
\algoName{algo:F4B}, it follows that $\deg(p) >
\curdeg{\indalgo}$.  Suppose now that $g_1$ lies in $\gbalgo{\indalgo}
\setminus \gbalgo{\indalgo-1}$. Using $\asseralgo{\indalgo}{3}$, we
have \[ \deg \left( \operatorname{lcm} \left(
{\lm{\grs}{g_1}},{\lm{\grs}{g_2}} \right) \right) > \max \left(
\deg(\lm{\grs}{g_1}),\deg(\lm{\grs}{g_2}) \right). \] Since $\deg(g_1)
= \curdeg{\indalgo}$ by $\asseralgo{\indalgo}{2}$, we have
$\deg(p) > \curdeg{\indalgo}$.  This proves $\asseralgo{\indalgo}{4}$.

Finally, we prove $\asseralgo{\indalgo}{5}$.  Using
\cref{lemma_icreasing_degree_in_algorithm_vector_space}, we know that
$V_{\indalgo} \subseteq \widetilde{V}_{\indalgo}$.  Consider $mg$ a
generator of $\widetilde{V}_{\indalgo}$, with $g \in
\gbalgo{\indalgo}$ and $m \in \mon{\leq
\curdeg{\indalgo}-\deg(g),}{n}$.  We separate two cases.  Suppose that
$g$ lies $\gbalgo{\indalgo-1}$, then we have that $\deg(g) \leq
\curdeg{\indalgo-1}$.  The polynomial $mg$ can be written $m_1m_2g$
where $\deg(m_1) \leq \curdeg{\indalgo} - \curdeg{\indalgo-1}$ and
$m_2g$ lies in $\widetilde{V}_{\indalgo-1}$.  Using
$\asseralgo{\indalgo-1}{5}$, we know that $m_2g$ lies in
$V_{\indalgo-1}$. Thus, it is easy to deduce that $m_1m_2g$ lies in
$V_{\indalgo}$.  Suppose now that $g$ lies in $\gbalgo{\indalgo}
\backslash \gbalgo{\indalgo-1}$. Hence $m = 1$ since $\deg(g) =
\curdeg{\indalgo}$ by $\asseralgo{\indalgo}{2}$.  We know that $g$ is
formed by polynomials that are coming from pairs of degree
$\curdeg{\indalgo}$ and the $\mathsf{SymbolicPreprocessing}$
function. Thus, $g$ lies in \[ \vectorspace{\{mg \mid g \in
\gbalgo{\indalgo-1}, m \in \mon{\leq \curdeg{\indalgo}-\deg(g),}{n}
\}}{\K} .\] By the previous case, the polynomial $g$ lies in
$V_{\indalgo}$. This gives $\asseralgo{\indalgo}{5}$ and concludes the proof.
\end{proof}

Further, we investigate properties of Gr\"obner bases (w.r.t.\ the
grevlex ordering) of sequences of length \(n\) of  \(n\)-variate
polynomials of degree \(\delta\) which lie in the nonempty Zariski
open sets $\zarifilling{\delta}{n}{n}$ defined in
\cref{subsection_Onthesyzygiesmodule} and $\zarimore{\delta}{n}$
defined in
\cref{subsectionThelastdimensionofthegeneriquemonomialstaircase}. 
Last, recall that given a polynomial ideal \(I\),
$\bdg{\degalgo,}{\ast}{I}{\grs}$ ($\bdgshort{\degalgo}$ for short, when
there is no ambiguity) denotes the set of polynomials of
degree  \(d\) in the reduced \(\grs\)-Gr\"obner basis. 

\begin{lemma}
\label{lemma_newbdg_each_step}
Let $\seqpol= (f_1, \ldots ,f_n)$ be a sequence that lies in
$\zarifilling{\delta}{n}{n} \cap \zarimore{\delta}{n} \cap
\zarisemi{\delta}{n}{n}$ where $\delta \geq 2$ and $n \geq 2$ and $I =
\langle \seqpol \rangle$ be the ideal it generates. Set $\dregmb =
n(\delta - 1) + 1$.   Then, for each $\degalgo$ in $\{\delta , \ldots
, \dregmb\}$, the set $\bdgshort{\degalgo}$ is nonempty.
\end{lemma}

\begin{proof}
Let $G$ denote the reduced \(\grs\)-Gr\"obner basis of $I$.
Since $\seqpol$ lies in $\zarisemi{\delta}{n}{n}$, one can apply
\cref{lemma_icreasing_degree_in_algorithm_vector_space} and deduce
that $\bdgshort{\delta}$ is nonempty.  In addition, the
sequence $\seqpol$ lies in $\zarimore{\delta}{n}$.  Consequently, one
can apply \cref{thm_card_bdg_deg_moitMB} and deduce that
$\bdgshort{\dregmb}$ is nonempty.  

Suppose, for the sake of contradiction, that for some $\degalgo \in
\{\delta + 1, \ldots , \dregmb-1\}$ the set
$\bdgshort{\degalgo}$ is empty.  We prove that this would
imply that $\bdgshort{\dregmb}$ is also empty.  
By \cref{lemma_icreasing_degree_in_algorithm}, in \algoName{algo:F4B},
each element of $\bdgshort{\degalgo+1}$ arises from a pair
$\pair{g_i}{g_j}$ that lies in $\taylorbasis^{\star}(G)$, where
$g_i$ and $g_j$ lie in $\bigsqcup_{k=\delta}^\degalgo \bdgshort{k}$
(see \cref{algo:F4-BuildTracer_critere}). 

As $\seqpol$ lies in $\zarisemi{\delta}{n}{n}$, then, by \cite[Proposition 1]{pardue2010generic}, the associated
Hilbert series is equal to \( \left( \sum_{i = 0}^{\delta-1}z^i
\right)^n \). 
This is a polynomial, which implies that the ideal $I$
is zero-dimensional.  Moreover, the sequence $\seqpol$ lies in
$\zarifilling{\delta}{n}{n}$. One can then apply
\cref{thm_type1_good_pairs} to deduce that such a pair $\pair{g_i}{g_j}$ is 
equal to $(t,t_i,g_i,t_j,g_j)$ where $(t_j,g_j)$ is distinguished.
Thus, the polynomial $g_j$ must lie in
$\bdgshort{\degalgo}$, which is empty by assumption.
Consequently, it follows that $\bdgshort{\degalgo+1}$ is
also empty.  

By iterating this reasoning, we conclude that
$\bdgshort{\dregmb}$ must be empty, which is a contradiction 
since we have already established that \(\bdgshort{\dregmb}\) is not empty.  
This ends the proof. 
\end{proof}

\begin{cor}
\label{cor_lien_homog_affine_reg}
Let $\seqpol=(f_1, \dots, f_{\nbpol})$ be a sequence of polynomials.
Suppose that $\hd{\seqpol}$ is a regular sequence. Then the following equality
holds:
\[ \lm{\grs}{ \langle \seqpol \rangle} = \lm{\grs}{\langle \hd{\seqpol} \rangle}.\]
\end{cor}

\begin{proof}
Recall that the set $\lm{\grs}{ \langle \seqpol \rangle}$ (resp.\
$\lm{\grs}{\langle \hd{\seqpol} \rangle}$) is the set of monomials
which are divisible by the leading monomials of a minimal grevlex
Gröbner basis $G$ (resp.\  $\widetilde{G}$) of $\langle \seqpol
\rangle$ (resp.\  $\langle \hd{\seqpol} \rangle$). 
Let $G$ be a minimal grevlex Gröbner basis of $\langle \seqpol
\rangle$ computed by \algoName{algo:F4B}.

Because $\hd{\seqpol}$ is a regular sequence, one can apply
\cref{lemma_icreasing_degree_in_algorithm} and use $\asseralgoshort{1}$
and $\asseralgoshort{2}$ to deduce that there is no degree fall
during $\mathsf{Reduction}$ at \cref{algo:F4-BuildTracer_Reduction}.
This implies that all the top-reductions performed in
\algoName{algo:F4B} involve only the homogeneous components of largest degree
of the polynomials under consideration.  Thus, truncating the
computation only to the part of maximum degree gives an exact way to
compute a minimal grevlex Gröbner basis of $\langle \hd{\seqpol}
\rangle$.  It follows that the new leading monomials found during the
execution of \algoName{algo:F4B} with $\seqpol$ as input depend solely
on $\hd{\seqpol}$.
\end{proof}

\section{Complexity analysis}
\label{sec:complexity_analysis}
\subsection{Complexity statements}
\label{sec:complexity_proof}

In this subsection, we state our complexity results on
\algoName{algo:F4T} under genericity assumptions.
As mentioned in \cref{section_intro}, we use several
Zariski open sets, which we intersect to formulate our hypotheses.
Some of these sets are known to be non-empty, while the non-emptiness of others
relies on conjectures.

\begin{itemize}
\item For all $i \geq n$, the set $\zarisemi{\delta}{i}{n}$ is non-empty by \cite[Theorem~1.5]{moreno2003degrevlex}.
\item The set $\zarimore{\delta}{n}$ is non-empty if \cite[Conjectures 1.2
and 1.6]{moreno2003degrevlex} hold.
\item The set $\zarifilling{\delta}{n}{n}$ is non-empty if \cite[Conjecture~4.1]{moreno2003degrevlex} holds.
\item For all $i < n$, the set $\zarisemi{\delta}{i}{n}$ is non-empty if \cite[Conjecture 1.1]{froberg1994hilbert} holds.
\end{itemize}
We recall that \cite[Conjecture 1.1]{froberg1994hilbert} and \cite[Conjecture 1.2]{moreno2003degrevlex} both designate the same statement.

In the following result, we make use of the notion of \emph{simultaneous Noether position},
as defined in \cite[Definition~3]{bardet2014complexity}. We introduce the following notation.
For a positive real number \(\proporMB\) in \( (0,1] \), we define 
\[E_\proporMB(\delta) = 
\ln\left( \frac{1+\proporMB(\delta-1)}{\delta} \right)
+ 
\proporMB(\delta-1)\ln\left(\frac{1+ \proporMB(\delta-1)}{\proporMB(\delta-1)}\right).\] 
For an integer $\delta \geq 2$ we define $L_{\omega}(\delta) = \frac{\delta}{1-e^{-\frac{1}{(\omega-2)}}}-\frac{1}{1-e^{-\frac{1}{\delta(\omega-2)}}}$ where
\(\omega\) is the constant of matrix multiplication over \(\K\).

Note that $L_\omega(\delta)$ is not defined if $\omega = 2$. 
By convention, we set $L_2(\delta) = 1$ for any $\delta \geq 2$.
Further, we denote by \(\bm{\ell}(\omega)\) the
constant \(L_\omega(2)\).
 Finally,
for any real number \(\varepsilon\), we define 
the constant  
\[
    \bm{c}(\varepsilon, \delta, \omega) = E_{1/2}(\delta) + E_{\bm{\ell}(\omega)+\varepsilon}(\delta).
\] 
The truncation $\phi$ is defined in \cref{subsection_Macaulay_matrices}.
We recall that we use the classical notation $\mathcal{O}$ and $\widetilde{\mathcal{O}}$
  for asymptotic bounds, as defined in
\cite[Definitions~25.7 and~25.8]{vzgg03}.
We recall that for two functions $f,g : \mathbb{N} \to \mathbb{R}_{\ge 0}$,
the function $f(n)$ belongs to $\widetilde{\mathcal{O}}\!\left(g(n)\right)$
if there exist constants $N,c$ such that
\[
f(n) \le g(n)\bigl(\log_2(3+g(n))\bigr)^c
\]
for all $n \ge N$.

\begin{thm}
\label{thm_complexity_F4_tracer_assymp}
Let $\seqpol= (f_1, \ldots ,f_n)$ be a sequence of nonzero polynomials
in the ring $\Kp{n}$ such that the sequence $\hd{\seqpol} = (\hd{f_1},
\ldots , \hd{f_n})$ lies in $\zarifilling{\delta}{n}{n} \cap
\zarimore{\delta}{n} \cap \zarisemi{\delta}{n}{n}$ where $\delta \geq
2$ and $n \geq 2$.
Suppose also that 
$\phi(\hd{\seqpol},n-1)$ lies in $\mathcal{S}_{\delta,n-1,n}$.
We suppose that the variables $x_1, \ldots , x_n$ are in simultaneous Noether 
position with respect to $\hd{\seqpol}$.
Let \(\mathcal{T}\) be the Gr\"obner trace returned when 
running \algoName{algo:F4B} with \(\seqpol\) as input and the graded
reverse lexicographical ordering $\grs$.  

For a fixed value of \(\delta \ge 2\),
running
\algoName{algo:F4T} on input a sequence $\bm{g} = \left( g_1, \ldots,
g_n\right)$ of polynomial equations which is compatible with
$\mathcal{T}$,
computes a minimal Gröbner basis for $(\bm{g}, \grs)$ using
\[ \widetilde{\mathcal{O}}\left(
\delta^{\omega n+1} e^{n \bm{c}(\varepsilon, \delta, \omega)}
\right) \]
arithmetic operations in $\K$ for all $\epsilon >0$.
\end{thm}

We recall that the hidden factor in the $\widetilde{\mathcal{O}}$ notation of
\cref{thm:main} is bounded by a polynomial in $n$ whose degree
grows linearly with $\omega$ and does not depend on $\delta$
or $\varepsilon$.
One notable fact in \cref{thm_complexity_F4_tracer_assymp} is that,
using the trace \(\mathcal{T}\), one can test whether several
of its hypotheses are satisfied.
Indeed, from the trace \(\mathcal{T}\) one can determine whether the
degree of the pairs decreases during the computation.
If this is not the case, then
\[
\lm{\grs}{\seqpol} = \lm{\grs}{\hd{\seqpol}},
\]
as shown in the proof of \cref{cor_lien_homog_affine_reg}.
Consequently, the set \(\lm{\grs}{\hd{\seqpol}}\) determines the Hilbert
series of \(\ideal{\hd{\seqpol}}\), which allows one to decide whether
\(\hd{\seqpol}\) lies in \(\zarisemi{\delta}{n}{n}\).
The same reasoning shows that the condition that \(\phi(\hd{\seqpol}, n-1)\) lies in
\(\mathcal{S}_{\delta,n-1,n}\) can also be verified from
\(\lm{\grs}{\ideal{\seqpol}}\) using
\cref{cor_lien_mono_base_quotient_tronqué}.
Finally, since \(\lm{\grs}{\ideal{\hd{\seqpol}}}\) completely describes the
monomial staircase of \(\ideal{\hd{\seqpol}}\), it can also be used to
verify that \(\hd{\seqpol}\) lies in
\(\zarifilling{\delta}{n}{n}\) and \(\zarimore{\delta}{n}\).

In order to prove that \cref{thm:main} follows from
\cref{thm_complexity_F4_tracer_assymp}, we need to show that the hypothesis of
simultaneous Noether position is generic. While this fact is considered part of the
folklore, we did not find a formal statement in the literature. The proof is given in
\cref{lemma_Noether_position_generic}. Combined with the fact that being a regular sequence is generic,
this lemma yields the desired result.

\cref{thm_complexity_F4_tracer_assymp} follows from a more precise formula,
given below in \cref{thm_complexity_F4_tracer}.
This formula holds even when $\delta$ is not fixed. Moreover, it requires fewer
hypotheses than \cref{thm_complexity_F4_tracer_assymp}, since it does not rely on
the simultaneous Noether position assumption.

Consider a sequence $\seqpol$ of polynomials in $\Kp{n}$, all of degree $\delta$.
We recall that the sets $\bdgh{\ast,}{\ast}{\ideal{\seqpol}}{\grs}$ and
$\bdgb{\ast,}{\ast}{\ideal{\seqpol}}{\grs}$ were defined in
\cref{subsection_Onthemonomialstaircase}.
For all $j$ in $\{1, \ldots , n\}$ and all integers $\degalgo$, we set
\[
\gamma_{\degalgo,j} = \card{\bdgh{\degalgo,}{j}{\ideal{\seqpol}}{\grs}}.
\]
Moreover, for all integers $\indalgo$ and all $i$ in $\{1, \ldots , n\}$, we define
\[
\nbred{\indalgo}{i}
=
\binom{i-1 + \indalgo + \delta}{\indalgo + \delta}
-
\binom{i-2 + \indalgo + \delta}{\indalgo + \delta}
-
\bigl(
\card{\bdgb{\indalgo + \delta,}{i}{\ideal{\seqpol}}{\grs}}
-
\card{\bdgb{\indalgo + \delta,}{i-1}{\ideal{\seqpol}}{\grs}}
\bigr).
\]
Now, we define
\begin{align*}
\formulaN{1}{\seqpol} &=
\sum_{\degalgo = \delta}^{\dregprev+1}
\sum_{i = 1}^{n}
\binom{n + \degalgo}{\degalgo}
\left(
\frac{
\nbred{\degalgo-\delta}{i}
\sum_{j = i}^n {\gamma_{\degalgo,j}}
}{
\min\!\left(
\nbred{\degalgo-\delta}{i},
\sum_{j = i}^n {\gamma_{\degalgo,j}}
\right)^{3-\omega}
}
+
\left(
\sum_{j = i}^n {\gamma_{\degalgo,j}}
\right)^{\omega-1}
\right), \\
\formulaN{2}{\seqpol} &=
\sum_{\degalgo = \dregprev+2}^{\Dmac}
{\gamma_{\degalgo,n}}^{\omega-2}
\binom{n + \degalgo}{\degalgo}
\sum_{k = 0}^2
\binom{n-2 + \Dmac-\degalgo+k}{\Dmac-\degalgo+k},
\end{align*}
and we set
\[
\compfqt = \formulaN{1}{\seqpol} + \formulaN{2}{\seqpol}.
\]

The integer $\dregmb = n(\delta - 1) + 1$ denotes the Macaulay bound, which,
for a zero-dimensional ideal generated by a regular sequence, is the maximum
degree of the elements in any graded Gröbner basis of the ideal 
\cite[Theorem~3]{lazard83}.
Moreover, we set
\[
\dregprev = \left\lfloor \frac{\dregmb-1}{2} \right\rfloor + 1,
\]
which corresponds to the degree of regularity of an ideal generated by a sequence
in $\mathcal{S}_{\delta,n-1,n}$ as seen in \cref{subsectionThelastdimensionofthegeneriquemonomialstaircase}.

\begin{thm}
\label{thm_complexity_F4_tracer}
Let $\seqpol= (f_1, \ldots ,f_n)$ be a sequence of nonzero polynomials in the ring $\Kp{n}$ such that $\hd{\seqpol} = (\hd{f_1}, \ldots , \hd{f_n})$ lies in $\zarifilling{\delta}{n}{n} \cap \zarimore{\delta}{n} \cap \zarisemi{\delta}{n}{n}$ where $\delta \geq 2$ and $n \geq 2$.
Suppose also that 
$\phi(\hd{\seqpol},n-1)$ lies in $\mathcal{S}_{\delta,n-1,n}$.
Let \(\mathcal{T}\) be the Gr\"obner trace returned when 
running \algoName{algo:F4B} with \(\seqpol\) as input.
Then running
\algoName{algo:F4T} on input a sequence $\bm{g} = \left( g_1, \ldots,
g_n\right)$ of polynomial equations which is compatible with
$\mathcal{T}$,
computes a minimal Gröbner basis for $(\bm{g}, \grs)$ using
\(\mathcal{O}( \compfqt )\)
arithmetic operations in $\K$,
with \(\compfqt\) as defined above.
\end{thm}

\subsection{Proof of \cref{thm_complexity_F4_tracer}}
\label{subsec_proof_comp_form}

Further, we denote by \(I\) (resp.~\(J\)) the ideal generated by
\(\seqpol\) (resp.~\(\hd{\seqpol}\)).
Since, by hypothesis, $\hd{\seqpol}$ lies in $\zarisemi{\delta}{n}{n}$,
the sequence $\hd{\seqpol}$ is regular. We can therefore
apply \cref{cor_lien_homog_affine_reg} to obtain
\begin{align*}
\lm{\grs}{ I }
  =
  \lm{\grs}{ J }. \tag*{$(\clubsuit)$} \label{eq:lm:equal}
\end{align*}
Thus, the set $\lm{\grs}{ I }$ inherits
all the properties established for $\lm{\grs}{ J }$,
since $I$ and $J$ share the same leading monomial ideal.

Using again the fact that $\hd{\seqpol}$ lies in
$\zarisemi{\delta}{n}{n}$, one can apply
\cref{lemma_icreasing_degree_in_algorithm} and deduce, from
$\asseralgo{\indalgo}{1}$ and $\asseralgo{\indalgo}{2}$ holding
for all $\indalgo \in \mathbb{N}$,
that all polynomials in $\bdg{}{}{J}{\grs}$ have
degree greater than~$\delta$. Moreover, since the
Hilbert series of $\hd{\seqpol}$ is equal to
$\left( \sum_{i = 0}^{\delta-1} z^i \right)^n$
by \cite[Proposition~1]{pardue2010generic},
we deduce that all polynomials in $\bdg{}{}{J}{\grs}$
have degree less than the Macaulay
bound \(\Dmac = n\left( \delta - 1 \right) + 1\)
\cite[Theorem~3]{lazard83}. Hence, the degrees of all
polynomials in $\bdg{}{}{J}{\grs}$ lie in
$\{\delta, \ldots, \Dmac\}$.

By assumption, $\hd{\seqpol}$ lies in
$\zarifilling{\delta}{n}{n} \cap
\zarimore{\delta}{n} \cap
\zarisemi{\delta}{n}{n}$.
Therefore, one can use \cref{lemma_newbdg_each_step}
and deduce that $\bdg{\degalgo,}{\ast}{J}{\grs}$ is nonempty
if and only if $\degalgo$ lies in $\{\delta, \ldots, \Dmac\}$.
Combined with \ref{eq:lm:equal}, this implies that
$\bdg{\degalgo,}{\ast}{I}{\grs}$ is nonempty
if and only if $\degalgo$ lies in $\{\delta, \ldots, \Dmac\}$.

Using \cref{lemma_icreasing_degree_in_algorithm}, each iteration of the
\texttt{while}-loop in \cref{algo:F4B:while} of \algoName{algo:F4B}
corresponds to a unique degree, since the degree never decreases.
This implies that there are exactly $\Dmac - \delta$ iterations of the
\texttt{while}-loop in \cref{algo:F4B:while} of \algoName{algo:F4B}.
Moreover, using the notation of \cref{lemma_icreasing_degree_in_algorithm},
we have
\[
\curdeg{\indalgo} = \indalgo + \delta
\quad \text{for all } \indalgo \in \{0, \ldots, \Dmac - \delta\}.
\]

Consider a sequence $\seqpolt$ that is compatible with
$\mathcal{T}$ and generates the ideal $\ideal{\seqpolt}$.
Note that, by \ref{eq:lm:equal}, for all integers $j \in \{1, \ldots, n\}$ and
all degrees $d$, we have
\begin{align*}
\bdgh{d,}{j}{J}{\grs}
=
\bdgh{d,}{j}{I}{\grs}
=
\bdgh{d,}{j}{\ideal{\seqpolt}}{\grs}.
\tag*{$(\diamondsuit)$} \label{eq:lm:equalbis}
\end{align*}

We bound the number of arithmetic operations over $\K$ performed by
\algoName{algo:F4T} when $\seqpolt$ is given as input.
Since $\seqpolt$ is compatible with $\mathcal{T}$,
it inherits all the previously established properties of
$\bdg{}{}{I}{\grs}$.
In particular, by \ref{eq:lm:equalbis}, the integer
$\indalgo_{\max}$ in \algoName{algo:F4T} is equal to $\Dmac - \delta$.

We denote by $\opark{\indalgo}$ the number of arithmetic operations in $\K$
performed by \algoName{algo:F4T} during the iteration of the
\texttt{for}-loop at \cref{algo:F4-UsingTracer_for} indexed by $\indalgo$,
for $\indalgo \in \{1, \ldots, \Dmac - \delta\}$.
We recall that, since
$\curdeg{\indalgo} = \indalgo + \delta$ by
\cref{lemma_icreasing_degree_in_algorithm},
the pairs selected during this iteration
have degree $\indalgo + \delta$ and are precisely those used to compute the
elements of a grevlex minimal Gröbner basis of $\ideal{\seqpolt}$ whose
leading monomial lies in
$\bdgh{\indalgo + \delta,}{\ast}{\ideal{\seqpolt}}{\grs}$.
The integer $\opark{0}$ denotes the number of arithmetic operations in $\K$
performed by \algoName{algo:F4T} before the \texttt{for}-loop at
\cref{algo:F4-UsingTracer_for}.

Finally, we denote by $\gamma_{d}$ (resp.~$\gamma_{d,j}$) the cardinality of
$\bdgh{d,}{\ast}{I}{\grs}$ (resp.~$\bdgh{d,}{j}{I}{\grs}$).
Moreover, using \ref{eq:lm:equalbis}, we have that $\gamma_{d}$
(resp.~$\gamma_{d,j}$) is equal to the cardinality of
$\bdgh{d,}{\ast}{\ideal{\seqpolt}}{\grs}$
(resp.~$\bdgh{d,}{j}{\ideal{\seqpolt}}{\grs}$).

\smallskip

Since, by assumption, $\hd{\seqpol}$ lies in $\zarisemi{\delta}{n}{n}$,
none of these polynomials is a $\K$-linear combination of the others.
This implies that the row echelon form of the Macaulay matrix associated
with $\seqpol$ has all its pivots in columns indexed by monomials of
degree~$\delta$, and that there are exactly $n$ such pivots.
Hence, we deduce that the Macaulay matrix considered at
\cref{algo:F4-UsingTracer_before_for} of \algoName{algo:F4T} has full rank.
Consequently, using \cite[Chapter~2]{Storjohann2000},
we have $\gamma_{\delta} = n$, which implies that
\[
\opark{0} \in \mathcal{O} \left(
\gamma_{\delta}^{\omega-1} \binom{\delta + n}{\delta}
\right) \in \mathcal{O} \left(
n^{\omega-1} \binom{\delta + n}{\delta}
\right).
\]
\smallskip

Let $\indalgo \in \{1, \ldots, \Dmac - \delta\}$. We now bound $\opark{\indalgo}$.
There are two steps to consider: the first one is described in
\cref{algo:F4-UsingTracer_constpairs}, and the second one in
\cref{algo:F4-UsingTracer_Reduction} of \algoName{algo:F4T}.
The numbers of arithmetic operations over $\K$ required for these two steps
are denoted respectively by
$\opark{\indalgo}^{\mathrm{pairs}}$ and $\opark{\indalgo}^{\mathrm{red}}$.

We first estimate $\opark{\indalgo}^{\mathrm{pairs}}$.
Since $\mathcal{T}$ is a Gröbner trace of $\seqpol$,
and that $\seqpolt$ is compatible with $\mathcal{T}$, the number of
S-polynomials to be constructed in
\cref{algo:F4-UsingTracer_constpairs} is equal to
$\gamma_{\delta + \indalgo}$.
The polynomials involved in the construction of these S-polynomials have
degree $\delta + \indalgo$.
This step therefore consists of performing at most
$\gamma_{\delta + \indalgo}$ additions of two polynomials of degree
bounded above by $\delta + \indalgo$.
Thus, $\opark{\indalgo}^{\mathrm{pairs}}$ can be bounded as follows:
\[
\opark{\indalgo}^{\mathrm{pairs}}
\;\le\;
\gamma_{\delta + \indalgo}
\binom{n + \delta + \indalgo}{\delta + \indalgo}.
\]

We now derive a bound for the integer $\opark{\indalgo}^{\mathrm{red}}$.
We denote by
\[
\mathscr{S} = \mac{}{}{S_{\indalgo}}{\supo{S_{\indalgo} \cup R_{\indalgo}}}{}{}{}
\quad\text{and}\quad
\mathscr{R} = \mac{}{}{R_{\indalgo}}{\supo{S_{\indalgo} \cup R_{\indalgo}}}{}{}{}
\]
the two matrices involved in the reduction step.
For $i \in \{1, \ldots, n\}$, we denote by
$\opark{\indalgo,i,1}$ (resp.~$\opark{\indalgo,i,2}$)
the number of arithmetic operations in $\K$ performed by
\algoName{algo:Reduction} at
\cref{algorithm_reduction_red} (resp.~\cref{algorithm_reduction_ech})
during the iteration of the \texttt{for}-loop indexed by~$i$.
Since we are in a top-reduction setting and the S-polynomials have the same
degree as the newly computed polynomials, the pivots of $\mathscr{R}$ lie in
columns indexed by monomials of maximal degree.
By \cref{remark_degree_last_lines},
\cref{algorithm_reduction_redap} and
\cref{algorithm_reduction_echap} of \algoName{algo:Reduction} do not occur.
It follows that
\[
\opark{\indalgo}^{\mathrm{red}}
\;=\;
\sum_{i=1}^{n} \left(
\opark{\indalgo,i,1} + \opark{\indalgo,i,2}
\right).
\]

Let $i \in \{1, \ldots, n\}$.
We start by estimating $\opark{\indalgo,i,2}$.
Let $\seqmono$ be the sequence of monomials indexing the columns of
$\mathscr{S}$ and $\mathscr{R}$.
As in \cref{sec:reduction}, we denote by $\seqmonot{i}$
the subset of $\seqmono$ consisting of monomials of maximal degree that lie in
$\mon{}{i} \setminus \mon{}{i-1}$.

At this stage of \algoName{algo:Reduction}, the matrix
$\selered{\mathscr{S}'}{\seqmonot}{\seqmonot{i}}$ has at most
$\binom{n + \indalgo + \delta}{\indalgo + \delta}$ columns.
Since all leading monomials in $\Kp{i-1}$ of degree $\indalgo + \delta$
have already been identified by the previous steps of
\algoName{algo:Reduction}, the number of rows of the matrix
$\selered{\mathscr{S}'}{\seqmonot}{\seqmonot{i}}$ is bounded above by
\[
\sum_{j = i}^{n} \gamma_{\indalgo + \delta, j}.
\]

We can make a bijection between each row of $\selered{\mathscr{S}'}{\seqmonot}{\seqmonot{i}}$
and a subset of monomials of degree $\indalgo + \delta$,
and therefore the number of columns is greater than the number of rows.
By \cite[Chapter~2]{Storjohann2000}, it follows that
$\opark{\indalgo,i,2}$ lies in
\begin{align*}
\mathcal{O}\!\left(
\left( \sum_{j = i}^{n}
\gamma_{\indalgo + \delta, j} \right)^{\omega - 1}
\binom{n + \indalgo + \delta}{\indalgo + \delta}
\right).
\end{align*}

We now estimate $\opark{\indalgo,i,1}$.
At this stage of \algoName{algo:Reduction}, the matrix
$\selered{\mathscr{S}'}{\seqmonot}{\seqmonot{i}}$ has
$\binom{n + \indalgo + \delta}{\indalgo + \delta}$ columns and at most
$\sum_{j = i}^{n} \gamma_{\indalgo + \delta,j}$ rows.
Similarly, the matrix
$\selered{\mathscr{R}}{\seqmonot}{\seqmonot{i}}$ has at most
$\binom{n + \indalgo + \delta}{\indalgo + \delta}$ columns.
We now bound the number of rows of
$\selered{\mathscr{R}}{\seqmonot}{\seqmonot{i}}$.
By definition, it is bounded above by the cardinality of $\seqmonot{i}$,
which is itself bounded by the cardinality of
$\mon{\delta + \indalgo,}{i} \setminus \mon{\delta + \indalgo,}{i-1}$, namely
\[
\binom{i-1 + \indalgo + \delta}{\indalgo + \delta}
-
\binom{i-2 + \indalgo + \delta}{\indalgo + \delta}.
\]

However, each row of
$\selered{\mathscr{R}}{\seqmonot}{\seqmonot{i}}$ represents a polynomial in
$\ideal{\seqpolt}$, so its leading monomial must belong to
$\lm{\grs}{\ideal{\seqpolt}}$ by \ref{eq:lm:equalbis}.
Hence, the number of rows of
$\selered{\mathscr{R}}{\seqmonot}{\seqmonot{i}}$ is bounded above by
$\nbred{\indalgo}{i}$, where
\[
\nbred{\indalgo}{i}
=
\binom{i-1 + \indalgo + \delta}{\indalgo + \delta}
-
\binom{i-2 + \indalgo + \delta}{\indalgo + \delta}
-
\bigl( \beta_{\iota, i} - \beta_{\iota, i-1} \bigr),
\]
with
\[
\beta_{\iota, i}
=
\card{\bdgb{\indalgo + \delta,}{i}{\ideal{\seqpolt}}{\grs}}.
\]

We apply \cref{prop_complexity_normal_form} with $a$ equal to the number of
rows of $\selered{\mathscr{R}}{\seqmonot}{\seqmonot{i}}$, $c$ equal to the
number of rows of $\selered{\mathscr{S}'}{\seqmonot}{\seqmonot{i}}$, and $b$
equal to the number of columns of these two matrices minus $a$.
Since $b$ is bounded above by
$\binom{n + \indalgo + \delta}{\indalgo + \delta}$, which is larger than both
$a$ and~$c$, we obtain that $\opark{\indalgo,i,1}$ lies in
\[
\mathcal{O}\!\left(
\frac{
\binom{n + \indalgo + \delta}{\indalgo + \delta}\,
\nbred{\indalgo}{i}\,
\displaystyle\sum_{j = i}^{n}
\gamma_{\indalgo + \delta,j}
}{
\min\!\left(
\nbred{\indalgo}{i},
\displaystyle\sum_{j = i}^{n}
\gamma_{\indalgo + \delta,j}
\right)^{3-\omega}
}
\right).
\]

\smallskip

From now on, the bound obtained for $\opark{\indalgo}^{\mathrm{pairs}}$
is dominated by that for $\sum_{i = 1}^{n} \opark{\indalgo,i,2}$.
Thus, we have $\opark{\indalgo} \in \mathcal{O}\!\left( N_{\indalgo} \right)$
where
\begin{align*}
N_{\indalgo}
=
\sum_{i = 1}^{n} \left(
\frac{
\binom{n + \indalgo + \delta}{\indalgo + \delta}\,
\nbred{\indalgo}{i}\,
\displaystyle\sum_{j = i}^{n}
\gamma_{\indalgo + \delta,j}
}{
\min\!\left(
\nbred{\indalgo}{i},
\displaystyle\sum_{j = i}^{n}
\gamma_{\indalgo + \delta,j}
\right)^{3-\omega}
}
+
\left(
\sum_{j = i}^{n}
\gamma_{\indalgo + \delta,j}
\right)^{\omega - 1}
\binom{n + \indalgo + \delta}{\indalgo + \delta}
\right).
\tag*{($\spadesuit$)}\label{eq:comp}
\end{align*}

Recalling that $\curdeg{\indalgo} = \indalgo + \delta$ by
\cref{lemma_icreasing_degree_in_algorithm}, this proves the result for all
$\indalgo \in \{0, \ldots, \dregprev - \delta + 1\}$.
Now, we define
\[
\formulaN{1}{\seqpol} = \sum_{\indalgo = 0}^{\dregprev - \delta + 1} N_{\indalgo}.
\]
It remains to prove the result for $\indalgo in \{ \dregprev + 2 - \delta ,\ldots , \Dmac\}$

\smallskip

We have just established a bound on $\opark{\indalgo}$.
We now refine this bound for $\indalgo \geq \dregprev + 2 - \delta$.
To this end, we first establish four key properties.

Since, by hypothesis, the sequence $\phi(\hd{\seqpol}, n-1)$ lies in
$\mathcal{S}_{\delta, n-1, n}$, the degree of regularity of the ideal
$\ideal{\phi(\hd{\seqpol}, n-1)}$ is equal to $\dregprev$, as shown in
\cref{subsectionThelastdimensionofthegeneriquemonomialstaircase}.
Moreover, since the Hilbert series of
$\ideal{\phi(\hd{\seqpol}, n-1)}$ is a polynomial by
\cite[Proposition~1]{pardue2010generic}, this ideal is zero-dimensional.
Thus, any element of a grevlex minimal Gröbner basis of
$\ideal{\phi(\hd{\seqpol}, n-1)}$ has degree at most $\dregprev$.

Using \cref{thm_semieregseq_macmat}, we know that any element of a grevlex
minimal Gröbner basis of $\ideal{\hd{\seqpol}}$ whose leading monomial has
degree greater than or equal to $\dregprev + 1$ lies in
$\mon{n}{} \setminus \mon{}{n-1}$.
Together with \ref{eq:lm:equalbis}, we deduce that, for all
$\indalgo \geq \dregprev + 1 - \delta$ and all
$i \in \{1, \ldots, n-1\}$,
\[
\gamma_{\indalgo + \delta, i} = 0.
\tag*{($\heartsuit,1$)}\label{eq:degnlm}
\]

By assumption, the sequence $\hd{\seqpol}$ lies in $\zarimore{\delta}{n}$.
Thus, using \cref{thm_card_bdg_deg_moitMB} applied to $\hd{\seqpol}$ together
with \ref{eq:lm:equalbis}, we obtain that, for all
$\indalgo \geq \dregprev + 1$,
\begin{align*}
\bdgh{\curdeg{\indalgo},}{n}{\ideal{\seqpolt}}{\grs}
=
\left\{
x_n^{2\curdeg{\indalgo} - \Dmac} \, m
\;\middle|\;
m \in \bdgb{\Dmac - \curdeg{\indalgo},}{n-1}{\ideal{\seqpolt}}{\grs}
\right\}.
\tag*{($\heartsuit,2$)} \label{eq:SM}
\end{align*}

Let $\indalgo$ be fixed in
$\{\dregprev + 2 - \delta, \ldots, \Dmac - \delta\}$.
We recall that $\curdeg{\indalgo} = \indalgo + \delta$.
This implies that $\curdeg{\indalgo} \geq \dregprev + 2$.
The goal is to derive a more precise bound on the number of reductors.

Let us consider a selected pair $p$ at step $\indalgo$ of
\algoName{algo:F4T}.
Since $\ideal{\seqpolt}$ is compatible with $\mathcal{T}$, there exists a corresponding
pair $\tilde{p}$ in \algoName{algo:F4B}, involving the same leading monomials,
with $\seqpol$ as input.
We denote $\tilde{p} = (t, t_i, g_i, t_j, g_j)$, where
\[
t = \operatorname{lcm}\!\bigl(\lm{\grs}{g_i}, \lm{\grs}{g_j}\bigr), \quad
t_i = \frac{t}{\lt{\grs}{g_i}}, \quad
\text{and} \quad
t_j = \frac{t}{\lt{\grs}{g_j}}.
\]
We prove that $(t_j, g_j)$ is \type.

First, it was established in \cite[Theorem~2.2]{faugere1999new} that the output
$G$ of \algoName{algo:F4B} is a grevlex Gröbner basis of $I$.
Moreover, using assertion $\asseralgo{\cdot}{3}$ of
\cref{lemma_icreasing_degree_in_algorithm}, we know that, since
$\hd{\seqpol}$ lies in $\zarisemi{\delta}{n}{n}$, this output is a minimal
grevlex Gröbner basis.

This situation fits into the framework of
\cref{subsection_Onthesyzygiesmodule}, which allows us to deduce that all
selected pairs lie in $\taylorbasis^{(1)}(G)$.
More precisely, these pairs belong to $\taylorbasis^{\star}(G)$, since they
are selected by the criterion described in
\cref{algo:F4-BuildTracer_critere} of \algoName{algo:F4B}.
Finally, since $\hd{\seqpol}$ lies in $\zarifilling{\delta}{n}{n}$ and the
ideal $I$ is zero-dimensional, we can apply
\cref{thm_type1_good_pairs} to conclude that $(t_j, g_j)$ is \type.

Since $\ideal{\seqpolt}$ is compatible with $\mathcal{T}$, the pair $p$
has exactly the same leading monomials as $\tilde{p}$.
We deduce that the pair $p$ is constructed using one polynomial whose leading
monomial lies in
$\bdgh{\curdeg{\indalgo}-1,}{n}{I}{\grs}$.
Since $\curdeg{\indalgo-1} = \curdeg{\indalgo}-1$, the pair $p$ is therefore
built using an element whose leading monomial lies in
$\bdgh{\curdeg{\indalgo-1},}{n}{I}{\grs}$.

Moreover, using \ref{eq:SM}, these elements have degree in $x_n$ equal to
$2\curdeg{\indalgo}-\Dmac-2$.
Therefore, the leading monomials of the polynomials used to construct the pairs
have degree in $x_n$ at least $2\curdeg{\indalgo}-\Dmac-2$.

Let us summarize the properties of the pairs at this step:
\begin{itemize}
\item they have total degree $\curdeg{\indalgo}$;
\item one element of each pair has degree $\curdeg{\indalgo}-1$;
\item the leading monomial of this element has degree in $x_n$ equal to
$2\curdeg{\indalgo}-\Dmac-2$;
\item the new polynomials produced at this step have degree in $x_n$ equal to
$2\curdeg{\indalgo}-\Dmac$, by \ref{eq:SM}.
\end{itemize}

Recall that we are in a top-reduction context.
We deduce that the trace $\mathcal{T}$ is such that the reductors used at this
step also have degree $\curdeg{\indalgo}$.
Consequently, for any leading monomial $\mu$ of a reductor, we have
\[
2\curdeg{\indalgo}-\Dmac-2
\;\le\;
\deg_{x_n}(\mu)
\;\le\;
2\curdeg{\indalgo}-\Dmac.
\]

Moreover, since $\curdeg{\indalgo} \geq \dregprev + 2$, we have
$2\curdeg{\indalgo}-\Dmac-2 > 0$.
We deduce that, for all integers
$i \in \{1, \ldots, n-1\}$, the number $\nbred{\indalgo}{i}$ of reductors
is zero.
It remains to bound the number of reductors corresponding to the case
$i = n$.
We can bound the number of such reductors by the number of monomials of degree
$\curdeg{\indalgo}$ whose degree in $x_n$ lies between
$2\curdeg{\indalgo}-\Dmac-2$ and $2\curdeg{\indalgo}-\Dmac$.
This yields
\[
\nbred{\indalgo}{n}
=
\sum_{k = 0}^{2}
\binom{n-2 + \Dmac-\curdeg{\indalgo}+k}{\Dmac-\curdeg{\indalgo}+k},
\quad\text{and}\quad
\nbred{\indalgo}{i} = 0
\;\text{ for all }\;
i \in \{1, \ldots, n-1\}.
\tag*{($\heartsuit,3$)} \label{eq:nbred}
\]

Using \ref{eq:SM}, we deduce that $\gamma_{\curdeg{\indalgo},n}$
 is bounded above by the cardinality of
$\bdgb{\Dmac-\curdeg{\indalgo},}{n-1}{\ideal{\seqpolt}}{\grs}$, which is itself bounded above by
$\binom{n-2 + \Dmac-\curdeg{\indalgo}}{\Dmac-\curdeg{\indalgo}}$.
Therefore, we obtain
\begin{align*}
\gamma_{\curdeg{\indalgo},n}
\leq
\binom{n-2 + \Dmac-\curdeg{\indalgo}}{\Dmac-\curdeg{\indalgo}} < 
\nbred{\indalgo}{n}. \tag*{($\heartsuit,4$)} \label{ineq:rednewel}
\end{align*}

Let us dive into a refined arithmetic complexity expression for this step.
First, we apply \ref{eq:comp} using \ref{eq:degnlm} and the new bound \ref{eq:nbred}
on the number of reductors. Thus, we obtain that 
$\opark{\indalgo}^{\mathrm{red}} \in \mathcal{O}\left( N_{\indalgo,\mathrm{red}} \right)$ with  
\begin{align*}
N_{\indalgo,\mathrm{red}}  &\leq \frac{
\binom{n + \indalgo + \delta}{\indalgo + \delta}\,
\nbred{\indalgo}{n}\,
\displaystyle 
\gamma_{\indalgo + \delta,n}
}{
\min\!\left(
\nbred{\indalgo}{i},
\displaystyle 
\gamma_{\indalgo + \delta,n}
\right)^{3-\omega}} + 
\gamma_{\indalgo + \delta,n}^{\omega - 1}
\binom{n + \indalgo + \delta}{\indalgo + \delta} \\
&\leq \binom{n + \indalgo + \delta}{\indalgo + \delta} \left(
\nbred{\indalgo}{n}\,
\gamma_{\indalgo + \delta,n}^{\omega -2}
+ \gamma_{\indalgo + \delta,n}^{\omega - 1} \right) \tag{by \ref{ineq:rednewel}} \\
&\leq 2\binom{n + \indalgo + \delta}{\indalgo + \delta}
\nbred{\indalgo}{n}\,
\gamma_{\indalgo + \delta,n}^{\omega -2}. \tag{by \ref{ineq:rednewel}}
\end{align*}
We have already seen that $\opark{\indalgo}^{\mathrm{pairs}} \in \mathcal{O}\left( N_{\indalgo,\mathrm{pairs}} \right)$ where
\begin{align*}
N_{\indalgo,\mathrm{pairs}}
&\;\le\;
\gamma_{\delta + \indalgo}
\binom{n + \delta + \indalgo}{\delta + \indalgo} \\
&\;\le\; \gamma_{\delta + \indalgo,n}
\binom{n + \delta + \indalgo}{\delta + \indalgo}  \tag{by \ref{ineq:rednewel}} \\
&\;\le\; \binom{n + \indalgo + \delta}{\indalgo + \delta}
\nbred{\indalgo}{n}\,
\gamma_{\indalgo + \delta,n}^{\omega -2}. \tag{by \ref{eq:degnlm}} 
\end{align*}
Consequently, it follows that $\opark{\indalgo} \in \mathcal{O}\left( \widetilde{N_{\indalgo}} \right)$ where 
\[ \widetilde{N_{\indalgo}} = \binom{n + \indalgo + \delta}{\indalgo + \delta}
\nbred{\indalgo}{n}\,
\gamma_{\indalgo + \delta,n}^{\omega -2}. \]
This ends the proof.
Moreover, we define $\formulaN{2}{\seqpol} = \sum_{\indalgo = \dregprev-\delta+2}^{\Dmac-\delta} \widetilde{N_{\indalgo}}$.

\subsection{Proof of \cref{thm_complexity_F4_tracer_assymp}}
\label{subsec_proofofassymp} 

Now we aim to understand the asymptotic behaviour of the complexity bound established in 
\cref{thm_complexity_F4_tracer},
when the number of variable $n$ goes to infinity.
We recall that for a sequence $\seqpol$
of $\Kp{n}$, we set
for all $j$ in $\{1, \ldots , n\}$ and all integers $\degalgo$ that
\[
\gamma_{\degalgo,j} = \card{\bdgh{\degalgo,}{j}{\ideal{\seqpol}}{\grs}}.
\]
Moreover, if $\seqpol$ generates a zero-dimensional ideal, for all integers $\indalgo$ and all $i$ in $\{1, \ldots , n\}$, we define
\[
\nbred{\indalgo}{i}
=
\binom{i-1 + \indalgo + \delta}{\indalgo + \delta}
-
\binom{i-2 + \indalgo + \delta}{\indalgo + \delta}
-
\bigl(
\card{\bdgb{\indalgo + \delta,}{i}{\ideal{\seqpol}}{\grs}}
-
\card{\bdgb{\indalgo + \delta,}{i-1}{\ideal{\seqpol}}{\grs}}
\bigr).
\] 

For a sequence $\seqpol$ of $\Kp{n}$, of degree $\delta$, that respects
the hypothesis of \cref{thm_complexity_F4_tracer}, we consider the 
decomposition of $\compfqt$ in its two terms using the notations of
\cref{subsec_proof_comp_form}
\[ \compfqt = \formulaN{1}{\seqpol} + \formulaN{2}{\seqpol} \]
where
\[ \formulaN{1}{\seqpol} = \sum_{\degalgo = \delta}^{\dregprev+1} \sum_{i = 1}^{n} \binom{n +
\degalgo}{\degalgo} \left( \frac{ \nbred{\degalgo-\delta}{i}
\sum_{j = i}^n {\gamma_{\degalgo,j}}} {\min\left(
\nbred{\degalgo-\delta}{i},  \sum_{j = i}^n
{\gamma_{\degalgo,j}} \right)^{3-\omega} } + \left(
\sum_{j = i}^n {\gamma_{\degalgo,j}} \right)^{\omega-1}
\right)
\]
and 
\[ \formulaN{2}{\seqpol} = \sum_{\degalgo = \dregprev+2}^{\Dmac}
{\gamma_{\degalgo,n}}^{\omega-2} \binom{n + \degalgo}{\degalgo}
\sum_{k = 0}^2 \binom{n-2 + \Dmac-\degalgo+k}{\Dmac-\degalgo+k}
\]
with the notation introduced in \cref{thm_complexity_F4_tracer}.

Recall that, for a positive real number \(\proporMB\) in \( (0,1] \), we define 
\[E_\proporMB(\delta) = 
\ln\left( \frac{1+\proporMB(\delta-1)}{\delta} \right)
+ 
\proporMB(\delta-1)\ln\left(\frac{1+ \proporMB(\delta-1)}{\proporMB(\delta-1)}\right).\] 
Moreover, we set
$L_{\omega}(2) = (1 + e^{-1/(2(\omega-2))})^{-1}$,
and we extend this definition by continuity by setting $L_{2}(2) = 1$.
Let us establish the following proposition in order to prove
\cref{thm_complexity_F4_tracer_assymp}.

\begin{prop}
\label{prop_asymp_two_terms}
With the previous notation and for a fixed value of \(\delta \ge 2\), we have that 
\begin{enumerate}
\item{} \label{prop_asymp_two_terms1} $\formulaN{1}{\seqpol} \in \widetilde{\mathcal{O}}\left(
\delta^{\omega n+1} e^{2E_{1/2}(\delta)n}
\right)$.
\end{enumerate}
Moreover, if the variables $x_1, \ldots , x_n$ are in simultaneous Noether 
position with respect to $\hd{\seqpol}$, we have that
\begin{enumerate}[resume]
\item{} \label{prop_asymp_two_terms2} $\formulaN{2}{\seqpol} \in \widetilde{\mathcal{O}}\left(
\delta^{\omega n+3-\omega} e^{n \left( E_{1/2}(\delta)+E_{L_{\omega}(2) +\epsilon}(\delta) \right)}
\right)$ for all $\epsilon > 0$.
\end{enumerate}
\end{prop}

We need the following two lemmata in order to complete the proof of
\cref{prop_asymp_two_terms}.
Here is a particular case of the Bézout bound and one of its consequences.

\begin{lemma}
\label{lemma_naiveboundcardbdg}
Let $\seqpol = (f_1, \ldots , f_n)$ be a regular sequence of homogeneous
polynomials of $\Kp{n}$ such that $\deg{f_i} = \delta_i$ for all $i$ 
in $\{1, \ldots , n\}$. 
Then it holds that
\[ \card{\bdgb{}{}{\ideal{\seqpol}}{\grs}} = \prod_{i = 1}^n \delta_i \quad \text{and} \quad \card{\bdg{}{}{\ideal{\seqpol}}{\grs}} \leq n\prod_{i = 1}^n \delta_i. \] 
\end{lemma}

\begin{proof}
Since $\seqpol$ is regular, then its Hilbert series is equal to 
$\prod_{i = 1}^{n}\frac{1-z^{\delta_i}}{1-z}$
by \cite[Proposition 1]{pardue2010generic}.
It is the same Hilbert series of the ideal $\ideal{x_1^{\delta_1}, \ldots ,x_n^{\delta_n}}$.
We deduce that 
\[\card{\bdgb{}{}{\ideal{\seqpol}}{\grs}} = \card{\bdgb{}{}{\ideal{x_1^{\delta_1}, \ldots ,x_n^{\delta_n}}}{\grs}} = \prod_{i = 1}^n \delta_i.\]
Using \cite[Proposition 2.1]{FGLM}, we have that 
\[ \card{\bdg{}{}{\ideal{\seqpol}}{\grs}} \leq n\card{\bdgb{}{}{\ideal{\seqpol}}{\grs}} \]
which ends the proof.
\end{proof}

Using generalized binomial coefficients expressed via the Gamma function, 
together with Stirling's formula, we obtain the following asymptotic result.

\begin{lemma}
\label{lemma_assymp_bino}
Let $\delta \geq 2$ be a fixed integer and let $p \in (0,1]$. Then, the
following holds: 
\begin{align*}
\binom{\,n + p(\delta-1)n\,}{\,p(\delta-1)n\,}
\underset{n \to +\infty}{\sim}\,
e^{E_p(\delta)\,n}\,
\delta^{n}\,
\sqrt{\frac{1 + p(\delta-1)}{2\pi n\, p(\delta-1)}},
\end{align*}
\end{lemma}

\begin{proof}
We first decompose the binomial
\begin{align*}
\binom{\,n + p(\delta-1)n\,}{\,p(\delta-1)n\,} &= 
\frac{\Gamma(n + p(\delta-1)n + 1)}{\Gamma(n + 1)\,\Gamma(p(\delta-1)n + 1)}. 
\end{align*}
Then we use the Stirling formula as follows
\begin{align*}
\binom{\,n + p(\delta-1)n\,}{\,p(\delta-1)n\,}
\underset{n \to +\infty}{\sim}\, & \left( \sqrt{\frac{1 + p(\delta-1)}{2\pi \, n\,p(\delta-1)}}\right) \left( \frac{(n + p(\delta-1)n)^{n + p(\delta-1)n}}{n^n\,(p(\delta-1)n)^{p(\delta-1)n}} \right) \\
\underset{n \to +\infty}{\sim}\, & \left( \sqrt{\frac{1 + p(\delta-1)}{2\pi \,n\,p(\delta-1)}}\right) (1 + p(\delta-1))^{n} 
\left( \frac{1+ p(\delta-1)}{p(\delta-1)} \right)^{p(\delta-1)n} \\
\underset{n \to +\infty}{\sim}\, & \left( \sqrt{\frac{1 + p(\delta-1)}{2\pi \,n\,p(\delta-1)}}\right) \delta^n \left( \frac{1 + p(\delta-1)}{\delta} \right)^{n} 
\left( \frac{1+ p(\delta-1)}{p(\delta-1)} \right)^{p(\delta-1)n} \\
\underset{n \to +\infty}{\sim}\, & \left( \sqrt{\frac{1 + p(\delta-1)}{2\pi \,n\,p(\delta-1)}}\right) \,\delta^n \,e^{E_p(\delta)\,n}  
\tag*{\qedhere}  % dirty trick due to misaligned qed symbol
\end{align*}
\end{proof}

\begin{proof}[Proof of \cref{prop_asymp_two_terms1} of \cref{prop_asymp_two_terms}]
We start by bounding $\formulaN{1}{\seqpol}$. 
\begin{align*}
&\sum_{\degalgo = \delta}^{\dregprev+1} \sum_{i = 1}^{n} \binom{n +
\degalgo}{\degalgo} \left( \frac{ \nbred{\degalgo-\delta}{i}
\sum_{j = i}^n {\gamma_{\degalgo,j}}} {\min\left(
\nbred{\degalgo-\delta}{i},  \sum_{j = i}^n
{\gamma_{\degalgo,j}} \right)^{3-\omega} } + \left(
\sum_{j = i}^n {\gamma_{\degalgo,j}} \right)^{\omega-1}
\right) \\
\leq & \sum_{\degalgo = \delta}^{\dregprev+1} \sum_{i = 1}^{n} \binom{n +\dregprev+1}{\dregprev+1} \left( \frac{ \nbred{\degalgo-\delta}{i}
\sum_{j = i}^n {\gamma_{\degalgo,j}}} {\min\left(
\nbred{\degalgo-\delta}{i},  \sum_{j = i}^n
{\gamma_{\degalgo,j}} \right)^{3-\omega} } + \left(
\sum_{j = i}^n {\gamma_{\degalgo,j}} \right)^{\omega-1}
\right) \\
\leq & \binom{n +\dregprev+1}{\dregprev+1} \sum_{\degalgo = \delta}^{\dregprev+1}
 \sum_{i = 1}^{n} \left( \frac{ \nbred{\degalgo-\delta}{i}
\sum_{j = i}^n {\gamma_{\degalgo,j}}} {\min\left(
\nbred{\degalgo-\delta}{i},  \sum_{j = i}^n
{\gamma_{\degalgo,j}} \right)^{3-\omega} } + \left(
\sum_{j = i}^n {\gamma_{\degalgo,j}} \right)^{\omega-1}
\right) \\
\leq & \binom{n +\dregprev+1}{\dregprev+1} \sum_{\degalgo = \delta}^{\dregprev+1}
 \sum_{i = 1}^{n} \left( \nbred{\degalgo-\delta}{i}
\left( \sum_{j = i}^n {\gamma_{\degalgo,j}} \right)^{\omega-2} + \nbred{\degalgo-\delta}{i}^{\omega-2}
\left( \sum_{j = i}^n {\gamma_{\degalgo,j}} \right) + \left(
\sum_{j = i}^n {\gamma_{\degalgo,j}} \right)^{\omega-1}
\right) \\
\leq & n^{\omega-1} \binom{n +\dregprev+1}{\dregprev+1}
\sum_{\degalgo = \delta}^{\dregprev+1}
\sum_{i = 1}^{n}
\left(
\nbred{\degalgo-\delta}{i}\,
\delta^{n(\omega-2)}
+ \nbred{\degalgo-\delta}{i}^{\omega-2}\,
\delta^{n}
+ \delta^{n(\omega-1)}
\right)
\tag{by \cref{lemma_naiveboundcardbdg}}
\\
&\text{Since, by definition, it holds that } \nbred{\degalgo-\delta}{i} \leq \binom{n +\dregprev+1}{\dregprev+1},\text{ we obtain}\\
&\sum_{\degalgo = \delta}^{\dregprev+1} \sum_{i = 1}^{n} \binom{n +
\degalgo}{\degalgo} \left( \frac{ \nbred{\degalgo-\delta}{i}
\sum_{j = i}^n {\gamma_{\degalgo,j}}} {\min\left(
\nbred{\degalgo-\delta}{i},  \sum_{j = i}^n
{\gamma_{\degalgo,j}} \right)^{3-\omega} } + \left(
\sum_{j = i}^n {\gamma_{\degalgo,j}} \right)^{\omega-1}
\right) \\
\leq &n^{\omega-1}
\binom{n +\dregprev+1}{\dregprev+1}\sum_{\degalgo = \delta}^{\dregprev+1}
\sum_{i = 1}^{n}\left(\binom{n +\dregprev+1}{\dregprev+1}
\delta^{n(\omega-2)}+\binom{n +\dregprev+1}{\dregprev+1}^{\omega-2}\delta^{n}+\delta^{n(\omega-1)}
\right). \\
\leq & \delta n^{\omega+1} \binom{n +\dregprev+1}{\dregprev+1}\left( \binom{n +\dregprev+1}{\dregprev+1}
\delta^{n(\omega-2)} + \binom{n +\dregprev+1}{\dregprev+1}^{\omega-2}
\delta^{n} + \delta^{n(\omega-1)}
\right). \tag*{$(1)$} \label{prop_asymp_first_term_ineg}
\end{align*}

We have that 
\begin{align*}
\binom{n +\dregprev+1}{\dregprev+1} &\leq \binom{n + \frac{n(\delta - 1)}{2} +2}{\frac{n(\delta - 1)}{2} +2} \leq \left(1 + \frac{2n}{n(\delta - 1)+4} \right)
\left(1 + \frac{2n}{n(\delta - 1)+2} \right) 
\binom{n + \frac{n(\delta - 1)}{2}}{\frac{n(\delta - 1)}{2}} \\
&\leq 9\binom{n + \frac{n(\delta - 1)}{2}}{\frac{n(\delta - 1)}{2}}.
\end{align*}
Applying \cref{lemma_assymp_bino} to the above bound with $p = \frac{1}{2}$,
we obtain that  
\[ \binom{n +\dregprev+1}{\dregprev+1} \in \widetilde{\mathcal{O}}\left(e^{E_{1/2}(\delta)\,n}\,
\delta^{n}
\right). \tag*{$(2)$} \label{prop_asymp_first_term_asymp} \]
Combining \ref{prop_asymp_first_term_ineg} and \ref{prop_asymp_first_term_asymp},
we have that 
\begin{align*}
\formulaN{1}{\seqpol} &\in \widetilde{\mathcal{O}}\left(
\delta^{\omega n+1} \left( e^{2E_{1/2}(\delta)n} + e^{(\omega-1)E_{1/2}(\delta)n} +
e^{E_{1/2}(\delta)n} \right)
\right) \\
&\subset \widetilde{\mathcal{O}}\left(
\delta^{\omega n+1} e^{2E_{1/2}(\delta)n}
\right).
\tag*{\qedhere}  % dirty trick due to misaligned qed symbol
\end{align*}
\end{proof}

We introduce some bounds before diving into the proof of 
\cref{prop_asymp_two_terms2} of \cref{prop_asymp_two_terms}.
To do so, we use \cite[Theorem~12]{bardet2014complexity}.
Let $i$ be in $\{1, \ldots , n\}$ and $d$ be an integer.
We first define $\cofc{d}{i}$ as follows
\begin{align*}
\sum_{d = 0}^{+\infty} \cofc{d}{i} z^d = z^{\delta} \prod_{j = 1}^{i-1} \frac{1-z^{\delta}}{1-z}.
\end{align*}
The following lemma establishes some bounds on the cardinalities of subsets 
of a grevlex minimal Gröbner basis under generic assumptions.
For this lemma, we use the following notation introduced in 
\cref{thm_complexity_F4_tracer}
\[ \gamma_{d}  = \card{\bdgh{d,}{\ast}{\ideal{\seqpol}}{\grs}} \quad \text{and} \quad \gamma_{d,i} = \card{\bdgh{d,}{i}{\ideal{\seqpol}}{\grs}}. \]

\begin{lemma}
\label{lemma_bound_bdi_combin}
Consider a sequence of polynomials $\seqpol = (f_1, \ldots , f_n)$ such that 
the variables \( x_1, \ldots, x_n \) are in simultaneous Noether
position with respect to the sequence of polynomials
$(\hd{f_1}, \ldots , \hd{f_n})$. 

Then, for all $i \in \{1, \ldots , n\}$ and all $\degalgo \in \mathbb{N}$, we have
\begin{align*}
\sum_{j = i}^n \gamma_{\degalgo,j} \leq \sum_{j = i}^n \cofc{\degalgo}{j}.
\end{align*}
\end{lemma}

\begin{proof}
We prove this lemma by establishing four key assertions.

First, since the variables \( x_1, \ldots, x_n \) are in simultaneous Noether
position with respect to the sequence of polynomials
$(\hd{f_1}, \ldots , \hd{f_n})$, it follows from
\cite[Proposition~6]{bardet2014complexity} that
$(\hd{f_1}, \ldots , \hd{f_n})$ is a regular sequence.
Therefore, we apply \cref{cor_lien_homog_affine_reg} to deduce that
\[
\lm{\grs}{ \ideal{\seqpol} } = \lm{\grs}{\ideal{\hd{\seqpol}}}.
\tag*{$(\star,1)$} \label{lemma_bound_bdi_combin_prooflm1}
\]
For the remainder of the proof, we set
$I_j = \ideal{\hd{f_1}, \ldots , \hd{f_j}}$ for all $j \in \{1,\ldots ,n\}$.

Let $i \in \{1, \ldots , n\}$.
Since the variables \( x_1, \ldots, x_n \) are in simultaneous Noether
position with respect to the sequence of polynomials
$(\hd{f_1}, \ldots , \hd{f_n})$,
we can apply \cite[Proposition~11]{bardet2014complexity}.
Then, we obtain
\[
\bigcup_{j = 1}^{i-1} \bdgh{\degalgo,}{\ast}{I_{j}}{\grs} 
\subseteq \mon{\degalgo,}{i-1}.
\]
On the other hand, by definition,
\[
\bigsqcup_{j = i}^n
\bdgh{\degalgo,}{j}{I_n}{\grs}
\subseteq \left( \mon{d}{} \setminus \mon{d,}{i-1} \right).
\]
It follows that
\[
\bigsqcup_{j = i}^n
\bdgh{\degalgo,}{j}{I_n}{\grs}
\subseteq 
\bdgh{\degalgo,}{\ast}{I_n}{\grs}
\setminus \left( \bigcup_{j = 1}^{i-1} \bdgh{\degalgo,}{\ast}{I_{j}}{\grs} \right).
\tag*{$(\star,2)$} \label{lemma_bound_bdi_combin_prooflm2}
\]

To establish the next property, we use the following remark.

\begin{remark}
\label{remark_lemma_bound_bdi_combin}
Let $E_1, \ldots , E_n$ be sets. Then, for all $i \in \{1, \ldots ,n\}$, we have
\[
E_n \setminus \left( \bigcup_{j = 1}^{i-1} E_j \right) \subseteq
\bigsqcup_{j = i}^{n} \left( E_j \setminus \bigcup_{k = 1}^{j-1} E_k \right).
\]
\end{remark}

Setting $E_i = \bdgh{\degalgo,}{\ast}{I_i}{\grs}$ and applying
\cref{remark_lemma_bound_bdi_combin}, we obtain
\[
\bdgh{\degalgo,}{\ast}{I_n}{\grs}
\setminus \left( \bigcup_{j = 1}^{i-1} \bdgh{\degalgo,}{\ast}{I_{j}}{\grs} \right)
\subseteq 
\bigsqcup_{j = i}^n
\bdgh{\degalgo,}{\ast}{I_j }{\grs} \setminus
\left( \bigcup_{k = 1}^{j-1} \bdgh{\degalgo,}{\ast}{I_{k}}{\grs} \right).
\tag*{$(\star,3)$} \label{lemma_bound_bdi_combin_prooflm3}
\]

Using once again that the variables \( x_1, \ldots, x_n \) are in simultaneous
Noether position with respect to the sequence of polynomials
$(\hd{f_1}, \ldots , \hd{f_n})$, it follows from
\cite[Theorem~12]{bardet2014complexity} that the inequality
\[
\card{\bdgh{\degalgo,}{\ast}{I_j}{\grs} \setminus
\bdgh{\degalgo,}{\ast}{I_{j-1}}{\grs}
} \leq \cofc{\degalgo}{j}
\tag*{$(\star,4)$} \label{lemma_bound_bdi_combin_prooflm4}
\]
holds for all $j \in \{1, \ldots , n\}$.

We can now conclude that
\begin{align*}
\sum_{j = i}^n \gamma_{\degalgo,j}
&= \card{\bigsqcup_{j = i}^n
\bdgh{\degalgo,}{j}{I_n}{\grs}}
\tag{by \ref{lemma_bound_bdi_combin_prooflm1}}\\
&\leq \card{\bdgh{\degalgo,}{\ast}{I_n}{\grs}
\setminus \left( \bigcup_{j = 1}^{i-1} \bdgh{\degalgo,}{\ast}{I_{j}}{\grs} \right)}
\tag{by \ref{lemma_bound_bdi_combin_prooflm2}} \\
&\leq \sum_{j = i}^n
\card{\bdgh{\degalgo,}{\ast}{I_j}{\grs} \setminus
\left( \bigcup_{k = 1}^{j-1} \bdgh{\degalgo,}{\ast}{I_{k}}{\grs} \right)}
\tag{by \ref{lemma_bound_bdi_combin_prooflm3}} \\
&\leq \sum_{j = i}^n
\card{\bdgh{\degalgo,}{\ast}{I_j}{\grs} \setminus
\bdgh{\degalgo,}{\ast}{I_{j-1}}{\grs}} \\
&\leq \sum_{j = i}^n \cofc{\degalgo}{j}
\tag{by \ref{lemma_bound_bdi_combin_prooflm4}}.
\end{align*}
\end{proof}

\begin{proof}[Proof of \cref{prop_asymp_two_terms2} of \cref{prop_asymp_two_terms}]

We first have that  
\begin{align*} 
\formulaN{2}{\seqpol} &= \sum_{\degalgo = \dregprev+2}^{\Dmac}
\gamma_{\degalgo,n}^{\omega-2} \binom{n +
\degalgo}{\degalgo} \sum_{k = 0}^2 \binom{n-2 +
\Dmac-\degalgo+k}{\Dmac-\degalgo+k} \\
&\leq \sum_{\degalgo = \dregprev+2}^{\Dmac}
\left( \cofc{d}{n} \right)^{\omega-2} \binom{n +
\degalgo}{\degalgo} \sum_{k = 0}^2 \binom{n-2 +
\Dmac-\degalgo+k}{\Dmac-\degalgo+k} \tag{by \cref{lemma_bound_bdi_combin}} \\
&\leq 3 \sum_{\degalgo = \dregprev+2}^{\Dmac}
\left( \cofc{d}{n} \right)^{\omega-2} \binom{n +
\degalgo}{\degalgo} \binom{n-2 +
\Dmac-\degalgo+2}{\Dmac-\degalgo+2}. 
\end{align*}

Now, we suppose that $\omega>2$. The case where $\omega=2$ is dealt with
at the end of the proof, as it is simpler than the case $\omega>2$.
We mimic the reasoning in \cite[Section~3.2]{bardet2014complexity},
while adapting it to our context.
We recall the material introduced there in order to establish our result.

For all $r>0$, we have $\cofc{d}{n} \leq \frac{\serieb{n}{r}}{r^d}$, where
$\serieb{n}{r} = r^{\delta} \prod_{j = 1}^{n-1} \frac{1-r^{\delta}}{1-r}$. For
fixed $d$ and $n$, there exists $\thegoodr{n}{d}>0$ that minimizes
$\frac{\serieb{n}{r}}{r^d}$.
We can now write
\[
\formulaN{2}{\seqpol} \leq 3 \sum_{d = \dregprev+2}^{\Dmac}\left(
\frac{\serieb{n}{\thegoodr{n}{d}}}{\thegoodr{n}{d}^d}
\right)^{\omega-2}
\binom{n + d}{d}
\binom{n-2 + \Dmac-d+2}{\Dmac-d+2}. \tag*{(1)} \label{prop_asymp_two_terms2_fsum}
\]

Moreover, by
\cite[Section~3.2]{bardet2014complexity}, we know that $\thegoodr{n}{d}$
satisfies the following equation:
\[
\frac{d-\delta}{n-1} = \frac{\delta}{1-r^{-\delta}} - \frac{1}{1-r^{-1}}
\tag*{(2)} \label{prop_asymp_two_terms2_eqr}
\]
and that $\thegoodr{n}{d}$ is differentiable and increasing with respect to $d$.

We now identify the degree $\thegoodd{n}$ that dominates all terms in
\ref{prop_asymp_two_terms2_fsum}
by computing the logarithmic derivative
of each term with respect to $d$.
Denoting by $\psi$ the logarithmic derivative of the $\Gamma$
function (see, e.g., \cite[Chapter~6]{abramowitz1965handbook}),
this derivative vanishes when
\begin{align*}
(\omega-2)\log(\thegoodr{n}{d}) & = \psi(d+n+1) - \psi(d+1) -
\psi(n+\Dmac-d+1) + \psi(\Dmac-d+3) \\
& = \sum_{i = 1}^{n}\frac{1}{d+i} - \sum_{i = 2}^{n-1}\frac{1}{\dregmb-d+i+1}.
\tag*{(3)} \label{prop_asymp_two_terms2_3}
\end{align*}
We deduce that \ref{prop_asymp_two_terms2_3} defines a unique function
$\thegoodd{n}$ such that
\begin{itemize}
\item $\thegoodd{n}$ satisfies \ref{prop_asymp_two_terms2_3};
\item $(\thegoodd{n},\thegoodr{n}{\thegoodd{n}})$ satisfies \ref{prop_asymp_two_terms2_eqr}.
\end{itemize}
Since $\dregmb-d+1 \leq d$ for all $d \geq \dregprev+2$, we have
\begin{align*}
(\omega-2)\log(\thegoodr{n}{\thegoodd{n}}) & =
\sum_{i = 1}^{n}\frac{1}{\thegoodd{n}+i}
- \sum_{i = 2}^{n-1}\frac{1}{\dregmb-\thegoodd{n}+i+1}
\tag{\text{by \ref{prop_asymp_two_terms2_3}}} \\
& \leq \frac{1}{\delta} + \frac{1}{\delta + n}\qquad
 \tag*{(4)} \label{prop_asymp_two_terms2_4}.
\end{align*}
The bound $\frac{1}{\delta} + \frac{1}{\delta + n}$ is denoted by $\valint{n}{\delta}$.

We deduce from \ref{prop_asymp_two_terms2_4} that
$\thegoodr{n}{\thegoodd{n}} \leq
e^{\frac{\valint{n}{\delta}}{\omega-2}}$, which implies that
$\thegoodr{n}{\thegoodd{n}}^{-1} \geq
e^{-\frac{\valint{n}{\delta}}{\omega-2}}$.

We consider the function $f: x \mapsto \frac{\delta}{1-x^{\delta}} - \frac{1}{1-x}$. 
The function $f$ is defined, continuous, and decreasing for $x \geq 0$, with
$f(0) = \delta - 1$.
Combined with \ref{prop_asymp_two_terms2_eqr}, we deduce that 
\[
\frac{\thegoodd{n}-\delta}{n-1} \leq 
f\left(e^{-\frac{\valint{n}{\delta}}{\omega-2}}\right).
\tag*{(5)} \label{prop_asymp_two_terms2_5}
\]

We recall the structure of this proof. We have found the bound 
\ref{prop_asymp_two_terms2_fsum} on $\formulaN{2}{\seqpol}$ and we 
identified the degree $\thegoodd{n}$ that corresponds to the greatest 
term in this bound. Nevertheless, we cannot provide an explicit 
expression for $\thegoodd{n}$. This is why, in the following reasoning, we
use \ref{prop_asymp_two_terms2_5} in order to bound from above 
the degree $\thegoodd{n}$ before substituting it into \ref{prop_asymp_two_terms2_fsum}.
To do so, we determine the possible values of $\pourcbm \in [0,1]$,
independently of $\delta$, for which 
there exists an integer $n_{0}$ such that, for all $n \geq n_{0}$, 
the inequality
\begin{align}
f\!\left(e^{-\frac{\valint{n}{\delta}}{\omega-2}}\right) 
\leq \pourcbm (\delta - 1)
\tag*{(6)} \label{prop_asymp_two_terms2_6}
\end{align}
holds.

We define $x(n) = e^{-\frac{\valint{n}{\delta}}{\omega-2}}$.
The function $x(n)$ is increasing, and thus $f(x(n))$ is decreasing.
We therefore need to find $n_{0}$, independent of $\pourcbm$, 
such that $f(x(n_{0})) \leq \pourcbm(\delta-1)$.
When $n$ tends to infinity, $x(n)$ tends to 
$x_{\infty}(\delta,\omega) = e^{-\frac{1}{(\omega-2)\delta}}$. 
We deduce that the limit of $f(x(n))$ is 
$L_{\omega}(\delta) = f(x_{\infty}(\delta,\omega))$.
Hence, $\pourcbm$ must satisfy
\[
\sup_{\substack{\delta \geq 2}} \frac{L_{\omega}(\delta)}{\delta-1} 
< \pourcbm.
\]
This supremum is attained for $\delta = 2$, so $L_{\omega}(2) < \pourcbm$. 
We then know that condition \ref{prop_asymp_two_terms2_6} is equivalent to $L_{\omega}(2) < \pourcbm$, with 
\[
L_{\omega}(2) = \frac{1}{1+e^{-\frac{1}{2(\omega-2)}}}.
\]

Finally, combining \ref{prop_asymp_two_terms2_5} and \ref{prop_asymp_two_terms2_6}, 
for all $p > L_{\omega}(2)$, there exists $n_{0}$ such that for all $n \geq n_{0}$, the following inequality holds:  
\begin{align*}
\thegoodd{n} &\leq p(\delta -1)(n-1) + \delta \\
&\leq \left( p + \frac{1}{n-1} \right)(\delta - 1)(n-1) + 1.
\end{align*}
This means that, for all $p > L_{\omega}(2)$, there exists $n_{0}$ such that for all $n \geq n_{0}$, we have
\begin{align*}
\thegoodd{n} & \leq p (\delta - 1)(n-1) + 1 \\
& \leq p (\delta - 1)n + 1.
\end{align*}

Suppose that $n \geq n_0$, we use that $\thegoodd{n}$ is the degree which maximizes the term of the sum in \ref{prop_asymp_two_terms2_fsum}.
 This
leads to the following inequality
\begin{align*}
\formulaN{2}{\seqpol} &\leq 3n\delta \left( 
\frac{\serieb{n}{\thegoodr{n}{\thegoodd{n}}}}{\thegoodr{n}{\thegoodd{n}}^{\thegoodd{n}}} 
\right)^{\omega-2} 
\binom{n + \thegoodd{n}}{\thegoodd{n}} 
\binom{n-2 + \Dmac-\thegoodd{n}+2}{\Dmac-\thegoodd{n}+2}.
\tag*{(7)} \label{prop_asymp_two_terms2_7}
\end{align*}
Now, we establish two inequalities that will be injected in 
\ref{prop_asymp_two_terms2_7}.
Since $\thegoodd{n} \geq \dregprev+2$ we have that 
\begin{align*}
\binom{n-2 + \Dmac-\thegoodd{n}+2}{\Dmac-\thegoodd{n}+2} &\leq  
\binom{n-2 + \Dmac-\dregprev}{\Dmac-\dregprev} \leq \binom{n-2 + \dregprev + 2}{\dregprev + 2} \leq \binom{n+\frac{n(\delta-1)}{2}+3}{\frac{n(\delta-1)}{2}+3} \\ 
&\leq \left(1 + \frac{2n}{n(\delta - 1)+6} \right)
\left(1 + \frac{2n}{n(\delta - 1)+4} \right)
\left(1 + \frac{2n}{n(\delta - 1)+2} \right)
\binom{n+\frac{n(\delta-1)}{2}}{\frac{n(\delta-1)}{2}} \\
&\leq 27 \binom{n+\frac{n(\delta-1)}{2}}{\frac{n(\delta-1)}{2}}.
\tag*{(8)} \label{prop_asymp_two_terms2_8}
\end{align*}

Since $\thegoodr{n}{\thegoodd{n}}$ minimizes 
$\left( \frac{\serieb{n}{r}}{r^{\thegoodd{n}}} \right)$ 
for all $r>0$, we can bound it from above by setting $r = 1$.
This gives
\begin{align*}
\left( 
\frac{\serieb{n}{\thegoodr{n}{\thegoodd{n}}}}{\thegoodr{n}{\thegoodd{n}}^{\thegoodd{n}}} 
\right)
 \leq \delta^{n-1}. \tag*{(9)} \label{prop_asymp_two_terms2_9}
 \end{align*}

Set $p = L_{\omega}(2) + \epsilon$ with $\epsilon >0$.
For the case $\omega = 2$, just set $L_{\omega}(2) = 1$.
We have $p>\frac{1}{2}$.   
Using \ref{prop_asymp_two_terms2_8} and \ref{prop_asymp_two_terms2_9} in 
\ref{prop_asymp_two_terms2_7} we obtain
\begin{align*}
\formulaN{2}{\seqpol} &\leq 81n \delta^{(\omega-2)(n-1)+1} \binom{n + \thegoodd{n}}{\thegoodd{n}}\binom{n+\frac{n(\delta-1)}{2}}{\frac{n(\delta-1)}{2}} \\
&\leq 81n \delta^{(\omega-2)(n-1)+1} \binom{n + p (\delta - 1)n + 1}{p (\delta - 1)n + 1}\binom{n+\frac{n(\delta-1)}{2}}{\frac{n(\delta-1)}{2}} \tag{because $n \geq n_0$} \\
&\leq 243n \delta^{(\omega-2)(n-1)+1} \binom{n + p (\delta - 1)n}{p (\delta - 1)n}\binom{n+\frac{n(\delta-1)}{2}}{\frac{n(\delta-1)}{2}}. \tag{because $p \geq \frac{1}{2}$}
\end{align*}
We recall that $\delta$ is fixed.
Now, we use \cref{lemma_assymp_bino} on both binomials to obtain that
\[
\formulaN{2}{\seqpol} \in \widetilde{\mathcal{O}}\!\left(
\delta^{\,\omega n + 3-\omega}\,
e^{\,n\left(
E_{1/2}(\delta)
+
E_{L_{\omega}(2)+\epsilon}(\delta)\right)}\right)
\qedhere
\]
\end{proof}

In order to finish the proof of \cref{thm_complexity_F4_tracer_assymp},
we establish the following lemma.

\begin{lemma}
\label{lemma_assymp_bino_ineg}
Let $\delta \geq 2$ be an integer and let $p_1$ and $p_2$ be in $(0,1]$ such that 
$p_1 < p_2$. Then
\[
E_{p_1}(\delta) < E_{p_2}(\delta) 
\]
\end{lemma}

\begin{proof}
Fix an integer $\delta \geq 2$.
We consider the function $E_p(\delta)$ as a function of the variable
$p \in (0,1]$.

This function is well defined and differentiable on $(0,1]$.
A direct computation shows that its derivative with respect to $p$ is given by
\begin{align*}
\frac{\partial}{\partial p} E_p(\delta)
&= \frac{\delta-1}{\delta} \times \frac{\delta}{1+p(\delta-1)}
+ (\delta - 1)\ln\!\left( \frac{1 + p(\delta -1)}{p(\delta -1)} \right)
- \frac{\delta-1}{1 + p(\delta-1)} \\
&= (\delta - 1)\ln\!\left( \frac{1 + p(\delta -1)}{p(\delta -1)} \right).
\end{align*}

Since $\delta \geq 2$ and $p \in (0,1]$, we have
\[
\frac{1 + p(\delta -1)}{p(\delta -1)} > 1,
\]
and therefore the logarithm is positive.
It follows that $\frac{\partial}{\partial p} E_p(\delta) > 0$ on $(0,1]$.

Consequently, the function $E_p(\delta)$ is strictly increasing with respect
to $p$ on $(0,1]$.
\end{proof}

\begin{proof}[Proof of \cref{thm_complexity_F4_tracer_assymp}]
By assumption, the sequence $\hd{\seqpol}$ lies in
$\zarifilling{\delta}{n}{n} \cap \zarimore{\delta}{n} \cap \zarisemi{\delta}{n}{n}$,
where $\delta \geq 2$ and $n \geq 2$.
Moreover, $\phi(\hd{\seqpol},n-1)$ lies in $\mathcal{S}_{\delta,n-1,n}$.
Therefore, we can apply \cref{thm_complexity_F4_tracer} and deduce that
running \algoName{algo:F4T} on a sequence polynomials
that is compatible with the Gr\"obner trace $\mathcal{T}$,
requires $\mathcal{O}\!\left( \compfqt \right)$ operations in $\K$.

Recall that $\compfqt = \formulaN{1}{\seqpol} + \formulaN{2}{\seqpol}$.
Since the variables $x_1, \ldots , x_n$ are in simultaneous Noether position
with respect to $\hd{\seqpol}$, we can apply \cref{prop_asymp_two_terms}
and obtain that
\[
\compfqt \in \widetilde{\mathcal{O}}
\left(
\delta^{\omega n+1}
\left(
e^{n \left( E_{1/2}(\delta)+E_{L_{\omega}(2)+\epsilon}(\delta) \right)}
+ e^{2nE_{1/2}(\delta)}
\right)
\right)
\]
for all $\epsilon > 0$.
Since $L_{\omega}(2) > \frac{1}{2}$, we can apply
\cref{lemma_assymp_bino_ineg} to deduce that
\[
\compfqt \in \widetilde{\mathcal{O}}
\left(
\delta^{\omega n+1}
e^{n \left( E_{1/2}(\delta)+E_{L_{\omega}(2)+\epsilon}(\delta) \right)}
\right) \quad \text{for all } \epsilon > 0.
\qedhere
\]
\end{proof}

\subsection{Comparisons of formulas}
\label{subsection_comparisons_formulas}

We first compare our complexity estimates with the one of Lazard's
algorithm, introduced in~\cite{lazard83}, which is usually applied to 
the family of F4 algorithms under some regularity assumptions 
(see e.g.\ \cite{FaSaSpa10, FaSaSpa11, FaSaSpa12, FaSaSpa13}). 

Our second comparison is with the more precise complexity on a variant
of the F5 algorithm~\cite{faugere2002F5}, which is established
in~\cite{bardet2014complexity}. A major advantage of F5 is
that, thanks to its signature system, it avoids all reductions to zero
under the regularity assumption. However, this same mechanism imposes
a significant constraint: it prevents the algorithm from fully
leveraging efficient linear algebra techniques for performing
row-echelon transformations of Macaulay matrices
\cite{eder2017survey}. In particular, the signature system restricts
the possible row operations, thereby limiting the use of optimized
dense or sparse linear algebra methods.  

Since \algoName{algo:F4T} is also designed to eliminate reductions to
zero (but under a more restrictive framework), the comparison between
the two algorithms remains natural and meaningful.

\smallskip

The context of computing a grevlex Gröbner basis of a sequence
$\seqpol = (f_1, \ldots , f_n)$ in $\Kp{n}$ remains the same as in the remainder
 of \cref{sec:complexity_analysis}.
Our comparisons rely on evaluating complexity formulas.

We recall that $\compfqt$ can be expressed as the sum 
$\formulaN{1}{\seqpol} + \formulaN{2}{\seqpol}$, as seen in 
\cref{subsec_proof_comp_form}.
Note that
\[ \formulaN{2}{\seqpol}  = \sum_{\degalgo = \dregprev+2}^{\Dmac}
{\gamma_{\degalgo,n}}^{\omega-2} \binom{n + \degalgo}{\degalgo}
\sum_{k = 0}^2 \binom{n-2 + \Dmac-\degalgo+k}{\Dmac-\degalgo+k}, \]
with $\gamma_{\degalgo,n} = \card{\bdg{\degalgo,}{n}{\ideal{\seqpol}}{\grs}}$.
Since we assumed in \cref{thm_complexity_F4_tracer} that
$\hd{\seqpol}$ lies in $\zarisemi{\delta}{n}{n}$, the sequence $\hd{\seqpol}$
is regular, and by \cref{cor_lien_homog_affine_reg} we have
\[ \card{\bdg{\degalgo,}{n}{\ideal{\seqpol}}{\grs}} = \card{\bdg{\degalgo,}{n}{\ideal{\hd{\seqpol}}}{\grs}} .\]
Moreover, since we assumed in \cref{thm_complexity_F4_tracer} that
$\hd{\seqpol}$ lies in $\zarimore{\delta}{n}$ and 
that $\phi(\hd{\seqpol},n-1)$ lies in $\mathcal{S}_{\delta,n-1,n}$, we can
apply \cref{thm_card_bdg_deg_moitMB} to $\hd{\seqpol}$ to deduce that
\begin{align*}
\gamma_{\degalgo,n} &= \card{\bdg{\degalgo,}{n}{\ideal{\hd{\seqpol}}}{\grs}} = \beta_{\Dmac-\degalgo},
\end{align*}
with $\left[ \frac{(1-z^{\delta})^n}{(1-z)^{n-1}} \right]_+ = \sum_{d = 0}^{\dregprev-1} \beta_d z^d$.

We introduce a new formula 
$\compfqtb = \formulaNb{1}{\seqpol} + \formulaN{2}{\seqpol}$ such that 
$\formulaN{1}{\seqpol} \leq \formulaNb{1}{\seqpol}$
under some generic assumptions.
To do so, we use the notation $\cofc{d}{i}$
introduced just before \cref{lemma_bound_bdi_combin}.
Moreover, for all $i$ in $\{0,\ldots ,n\}$, we introduce the following notation:
\begin{align*}
\left[ \frac{(1-z^{\delta})^n}{(1-z)^{i}} \right]_+
  = \sum_{\degalgo = 0}^{+\infty} 
    \beta_{\degalgo, i}\, z^{\degalgo}.
\end{align*}
We also use the truncation $\phi$ introduced in \cref{subsection_Macaulay_matrices}.
We now express $\formulaNb{1}{\seqpol}$ in the following proposition.

\begin{prop}
\label{prop_comp_formula_eval}
Suppose that, for all $i \in \{1,\ldots,n\}$, the sequence $\phi(\hd{\seqpol},i)$ 
lies in $\mathcal{S}_{\delta,i,n}$ and that the variables $x_1, \ldots , x_n$ are
in simultaneous Noether position with respect to $\hd{\seqpol}$.
Then we have $\formulaN{1}{\seqpol} \leq \formulaNb{1}{\seqpol}$ with
\[ \formulaNb{1}{\seqpol} =  \sum_{\degalgo = \delta}^{\dregprev +1} 
\sum_{i = 1}^n \binom{n + \degalgo}{\degalgo} 
\left( \frac{ \nbredw{\degalgo-\delta}{i} 
\sum_{j = 1}^n \cofc{\degalgo}{j}} {\min\left( \nbred{\degalgo-\delta}{i}, 
\sum_{j = 1}^n \cofc{\degalgo}{j} \right)^{3-\omega} } + 
\left( \sum_{j = 1}^n \cofc{\degalgo}{j} \right)^{\omega-1} \right)\]
where
\[ \nbredw{\degalgo-\delta}{i} =  \binom{i-1 + \degalgo}{\degalgo}
-
\binom{i-2 + \degalgo}{\degalgo}
-
\bigl( \beta_{\degalgo, i} - \beta_{\degalgo, i-1} \bigr) . \]
\end{prop}

\begin{proof}
We recall that
\[  \formulaN{1}{\seqpol} = \sum_{\degalgo = \delta}^{\dregprev+1} \sum_{i = 1}^{n} \binom{n +
\degalgo}{\degalgo} \left( \frac{ \nbred{\degalgo-\delta}{i}
\sum_{j = i}^n {\gamma_{\degalgo,j}}} {\min\left(
\nbred{\degalgo-\delta}{i},  \sum_{j = i}^n
{\gamma_{\degalgo,j}} \right)^{3-\omega} } + \left(
\sum_{j = i}^n {\gamma_{\degalgo,j}} \right)^{\omega-1}
\right). \]
Since the variables $x_1, \ldots , x_n$ are
in simultaneous Noether position with respect to $\hd{\seqpol}$, we can apply 
\cref{lemma_bound_bdi_combin} to deduce that
\[ \formulaN{1}{\seqpol} \leq \sum_{\degalgo = \delta}^{\dregprev +1} 
\sum_{i = 1}^n \binom{n + \degalgo}{\degalgo} 
\left( \frac{ \nbred{\degalgo-\delta}{i} 
\sum_{j = 1}^n \cofc{\degalgo}{j}} {\min\left( \nbred{\degalgo-\delta}{i}, 
\sum_{j = 1}^n \cofc{\degalgo}{j} \right)^{3-\omega} } + 
\left( \sum_{j = 1}^n \cofc{\degalgo}{j} \right)^{\omega-1} \right). \]

Let $d$ be in $\{\delta , \ldots , \dregprev +1\}$ and let $i$ be in 
$\{1, \ldots , n\}$.
Using the definition of 
$\nbred{\degalgo-\delta}{i}$ given in \cref{subsec_proof_comp_form}, we have
\[ \nbred{\degalgo-\delta}{i} = \binom{i-1 + \degalgo}{\degalgo}
-
\binom{i-2 + \degalgo}{\degalgo} 
- \bigl( \card{\bdgb{\degalgo,}{i}{\ideal{\hd{\seqpol}}}{\grs}} - \card{\bdgb{\degalgo,}{i-1}{\ideal{\hd{\seqpol}}}{\grs}} \bigr) . \]
Applying \cref{cor_lien_mono_base_quotient_tronqué}, we obtain
\[ \card{\bdgb{\degalgo,}{i}{\ideal{\hd{\seqpol}}}{\grs}} = \card{\bdgb{\degalgo,}{i}{\ideal{\phi(\hd{\seqpol},i)}}{\grs}}.\]
Since, by hypothesis, the sequence $\phi(\hd{\seqpol},i)$ 
lies in $\mathcal{S}_{\delta,i,n}$, we deduce from \cite[Proposition~1]{pardue2010generic}
that
\[ \card{\bdgb{\degalgo,}{i}{\ideal{\phi(\hd{\seqpol},i)}}{\grs}} = \beta_{\degalgo, i}. \]
Applying the same reasoning to $\card{\bdgb{\degalgo,}{i-1}{\ideal{\hd{\seqpol}}}{\grs}}$
concludes the proof.
\end{proof}
We now have a new formula $\compfqtb = \formulaNb{1}{\seqpol} + \formulaN{2}{\seqpol}$ which is much easier to evaluate
and is an upper bound on $\compfqt$.

\smallskip

We start the comparisons with the complexity of Lazard's algorithm.
Consider a sequence of polynomials $\seqpol = (f_1, \ldots , f_n)$ such that $\hd{\seqpol}$ lies
in the open Zariski set $\zarisemi{\delta}{n}{n}$ introduced in
\cref{subsectionThelastdimensionofthegeneriquemonomialstaircase}.
The fact that $\hd{\seqpol}$ lies in $\zarisemi{\delta}{n}{n}$
leads to the following implications.

First, the ideal $\ideal{\hd{\seqpol}}$ is zero-dimensional, since its Hilbert
series is a polynomial \cite[Proposition~1]{pardue2010generic}.
Therefore, by \cite[Theorem~3]{lazard83}, the maximal degree of an element of a
minimal grevlex Gröbner basis of $\ideal{\hd{\seqpol}}$ is the Macaulay bound
\[
\dregmb = n(\delta - 1) + 1.
\]
Using \cref{cor_lien_homog_affine_reg}, the same result holds for
$\ideal{\seqpol}$.

Moreover, using Assertion~$\asseralgo{\cdot}{5}$ of
\cref{lemma_icreasing_degree_in_algorithm}, the $\K$-vector space
\[
\langle S \rangle_{\K}
\quad \text{where} \quad
S = \{ m f_i \mid i \in \{1, \ldots , \nbpol\},\;
m \in \mon{\leq \dregmb-\delta,}{n} \}
\]
contains a grevlex Gröbner basis of $\ideal{\seqpol}$.
Thus, Lazard's algorithm consists in computing a row echelon form of the matrix
\[
\mac{}{}{S}{\mon{\leq \dregmb,}{n}}{}{}.
\]
By \cite[Chapter~2]{Storjohann2000}, its arithmetic complexity is
\[
\mathcal{O}\!\left( \compfq \right),
\qquad
\text{with}
\qquad
\compfq =
n \binom{n+\dregmb}{\dregmb}^{\omega}.
\]

\cref{fig:image1} compares the logarithm of the ratio of 
$\compfq$ over $\compfqtb$ for \(\omega = 3\) and
different values of  \(\delta\). It illustrates the exponential
improvement on \(n\) that our analysis of \algoName{algo:F4T}
provides. Note that this gain seems to slightly increase when  \(\delta\) increases. 
For instance, when $\omega = 3$ and $\delta = 30$, the value of the slope of 
$\ln \left(\frac{\compfq}{\compfqtb} \right)$ as a function of $n$ is approximately
$2.550$.

\cref{fig:image2} provides the same comparison but fixing \(\delta\)
to  \(10\) and for various values of  \(\omega\). It illustrates that
this exponential gain gets smaller when  \(\omega\) decreases but
remains, even assuming  \(\omega=2\). 
For instance, when $\omega = 2$ and $\delta = 10$, the value of the slope of 
$\ln \left(\frac{\compfq}{\compfqtb} \right)$ as a function of $n$ is approximately
$1.335$.

\begin{figure}[htb]
    \centering
\bigskip
\bigskip
    \begin{minipage}{0.45\textwidth}
        \centering
        \begin{overpic}[width=0.7\textwidth]{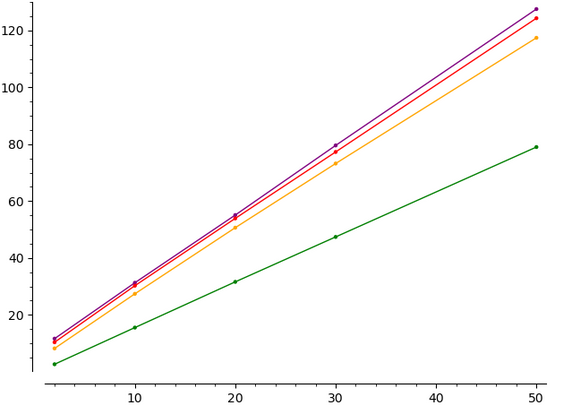}
          \put(-15,80){$\ln \left(\frac{\compfq}{\compfqtb} \right)$}
          \put(105,0){$n$}
          \put(100,45){\color{green!60!black}$\delta = 2$}
          \put(100,61){\color{orange}$\delta = 10$}
          \put(100,68){\color{red}$\delta = 20$}
          \put(100,75){\color{purple}$\delta = 30$}
        \end{overpic}
        \caption{$\omega = 3$}
        \label{fig:image1}
    \end{minipage} \hfill
    \begin{minipage}{0.45\textwidth}
        \centering
        \begin{overpic}[width=0.7\textwidth]{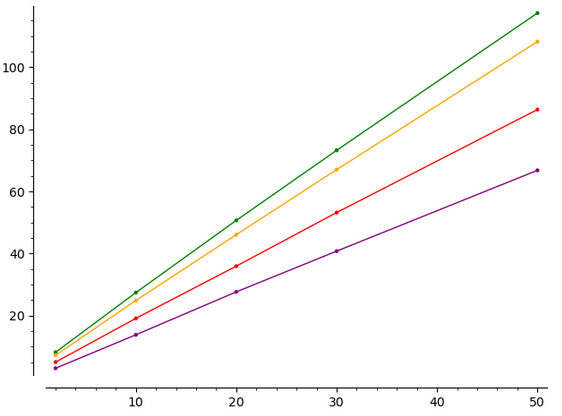}
          \put(-15,80){$\ln \left(\frac{\compfq}{\compfqtb} \right)$}
          \put(105,0){$n$}
          \put(100,72){\color{green!60!black}$\omega = 3$}
          \put(100,65){\color{orange}$\omega = 2.81$}
          \put(100,55){\color{red}$\omega = 2.38$}
          \put(100,43){\color{purple}$\omega = 2$}
        \end{overpic}
        \caption{$\delta = 10$}
        \label{fig:image2}
    \end{minipage}
\end{figure}

We now compare $\compfqtb$ with $\compfc$, which is the complexity
result obtained in \cite{bardet2014complexity}.
We recall that this analysis assumes that $\hd{\seqpol}$ lies in
$\zarisemi{\delta}{n}{n}$ and that the variables $x_1, \ldots, x_n$ are in
simultaneous Noether position with respect to the sequence $\hd{\seqpol}$
\cite[Definition~3]{bardet2014complexity}.
Nevertheless, the complexity result for \algoName{algo:F4T} requires stronger
hypotheses.
In addition to the previously stated assumptions, it is required that the
sequence $\hd{\seqpol}$ lies in
$\zarifilling{\delta}{n}{n} \cap \zarimore{\delta}{n}$, where $\delta \geq 2$
and $n \geq 2$.
We also assume that for all $i \in \{1,\ldots,n\}$, the sequence $\phi(\hd{\seqpol},i)$ 
lies in $\mathcal{S}_{\delta,i,n}$. In particular, $\phi(\hd{\seqpol},n-1)$ lies in
$\mathcal{S}_{\delta,n-1,n}$.

A fact here is that the variables $x_1, \ldots, x_n$ being in simultaneous
Noether position with respect to the sequence
$(\hd{f_1}, \ldots , \hd{f_n})$ implies that the sequence
$(\hd{f_1}, \ldots , \hd{f_n})$ is regular, by
\cite[Proposition~6]{bardet2014complexity}.
The complexity of the F5 algorithm, as stated in
\cite[Section~3.2]{bardet2014complexity}, is given by
\[
N_{F5}
=
\sum_{i=1}^n
\sum_{\degalgo=\thedeg}^{\dregmb}
b_\degalgo^{(i)}
\binom{i-1 + \degalgo}{\degalgo}
\binom{n+\degalgo}{\degalgo},
\]
using the notation $\cofc{d}{i}$
introduced just before \cref{lemma_bound_bdi_combin}.

\cref{fig:image1bis} and \cref{fig:image2bis} perform the same
comparisons between $\compfqtb$ and $\compfc$ as
\cref{fig:image1,fig:image2} do for $\compfqtb$ and
$\compfq$.
\begin{figure}[htb]
  \bigskip
  \bigskip
    \centering
    \begin{minipage}{0.45\textwidth}
        \centering
        \begin{overpic}[width=0.7\textwidth]{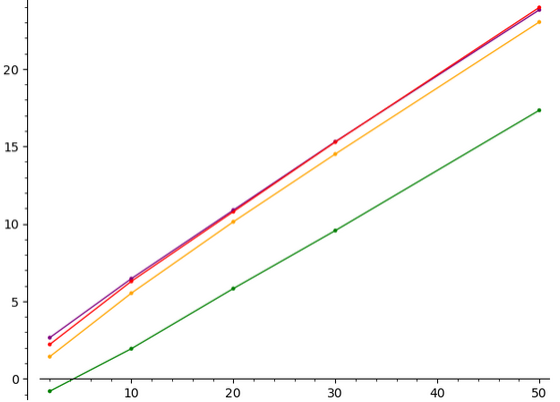}
          \put(-15,80){$\ln \left(\frac{\compfc}{\compfqtb} \right)$}
          \put(105,0){$n$}
          \put(100,50){\color{green!60!black}$\delta = 2$}
          \put(100,61){\color{orange}$\delta = 10$}
          \put(100,68){\color{red}$\delta = 20$}
          \put(100,75){\color{purple}$\delta = 30$}
        \end{overpic}
        \caption{$\omega = 3$}
        \label{fig:image1bis}
    \end{minipage} \hfill
    \begin{minipage}{0.45\textwidth}
        \centering
        \begin{overpic}[width=0.7\textwidth]{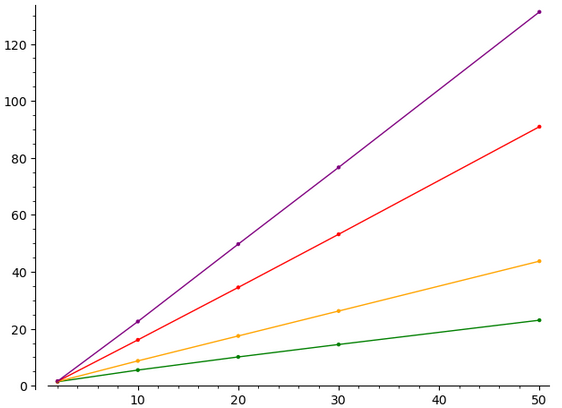}
          \put(-15,80){$\ln \left(\frac{\compfc}{\compfqtb} \right)$}
          \put(105,0){$n$}
          \put(100,15){\color{green!60!black}$\omega = 3$}
          \put(100,25){\color{orange}$\omega = 2.81$}
          \put(100,50){\color{red}$\omega = 2.38$}
          \put(100,70){\color{purple}$\omega = 2$}
        \end{overpic}
        \caption{$\delta = 10$}
        \label{fig:image2bis}
    \end{minipage}
\end{figure}

Fixing \(\omega\) to  \(3\), \cref{fig:image1bis} illustrates 
the exponential gain in \(n\) provided by $\compfqtb$, 
for different values of \(\delta\).
\cref{fig:image2bis} illustrates
that, when fixing \(\delta\), this exponential gain tends to
\emph{increase} when  \(\omega\) decreases.
For instance, when $\omega = 3$ and $\delta = 30$, the value of the slope of 
$\ln \left(\frac{\compfq}{\compfqtb} \right)$ as a function of $n$ is approximately
$0.476$.

\cref{fig:tablecomp} provides an overview of the aforementioned 
complexity statements but this time, expressed
as a power of the degree of the ideal.  The value reported in each
entry of the table is given by $a =
\frac{\log_{\delta}(N)}{n}$ depending on $\delta,n$ and
$\omega$.

\begin{figure}[htb]
\centering
\begin{NiceTabular}{|@{}c|cccc|cccc|c@{}|}
\hline
  & \Block{2-4}{ $\compfqtb$ } & & & & \Block{2-4}{$\compfq$ } & & & & \Block{2-1}{$\compfc$ \quad }  \\
  & & & & & & & & \\
    & $\omega = 3$ & $\omega = 2.81$ & $\omega = 2.38$ & $\omega = 2$ & $\omega = 3$ & $\omega = 2.81$ & $\omega = 2.38$ & $\omega = 2$ &   \\  
\hline
\multicolumn{10}{c}{$\delta = 2$} \\
\hline
$n =2$  & $3.58$ & $3.47$ & $3.25$ & $3.08$ & $5.48$ & $5.16$ & $4.45$ & $3.82$ &  $3.00$ \\
~$n =10$ & $3.62$ & $3.48$ & $3.16$ & $2.90$ & $5.86$ & $5.51$ & $4.71$ & $4.01$ & $3.90$  \\
~$n =20$ & $3.63$ & $3.47$ & $3.11$ & $2.80$ & $5.91$ & $5.55$ & $4.73$ & $4.01$ &  $4.05$ \\
~$n =30$ & $3.65$ & $3.48$ & $3.10$ & $2.77$ & $ 5.93$& $ 5.56$ & $ 4.74$ & $4.00$ & $4.11$ \\
~$n =50$ & $3.67$ & $3.49$ & $3.10$ & $2.76$ & $5.95$ & $5.58$ & $4.74$ & $4.00$ & $4.17$  \\
\hline
\multicolumn{10}{c}{$\delta = 10$} \\
\hline
$n =2$ & $1.85$ & $1.84$ & $1.82$ & $1.81$ & $3.63$ & $3.41$ &  $2.91$ & $2.47$ & $2.16$  \\
~$n =10$ & $2.89$ & $2.75$ & $2.43$ & $2.15$ & $4.08$ & $3.83$ & $3.26$ & $2.75$ & $3.13$ \\
~$n =20$ & $3.05$ & $2.89$ & $2.52$ & $2.19$ & $4.15$ & $3.89$ & $3.30$ & $2.79$ &  $3.27$ \\
~$n =30$ & $3.11$ & $2.94$ & $2.55$ & $2.21$ & $4.17$ & $3.91$ & $3.32$ & $2.80$ & $3.32$ \\
~$n =50$ & $3.17$ & $2.99$ & $2.58$ & $2.23$ & $4.19$ & $3.93$ & $3.33$ & $2.81$ & $3.37$  \\
\hline
\multicolumn{10}{c}{$\delta = 20$} \\
\hline
$n =2$ & $1.74$ & $1.73$ & $1.73$ & $1.72$ & $3.47$ & $3.26$ & $2.78$ & $2.35$ &  $2.11$ \\
~$n =10$& $2.84$ & $2.70$ & $2.37$ & $2.09$ & $3.85$ & $3.61$ & $3.07$ & $2.59$ & $3.05$  \\
~$n =20$ & $3.00$ & $2.84$ & $2.46$ & $2.13$ & $3.90$ & $3.66$ &  $3.11$ & $2.62$ & $3.18$  \\
~$n =30$ & $3.06$ & $2.89$ & $2.50$ & $2.15$ & $3.92$ & $3.68$ & $3.12$ & $2.63$ & $3.23$ \\
~$n =50$ & $3.11$ & $2.93$ & $2.53$ & $2.17$ & $3.94$ & $3.69$ & $3.13$ & $2.63$ & $3.27$ \\
\hline
\multicolumn{10}{c}{$\delta = 30$} \\
\hline
$n =2$ & $1.70$ & $1.70$ & $1.69$ & $1.69$ & $3.41$ & $3.20$ & $2.73$ & $2.31$ & $2.09$ \\
~$n =10$ & $2.83$ & $2.68$ & $2.35$ & $2.07$ & $3.75$ & $3.52$ & $2.99$ & $2.52$ & $3.02$ \\
~$n =20$ & $2.99$ & $2.82$ & $2.44$ & $2.11$ &$3.80$ & $3.56$ & $3.02$ & $2.55$ & $3.15$ \\
~$n =30$ & $3.04$ & $2.87$ & $2.48$ & $2.13$ & $3.82$ & $3.58$ & $3.04$ & $2.56$ & $3.19$ \\
~$n =50$ & $3.09$ & $2.91$ & $2.51$ & $2.15$ & $3.84$ &$3.59$  & $3.05$ & $2.56$ & $3.23$ \\
\hline
\multicolumn{10}{c}{$\delta = 50$} \\
\hline
$n =2$ & $1.66$ & $1.66$ & $1.66$ & $1.66$ & $3.35$ & $3.15$ & $2.68$ & $2.26$ &  $2.08$ \\
~$n =10$ & $2.81$ & $2.67$ & $2.34$ & $2.05$ & $3.66$ & $3.43$ & $2.91$ & $2.46$ & $2.99$  \\
~$n =20$ & $2.97$ & $2.80$ & $2.43$ & $2.09$ & $3.70$ & $3.47$ &  $2.94$ & $2.48$ & $3.11$  \\
~$n =30$ & $3.03$ & $2.85$ & $2.46$ & $2.11$ & $3.72$ & $3.48$ & $2.95$ & $2.49$ & $3.16$ \\
~$n =50$ & $3.07$ & $2.89$ & $2.48$ & $2.12$ & $3.73$ & $3.50$ & $2.96$ & $2.49$ & $3.20$ \\
\hline
\end{NiceTabular}
\caption{Complexities with respect to the degree of the ideal}
\label{fig:tablecomp}
\end{figure}

We observe that, in all cases, the smaller $\omega$ is, the lower the bounds on the number of operations.
Moreover, for fixed values of $n$, $\delta$, and $\omega$, the values of $a$ for
$\compfqtb$ are always smaller than those for $\compfq$.
The same holds for the comparison between $\compfqtb$ and $\compfc$,
except when $n = 2$ and $\delta = 2$.

For all columns of \cref{fig:tablecomp}, we observe that the value of $a$
increases as $n$ increases, for almost all fixed values of $\delta$ and $\omega$.
However, for fixed $n$ and $\omega$, it decreases as $\delta$ increases.
Nevertheless, it is bounded from below.
This can be shown as follows.
All the complexity formulas $\compfqtb$, $\compfq$ and $\compfc$, are greater than $\binom{n+\Dmac}{\Dmac}$,
which is itself greater than $\delta^n$.
Since
\[
\frac{\log_{\delta}(\delta^{n})}{n} = 1,
\]
each value reported in \cref{fig:tablecomp}
is necessarily greater than $1$.
Finally, for fixed $n$ and $\delta$, it decreases as $\omega$ decreases.
To better illustrate the gain shown in the data of \cref{fig:tablecomp},
we fix a single row.
For instance, in the row corresponding to $\delta = 2$ and $n = 50$,
the gain of $\compfqtb$ over $\compfc$ is
$\delta^{(4.17-3.67)n} = \delta^{0.50n}$ when $\omega = 3$,
and $\delta^{1.07n}$ when $\omega = 2.38$.

\subsection{Cardinality of a generic Gröbner basis}

To conclude, the size of the output is a major indicator of the attainable
performance in terms of complexity.
We derive a new bound on the output size in the generic case, assuming several
conjectures introduced earlier.
Let $G$ denote the reduced grevlex Gröbner basis of an ideal $I$ generated by
a sequence $\seqpol$ such that $\hd{\seqpol}$ lies in $\zarisemi{\delta}{n}{n}$.
The size of $G$ is
\[
\sum_{g \in G} \card{\supo{g}}.
\]

Since $G$ is reduced, the following property holds.
For each polynomial $g \in G$, every monomial in $\supo{g}$ distinct from
$\lm{\grs}{g}$ belongs to $\mon{}{n} \setminus \lm{\grs}{I}$.
Moreover, since $\hd{\seqpol}$ lies in $\zarisemi{\delta}{n}{n}$,
\cref{cor_lien_homog_affine_reg} implies that
\[
\mon{}{n} \setminus \lm{\grs}{I}
=
\mon{}{n} \setminus \lm{\grs}{\ideal{\hd{\seqpol}}}.
\]

Using \cite[Proposition~1]{pardue2010generic}, we deduce that the Hilbert series
of $\ideal{\hd{\seqpol}}$ is equal to $\left( \sum_{d = 0}^{\delta-1} z^d \right)^n$.
Since this is also the Hilbert series of the ideal generated by
$(x_1^{\delta}, \ldots, x_n^{\delta})$, it follows that
\[
\card{\mon{}{n} \setminus \lm{\grs}{I}} \leq \delta^n.
\]
Consequently, for any $g \in G$, we have
\[
\card{\supo{g}} \leq \delta^n + 1.
\]

We have just shown that the size of the output is bounded by
$\card{G}(\delta^n+1)$.
We now focus on bounding the cardinality of $G$.

We aim to apply \cref{thm_bijection_monomials}, which requires that
$\lm{\grs}{I}$ be stable and that $I$ be zero-dimensional.
We recall that $I$ is zero-dimensional since its Hilbert series is a polynomial.
Moreover, by \cref{remark_fillingup_implies_stable}, if $\hd{\seqpol}$ lies in
$\zarifilling{\delta}{n}{n}$, then $\lm{\grs}{\hd{\seqpol}}$ is stable.
By \cref{cor_lien_homog_affine_reg}, this would also holds for $\lm{\grs}{I}$.
In addition, \cite[Conjecture~4.1]{moreno2003degrevlex} states that
$\zarifilling{\delta}{n}{n}$ is a nonempty Zariski open set.

We therefore add the hypotheses that $\delta \geq 0$, $n \geq 1$, and that
$\lm{\grs}{I}$ is stable.
Using \cref{thm_bijection_monomials}, we deduce that
\begin{align*}
\card{G}
= \sum_{i = 0}^{n-1} \card{\bdgb{}{i}{I}{\succ}}
= \sum_{i = 0}^{n-1} \sum_{d = 0}^{+\infty} \card{\bdgb{d,}{i}{I}{\succ}}
=1 + \sum_{i = 1}^{n-1} \sum_{d = 0}^{+\infty} \card{\bdgb{d,}{i}{I}{\succ}}.
\tag*{($\star$)} \label{eq:cardbdg}
\end{align*}

At this stage, we consider an integer $d$ and $i$ in $\{1, \cdots ,n-1\}$ and we find 
an expression of $\card{\bdgb{d,}{i}{I}{\succ}}$. 
Using \cref{cor_lien_mono_base_quotient_tronqué}, we already know that 
\[ \card{\bdgb{d,}{i}{I}{\succ}} = \card{\bdgb{d,}{i}{\ideal{\phi(\hd{\seqpol},i)}}{\succ}}. \]
Supposing that for all $i$ in $\{1, \ldots , n \}$ we have that the 
sequence $\phi(\hd{\seqpol},i)$ lies in
$\mathcal{S}_{\delta,i,n}$, we obtain that
for all $i$ in $\{1, \ldots , n\}$
\[ \sum_{d = 0}^{+ \infty} \card{\bdgb{d,}{i}{\ideal{\phi(\hd{\seqpol},i)}}{\succ}}z^d =
\sum_{d = 0}^{+ \infty} \card{\bdgb{d,}{i}{I}{\succ}}z^d = \left[ \frac{(1-z^{\delta})^{n}}{(1-z)^i} \right]_+  .\]
By injecting this last formula in \ref{eq:cardbdg}, we are able to compute 
$\card{G}$ without computing $G$ itself.
We recall that for all $i$ in $\{1, \ldots , n-1\}$, the fact that 
$\zarisemi{\delta}{i}{n}$ is a non-empty open Zariski set rely on 
\cite[Conjecture 1.1]{froberg1994hilbert}.
This gives that following result.

\begin{thm}
\label{thm_cardbdg}
Let $\seqpol$ be a sequence of polynomials such that 
for all $i$ in $\{1, \ldots , n \}$ the 
sequence $\phi(\hd{\seqpol},i)$ lies in $\mathcal{S}_{\delta,i,n}$.
Suppose that $\lm{\grs}{I}$ is stable and let $G$ be a minimal grevlex Gröbner 
basis of $\ideal{f}$.
Then we have that 
\[ \card{G} = 1 + \sum_{i = 1}^{n-1} \sum_{d = 0}^{+\infty} \beta_{d,i} 
\quad \text{with} \quad  \left[ \frac{(1-z^{\delta})^{n}}{(1-z)^i} \right]_+ = 
\sum_{d = 0}^{+\infty} \beta_{d,i}z^d \quad \text{for all $i$ in $\{1, \ldots , n-1\}$ }. \]
\end{thm}

\cref{fig:tablecardbdg} presents numerical data illustrating how the
cardinality of a Gröbner basis grows as a power of the degree of the ideal.
The value reported in the table in the column indexed by $n$
and the row indexed by $\delta$ is equal to
$b = \frac{\log_{\delta}(\card{G})}{n}$.

\begin{figure}
  \centering
\begin{NiceTabular}{|c|c|c|c|c|c|c|c|}
\hline
 & $n=2$ & $n=5$ & $n=10$ & $n=15$ & $n=20$ &  $n=30$ & $n=50$ \\ 
\hline
$\delta=2$ & $0.79$& $0.87$& $0.87$& $0.88$& $0.89$&  $0.91$ & $0.94$ \\
\hline
$\delta=5$ & $0.55$& $0.76$& $0.85$& $0.89$& $0.91$&  $0.94$ & $0.96$ \\
\hline
$\delta=10$ & $0.52$& $0.76$& $0.86$& $0.90$& $0.92$&  $0.95$ & $0.96$ \\
\hline
$\delta=15$ & $0.51$& $0.76$& $0.87$& $0.90$& $0.93$& $0.95$ & $0.96$ \\
\hline
$\delta=20$ & $0.50$& $0.76$& $0.87$& $0.91$& $0.93$& $0.95$ & $0.96$ \\
\hline
$\delta=30$ & $0.50$& $0.77$& $0.87$& $0.91$& $0.93$& $0.95$ & $0.97$ \\
\hline
$\delta=50$ & $0.50$ & $0.77$ & $0.87$ & $0.91$ & $0.93$ & $0.95$ & $0.97$ \\
\hline
\end{NiceTabular}
\caption{Cardinality of $G$ with respect to the degree of the ideal}
\label{fig:tablecardbdg}
\end{figure}

\appendix

\section{}
\label{sectionappendix}

\counterwithout{thm}{section}
\renewcommand{\thelemma}{\arabic{lemma}}

\begin{lemma}
\label{thelemma_Macmat_echlo_contains_GB_homog}
Let $\seqpol = (f_1, \ldots, f_\nbpol)$ be a sequence of
homogeneous polynomials of degree $\delta$ in $\Kp{n}$, and
let $e\in\mathbb{N}$. Let $G=\bdg{}{}{I}{\grs}$ be 
the reduced $\grs$-Gröbner basis of the ideal $I = \langle
\seqpol \rangle$. 
We denote by $E_{e}$ the row reduced echelon form of the
matrix  \[
\mac{}{}{\enspol{e}{n}{\seqpol}}{\mon{e,}{n}}{\grs}{\succ_{\seqpol}}.
\] Then, the following assertions hold.
\begin{enumerate}
\item
The matrix $E_{e}$ contains, as rows, all polynomials $g$ in $G$ such that $\deg(g) = e$.
\end{enumerate}
\end{lemma}

\begin{proof}[Proof of item 1.\ of \cref{thelemma_Macmat_echlo_contains_GB_homog}]
Let $g$ be in $G$ with $\deg(g) = e$. 
Since $G$ is the reduced Gröbner basis of an homogeneous ideal,
by \cite[Chapter 8, Section 3, Theorem 2]{cox94}, $g$ is homogeneous.

By \cref{lemma_generator_homog_space}, the polynomial $g$ belongs to the 
$\K$-vector space generated by the rows of the Macaulay matrix $\mac{}{}{\enspol{e}{n}{\seqpol}}{\mon{e,}{n}}{\grs}{\succ_{\seqpol}}$. Therefore, the polynomial $g$ also lies in the $\K$-vector space generated by the rows of its row-reduced echelon form $E_{e}$.
By \cref{lemma_generator_homog_space} and \cref{remark_pivot_equal_leadingmonomial}, the rows of $E_{e}$ have pivot columns indexed by all possible leading monomials (with respect to $\succ_\drl$) of degree $e$.
Hence, there exists a row of $E_{e}$ representing a polynomial $f$ such that 
\[ \lmdrl{f} = \lmdrl{g}. \]
Since $f$ and $g$ are both homogeneous of degree $e$, 
the difference $f - g$ is also homogeneous of degree $e$ or zero.

As $E_{e}$ is in reduced row echelon form, we have $\lc{\succ_\drl}{f}
= 1$. Since $G$ is a reduced $\grs$-Gröbner basis, we also have
$\lc{\succ_\drl}{g} = 1$ by the definition of reduced Gröbner bases. 
Thus, we have that $\lt{\succ_\drl}{f} = \lt{\succ_\drl}{g}$. This implies that the leading monomial of $f - g$ is not $\lmdrl{f}$.
Suppose that $f-g$ is not zero and write $m = \lmdrl{f - g}$.

\medskip

Suppose that $m$ lies in $\supo{f} \setminus \lmdrl{f}$.
 Then, since $E_{e}$ is in row reduced echelon form, the monomial $m$ does not index a pivot column of $E_{e}$. But since $f - g$ is in $I$ and homogeneous of degree $e$, its leading monomial must index a pivot column of $E_{e}$. This is a contradiction.

\medskip

Suppose that $m$ lies in $\supo{g} \setminus \lmdrl{g}$. Then, by the
definition of Gröbner bases, $m$ is not divisible by any leading monomial of elements of $G$, and then of $I$. In particular, 
it cannot index a pivot column of $E_{e}$. Again, since $f - g$ lies in $I$ and is homogeneous of degree $e$, this contradicts the fact that $m$ must be a pivot index in $E_{e}$.
\medskip

We conclude that $f - g = 0$, hence $g$ appears as a row in $E_{e}$.
\end{proof}

\begin{lemma}
\label{lemma_Noether_position_generic}
Let $h_1, \ldots , h_n$ be homogeneous polynomials in $\Kp{n}$ such that
$1 \notin \ideal{h_1, \ldots , h_n}$.
The variables $x_1, \ldots , x_n$ are in simultaneous Noether position with
respect to $h_1, \ldots , h_n$ if and only if, for all integer
$i \in \{1, \ldots , n\}$, the sequence $\phi((h_1,\ldots , h_i), i)$ is regular.
\end{lemma}

\begin{proof}
By \cite[Proposition~6]{bardet2014complexity} and
\cite[Definition~3]{bardet2014complexity}, it suffices to prove that for all
$i \in \{1, \ldots , n\}$, the sequence
$(h_1, \ldots , h_i, x_{i+1}, \ldots , x_n)$ is regular.
Let $i \in \{1, \ldots , n\}$. By hypothesis, we know that
\begin{align*}
& \phi((h_1,\ldots , h_i), i) \text{ is regular} \\
\Longleftrightarrow\; &
(x_{i+1}, \ldots , x_n, h_1, \ldots , h_i) \text{ is regular} \\
\Longleftrightarrow\; &
(h_1, \ldots , h_i, x_{i+1}, \ldots , x_n) \text{ is regular},
\end{align*}
which concludes the proof.
\end{proof}

\end{document}